\documentclass[12pt,reqno,a4paper]{amsart}
\usepackage{amsmath,amssymb,dsfont,graphicx,cite,subfigure,enumerate}
\usepackage[colorlinks=false]{hyperref}
\def\beq{\begin{equation}}
\def\eeq{\end{equation}}
\def\bea{\begin{eqnarray}}
\def\eea{\end{eqnarray}}
\def\wx{\widetilde x}
\def\wy{\widetilde y}
\def\wz{\widetilde z}
\def\wm{\widetilde m}
\def\wip{\widetilde p}
\newtheorem{theorem}{Theorem}
\newtheorem{proposition}[theorem]{Proposition}
\newtheorem{prop}[theorem]{Proposition}

\newtheorem{definition}[theorem]{Definition}
\newtheorem{corollary}[theorem]{Corollary}

\makeatletter
\expandafter\let\expandafter
\reset@font\csname reset@font\endcsname
\def\subeqnarray{\arraycolsep1pt
   \def\@eqnnum\stepcounter##1{\stepcounter{subequation}
       {\reset@font\rm(\theequation\alph{subequation})}}
\jot5mm     \eqnarray}

\makeatother

\newcommand{\bbZ}{{\mathbb Z}}
\newcommand{\bbR}{{\mathbb R}}

\def\al{\alpha}
\def\ep{\epsilon}
\def\l{\lambda}

\def\ri{{\rm{i}}}
\def\til{\widetilde}

\def\su2{{\mathfrak {su}}(2)}
\def\e3{{\mathfrak {e}}(3)}

\def\half{\frac{1}{2}}
\def\nn{\nonumber}
\newbox\meibox
\def\placeunder#1#2#3#4{\setbox\meibox%
\vbox{\hbox{\hskip#4$\hphantom{#2}$}\hbox{$\hphantom{#1}$}}%
\vtop{\baselineskip=0pt\lineskiplimit=\baselineskip%
\lineskip=#3\hbox to \wd\meibox{\hfil\hskip#4$#2$\hfil}%
\hbox to \wd\meibox{\hfil$#1$\hfil}}}
\def\undertilde#1{\mathchoice{%
\placeunder{\vbox to 1.4pt{\hbox{$\displaystyle\widetilde{\,\,\,
}$}\vss}}{\displaystyle#1}{1.5pt}{1.5pt}}%
{\placeunder{\vbox to 1.4pt{\hbox{$\textstyle\widetilde{\,\,
}$}\vss}}{\textstyle#1}{1.5pt}{1.5pt}}%
{\placeunder{\vbox to 1.4pt{\hbox{$\scriptstyle\tilde{
}$}\vss}}{\scriptstyle#1}{1pt}{1pt}}%
{\placeunder{\vbox to 1.4pt{\hbox{$\scriptscriptstyle\tilde{
}$}\vss}}{\scriptscriptstyle#1}{1pt}{1pt}}%
}
\begin{document}
\title{On integrability of Hirota-Kimura type discretizations}

\author{Matteo Petrera \and Andreas Pfadler\and Yuri B. Suris }

\thanks{Institut f\"ur Mathematik, MA 7-2,
Technische Universit\"at Berlin, Str. des 17. Juni 136, 10623 Berlin, Germany.}
\thanks{E-mail: {\tt  petrera@math.tu-berlin.de, pfadler@math.tu-berlin.de, suris@math.tu-berlin.de}}

\date{\today}
\begin{abstract}
We give an overview of the integrability of the Hirota-Kimura discretization method applied to algebraically completely integrable (a.c.i.) systems with quadratic vector fields. Along with the description of the basic mechanism of integrability (Hirota-Kimura bases), we provide the reader with a fairly complete list of the currently available results for concrete a.c.i. systems.
\end{abstract}
\maketitle

\tableofcontents



\section{Introduction}
\label{Sect: intro}

The discretization method studied in this paper seems to be
introduced in the geometric integration literature by W. Kahan in
the unpublished notes \cite{K}. It is applicable to any system of
ordinary differential equations for $x:\bbR\to\bbR^n$ with a
quadratic vector field:
\begin{equation}\nn
\dot{x}=Q(x)+Bx+c,
\end{equation}
where each component of $Q:\bbR^n\to\bbR^n$ is a quadratic form,
while $B\in{\rm Mat}_{n\times n}$ and $c\in\bbR^n$. Kahan's
discretization reads as
\begin{equation}\label{eq: Kahan gen}
\frac{\widetilde{x}-x}{\epsilon}=Q(x,\widetilde{x})+
\frac{1}{2}B(x+\widetilde{x})+c,
\end{equation}
where
\[
Q(x,\widetilde{x})=\frac{1}{2}\left(Q(x+\widetilde{x})-Q(x)-
Q(\widetilde{x})\right)
\]
is the symmetric bilinear form corresponding to the quadratic form
$Q$. Here and below we use the following notational convention
which will allow us to omit a lot of indices: for a sequence
$x:\bbZ\to\bbR$ we write $x$ for $x_k$ and $\widetilde{x}$ for
$x_{k+1}$. Eq. (\ref{eq: Kahan gen}) is {\em linear} with respect
to $\widetilde x$ and therefore defines a {\em rational} map
$\widetilde{x}=f(x,\epsilon)$. Clearly, this map approximates the
time-$\epsilon$-shift along the solutions of the original
differential system, so that $x_k\approx x(k\epsilon)$. (Sometimes it will be more convenient to use  $2\epsilon$ for the time step, in order to avoid appearance of various powers of 2 in numerous
formulas.) Since eq. (\ref{eq: Kahan gen}) remains invariant under the interchange $x\leftrightarrow\widetilde{x}$ with the simultaneous sign inversion $\epsilon\mapsto-\epsilon$, one has the {\em reversibility} property
\begin{equation}\nn
f^{-1}(x,\epsilon)=f(x,-\epsilon).
\end{equation}
In particular, the map $f$ is {\em birational}.

W.~Kahan applied this discretization scheme to the famous
Lotka-Volterra system and showed that in this case it possesses a
very remarkable non-spiralling property. This example is briefly discussed in \cite{PPS}.
Some further applications of
this discretization have been explored in \cite{KHLI}.

The next, even more intriguing appearance of this discretization
was in the two papers by R. Hirota and K. Kimura who (being
apparently unaware of the work by Kahan) applied it to two famous
{\em integrable} system of classical mechanics, the Euler top and
the Lagrange top \cite{HK, KH}. For the purposes of the present
text, integrability of a dynamical system is synonymous with the
existence of a sufficient number of functionally independent
conserved quantities, or integrals of motion, that is, functions
constant along the orbits. We leave aside other aspects of the
multi-faceted notion of integrability, such as Hamiltonian ones or
explicit solutions. Surprisingly, the Kahan-Hirota-Kimura
discretization scheme produced in both the Euler and the Lagrange
cases of the rigid body motion {\em integrable} maps. Even more
surprisingly, the mechanism which assures integrability in these
two cases seems to be rather different from the majority of
examples known in the area of integrable discretizations, and,
more generally, integrable maps, cf. \cite{S}. The case of the
discrete time Euler top is relatively simple, and the proof of its
integrability given in \cite{HK} is rather straightforward and
easy to verify by hands. As it often happens, no explanation was
given in \cite{HK} about how this result has been {\em
discovered}. The ``derivation'' of integrals of motion for the
discrete time Lagrange top in \cite{KH} is rather cryptic and
almost uncomprehensible.

We use the term ``Hirota-Kimura (HK) type discretization'' for Kahan discretization in the context of
integrable systems. At the Oberwolfach meeting on Geometric Integration in 2006, T.~Ratiu proposed to apply the Hirota-Kimura discretization to the Clebsch case of the rigid body motion in an ideal fluid and to the Kovalevsky top. The claim on the integrability of the HK discretization of the Clebsch case was proven only several years later in \cite{PPS}, while the integrability of the HK discretization of the Kovalevsky top remains an open problem (although there are some indications in favor of its non-integrability). Anyway, the general question of integrability of the HK type discretizations turns out to be very intriguing and rather difficult. In the present overview, we will present a rather long list of examples of integrable HK discretizations. Actually, this list is so impressive that in \cite{PPS} we conjectured that HK discretizations of algebraically integrable systems always remain integrable. At present, we have some indications that this conjecture is wrong (although a rigorous proof of non-integrability for any of the apparently non-integrable cases remains elusive). Nevertheless, the sheer length of our list of examples clearly shows that there exist some general mechanisms that ensure integrability at least under certain additional assumptions. We think that to uncover general structures behind the integrability of HK discretizations is a big and important challenge for the modern theory of (algebraically) integrable systems.

The structure of our overview is as follows. In Section \ref{Sect: meas} we demonstrate some general sufficient conditions for a HK discretization to be measure preserving. These conditions are not related to integrability and do not cover all special cases considered in the main text. All our examples turn out to be measure preserving but for the majority of them we only can prove this property individually and do not know any general mechanisms. In Section \ref{Sect: HK mechanism} we present a formalization of the HK mechanism from \cite{KH}, which will hopefully unveil its main idea and contribute towards demystifying at least some of its aspects. We introduce a notion of a ``Hirota-Kimura basis'' (HK basis) for a given map $f$. Such a basis $\Phi$ is a set of simple (often monomial) functions, $\Phi=(\varphi_l,\ldots,\varphi_l)$, such that for every orbit $\{f^i(x)\}$ of the map $f$ there is a certain linear combination $c_1\varphi_1+\ldots+c_l\varphi_l$ of functions from $\Phi$ vanishing on this orbit. As explained in Sect. \ref{Sect: HK mechanism}, this is a {\em new} mathematical notion, not reducible to that of integrals of motion, although closely related to the latter. We lay a theoretical fundament for the search for HK bases for a given discrete time system, and discuss some practical recipes and tricks for doing this. Sections \ref{Sect: weierstrass}--\ref{Sect: dG} contain our list of examples of algebraically integrable systems with quadratic vector fields, for which the HK discretization preserves integrability. For all the examples we provide the reader with a rather complete set of currently available results. The proofs are omitted almost everywhere. One of the reasons is that our investigations are based mainly on computer experiments, which are used both for {\em discovery} of new results and for their rigorous {\em proof}. For a given system, a search for HK bases can be done with the help of {\em numerical experiments}. If the search has been successful and a certain set of functions $\Phi$
has been identified as a HK basis for a given map $f$, then numerical experiments can provide a very convincing evidence in favor of such a statement. A rigorous proof of such a statement turns out to be much more demanding. At present, we are not in possession of any theoretical proof strategies and are forced to
verify the corresponding statements by means of {\em symbolic computations} which in some cases turn out to be hardly feasible due to complexity issues. All these issues are intentionally avoided in our presentation here; an interested reader may consult \cite{PPS} for a detailed discussion of one concrete example (Clebsch system). Our main goal here is to document the available results on integrable HK discretizations and to attract  attention of specialists in integrable systems and in algebraic geometry to these beautiful and mysterious objects which are definitely worth further investigation.

\section{Some general results on invariant measures}
\label{Sect: meas}

In this section, we prove the existence of invariant measures for Kahan discretizations of two classes of dynamical systems with quadratic vector fields (not necessarily integrable).

The first class reads:
\begin{equation}\nn
    \dot{x}_i=\sum_{j=1}^N a_{ij}x_j^2+c_i,\qquad 1\le i\le N,
\end{equation}
with a  skew-symmetric matrix $A=(a_{ij})_{i,j=1}^N=-A^{\rm T}$. The Kahan's discretization reads:
\begin{equation}\label{eq: ddress type}
   \wx_i-x_i=\epsilon\sum_{j=1}^N a_{ij}x_j\wx_j+\epsilon c_i,\qquad 1\le i\le N.
\end{equation}
\begin{prop}\label{th: ddress type}
    The map $\wx=f(x,\epsilon)$ defined by eqs. (\ref{eq: ddress type}) has an invariant volume form:
\begin{equation} \nn
\det\frac{\partial \widetilde{x}}{\partial x}=
\frac{\phi(\widetilde{x},\epsilon)}{\phi(x,\epsilon)}\quad\Leftrightarrow\quad f^*\omega=\omega,\quad
\omega= \frac{dx_1\wedge\ldots\wedge dx_N}{\phi(x,\epsilon)},
\end{equation}
where $\phi(x,\epsilon)=\det(\mathds{1}-\epsilon AX)$, with $X={\rm diag}(x_1,\ldots, x_N)$, is an even polynomial in $\epsilon$.
\end{prop}
\begin{proof} Equations (\ref{eq: ddress type}) can be put as
\begin{equation}\nn
    \wx=A^{-1}(x,\epsilon)(x+\epsilon c),\qquad A(x,\epsilon)=\mathds{1}-\epsilon AX.
\end{equation}
As for any Kahan type discretization, there holds the formula
\begin{equation} \label{Jac Deltas}
\det\frac{\partial \widetilde{x}}{\partial x}=
\frac{\det A(\widetilde{x},-\epsilon)}{\det A(x,\epsilon)}.
\end{equation}
Indeed, differentiation of equations (\ref{eq: ddress type}) with respect to $x_j$ will give the $j$-th column of the matrix equation
\[
A(x,\epsilon)\frac{\partial \widetilde{x}}{\partial x}=A(\wx,-\epsilon).
\]
It remains to notice that $\det A(x,\epsilon)=\det A(x,-\epsilon)$. Indeed, due to the skew-symmetry of $A$, we have:
$
\det(\mathds{1}-\epsilon AX)=\det(\mathds{1}-\epsilon X^{\rm T}A^{\rm T})=\det(\mathds{1}+\epsilon XA)=\det(\mathds{1}+\epsilon AX).
$
\end{proof}

The second class consists of equations of the Lotka-Volterra type:
\begin{equation} \nn
    \dot{x}_i=x_i\left(b_i+\sum_{j=1}^N a_{ij}x_j\right),\qquad 1\le i\le N,
\end{equation}
with a  skew-symmetric matrix $A=(a_{ij})_{i,j=1}^N=-A^{\rm T}$. The Kahan's discretization (with the stepsize $2\epsilon$) reads:
\begin{equation}\label{eq: dLV type}
   \wx_i-x_i=\epsilon b_i(x_i+\wx_i)+\epsilon\sum_{j=1}^N a_{ij}(x_i\wx_j+\wx_ix_j),\qquad 1\le i\le N.
\end{equation}
\begin{prop}\label{th: dLV type}
    The map $\wx=f(x,\epsilon)$ defined by equations (\ref{eq: dLV type}) has an invariant volume form:
\begin{equation} \nn
\det\frac{\partial \widetilde{x}}{\partial x}=
\frac{\wx_1\wx_2\cdots\wx_N}{x_1x_2\cdots x_N}\quad\Leftrightarrow\quad f^*\omega=\omega,\quad
\omega= \frac{dx_1\wedge\ldots\wedge dx_N}{x_1x_2\cdots x_N}.
\end{equation}
\end{prop}
\begin{proof} Equations  (\ref{eq: dLV type}) are equivalent to
\begin{equation}\label{eq: dLV type alt}
   \frac{\wx_i}{1+\epsilon b_i+\epsilon\sum_{j=1}^Na_{ij}\wx_j}=
    \frac{x_i}{1-\epsilon b_i-\epsilon\sum_{j=1}^Na_{ij}x_j}=: y_i.
\end{equation}
We denote $d_i(x,\epsilon)=1-\epsilon b_i-\epsilon\sum_{j=1}^Na_{ij}x_j$. In the matrix form equation (\ref{eq: dLV type}) can be put as
\begin{equation}\nn
    \wx=A^{-1}(x,\epsilon)(\mathds{1}+\epsilon B)x,
\end{equation}
where the $i$-th diagonal entry of $A(x,\epsilon)$ equals $d_i(x,\epsilon)$, while the $(i,j)$-th off-diagonal entry equals $-\epsilon x_ia_{ij}$. In other words, $A(x,\epsilon)=D(\mathds{1}-\epsilon YA)$, where $D=D(x,\epsilon)={\rm diag}(d_1,\ldots,d_N)$ and $Y={\rm diag}(y_1,\ldots y_N)$.
Formula (\ref{Jac Deltas}) holds true also in the present case, and it implies:
\[
\det\frac{\partial \widetilde{x}}{\partial x}=\frac{\det D(\wx,-\epsilon)}{\det D(x,\epsilon)}\,
\frac{\det(\mathds{1}+\epsilon YA)}{\det(\mathds{1}-\epsilon YA)}.
\]
The second factor equals 1 due to the skew-symmetry of $A$, while the first factor equals
\[
\frac{d_1(\wx,-\epsilon)\cdots d_N(\wx,-\epsilon)}{d_1(x,\epsilon)\cdots d_N(x,\epsilon)}=
\frac{\wx_1\cdots \wx_N}{x_1\cdots x_N},
\]
by virtue of (\ref{eq: dLV type alt}).
\end{proof}

The statement of Theorem \ref{th: dLV type} for the Kahan discretization of the classical Lotka-Volterra model with two species, $\dot{x}=x(a-by)$,\, $\dot{y}=y(cx-d)$, was found in \cite{SZ} and used to explain the non-spiralling behavior of the numerical orbits in this case.

\section{Hirota-Kimura bases and integrals}
\label{Sect: HK mechanism}

In this section a general formulation of a remarkable mechanism
will be given, which seems to be responsible for the integrability
of the Hirota-Kimura type (or Kahan type) discretizations of
algebraically completely integrable systems. This mechanism is so
far not well understood, in fact at the moment we do not know what
mathematical structures make it actually work.

Throughout this section $f:\bbR^n\to\bbR^n$ is a birational map,
while $h_i,\varphi_i:\bbR^n\to\bbR$ stand for rational, usually
polynomial functions on the phase space. We start with recalling a
well known definition.

\begin{definition}\label{Def: integral}
A function $h:\bbR^n\to\bbR$ is called an {\bf integral}, or a
{\bf conserved quantity}, of the map $f$, if for every
$x\in\bbR^n$ there holds
\begin{equation}
\nonumber
h(f(x))=h(x),
\end{equation}
so that $h(f^i(x))=h(x)$ for all $i\in\bbZ$.
\end{definition}

Thus, each orbit of the map $f$ lies on a certain level set of its
integral $h$. As a consequence, if one knows $d$ functionally
independent integrals $h_1,\ldots,h_d$ of $f$, one can claim that
each orbit of $f$ is confined to an $(n-d)$-dimensional invariant
set, which is a common level set of the functions
$h_1,\ldots,h_d$.

\begin{definition}\label{Def: Hirota mech}
A set of functions $\Phi=(\varphi_1,\ldots,\varphi_l)$, linearly
independent over $\bbR$, is called a {\bf Hirota-Kimura basis
(HK basis)}, if for every $x\in\bbR^n$ there exists a vector
$c=(c_1,\ldots,c_l)\neq 0$ such that
\begin{equation}\label{eq: fundamental}
c_1\varphi_1(f^i(x))+\ldots+c_l\varphi_l(f^i(x))=0
\end{equation}
holds true for all $i\in\bbZ$. For a given $x\in\bbR^n$, the vector space consisting of all $c\in\bbR^l$ with
this property will be denoted by $K_\Phi(x)$ and called the
null-space of the basis $\Phi$ (at the point $x$).
\end{definition}

Thus, for a HK basis $\Phi$ and for $c\in K_\Phi(x)$ the function
$h=c_1\varphi_1 + ... + c_l\varphi_l$ vanishes along the $f$-orbit
of $x$. Let us stress that we cannot claim that $h=c_1\varphi_1 +
..... + c_l\varphi_l$ is an integral of motion, since vectors $c\in
K_\Phi(x)$ do not have to belong to $K_\Phi(y)$ for initial points
$y$ not lying on the orbit of $x$. However, for any $x$ the orbit
$\{f^i(x)\}$ is confined to the common zero level set of $d$
functions
\begin{equation}
\nonumber
h_j=c_1^{(j)}\varphi_1+\ldots+c_l^{(j)}\varphi_l=0,\quad
j=1,\ldots,d,
\end{equation}
where the vectors
$c^{(j)}=\big(c_1^{(j)},\ldots,c_l^{(j)}\big)\in\bbR^l$ form a
basis of $K_\Phi(x)$. We will say that the HK basis $\Phi$ is {\bf regular}, if the differentials
$dh_1,\ldots,dh_d$ are lineraly independent along the the common zero level set of the functions
$h_1,\ldots,h_d$. Thus, knowledge of a regular HK basis with a $d$-dimensional
null-space leads to a similar conclusion as
knowledge of $d$ independent integrals of $f$, namely to the
conclusion that the orbits lie on $(n-d)$-dimensional invariant
sets. Note, however, that a HK basis gives no immediate
information on how these invariant sets foliate the phase space
$\bbR^n$, since the vectors $c^{(j)}$, and therefore the functions
$h_j$, change from one initial point $x$ to another.

Although the notions of integrals and of HK bases cannot be
immediately translated into one another, they turn out to be
closely related.

The simplest situation for a HK basis corresponds to $l=2$, $\dim
K_\Phi(x)=d=1$. In this case we immediately see that
$h=\varphi_1/\varphi_2$ is an integral of motion of the map $f$.
Conversely, for any rational integral of motion
$h=\varphi_1/\varphi_2$ its numerator and denominator $\varphi_1$,
$\varphi_2$ satisfy
$$
c_1\varphi_1(f^i(x))+c_2\varphi_2(f^i(x))=0,\quad i\in\bbZ,
$$
with $c_1=1$, $c_2=-h(x)$, and thus build a HK basis with $l=2$.
Thus, the notion of a HK basis {\em generalizes} (for $l\ge 3$)
the notion of integrals of motion.

On the other hand, knowing a HK basis $\Phi$ with $\dim
K_\Phi(x)=d\ge 1$ allows one to find integrals of motion for the
map $f$. Indeed, from Definition \ref{Def: Hirota mech} there
follows immediately:
\begin{proposition}\label{Th: K_integral}
If $\Phi$ is a HK basis for a map $f$, then
\begin{equation}
\nonumber
K_\Phi(f(x))=K_\Phi(x).
\end{equation}
\end{proposition}
Thus, the $d$-dimensional null-space $K_\Phi(x)\in Gr(d,l)$,
regarded as a function of the initial point $x\in\bbR^n$, is
constant along trajectories of the map $f$, i.e., it is a
$Gr(d,l)$-valued integral. Its Pl\"ucker coordinates are then scalar integrals:
\begin{corollary}\label{Th: mechanism} Let $\Phi$ be a
HK basis for $f$ with $\dim K_\Phi(x)=d$ for all
$x\in\mathbb{R}^n$. Take a basis of $K_\Phi(x)$ consisting of $d$
vectors $c^{(i)}\in\bbR^l$ and put them into the columns of a
$l\times d$ matrix $C(x)$. For any $d$-index
$\alpha=(\alpha_1,\ldots,\alpha_d)\subset\{1,2,\ldots,n\}$ let
$C_\alpha=C_{\alpha_1\ldots\alpha_d}$ denote the $d\times d$ minor
of the matrix $C$ built from the rows $\alpha_1,\ldots,\alpha_d$.
Then for any two $d$-indices $\alpha,\beta$ the function
$C_\alpha/C_\beta$ is an integral of $f$.
\end{corollary}

Especially simple is the situation when the null-space of a
HK basis has dimension $d=1$.
\begin{corollary}\label{Th: mechanism d=1} Let $\Phi$ be a
HK basis for $f$ with $\dim K_\Phi(x)=1$ for all
$x\in\mathbb{R}^n$. Let $K_\Phi(x)=[c_1(x):\ldots :c_l(x)]\in
\bbR\mathbb{P}^{l-1}$. Then the functions $c_j/c_k$ are integrals
of motion for $f$.
\end{corollary}
An interesting (and difficult) question is about the number of
functionally independent integrals obtained from a given HK basis
according to Corollaries \ref{Th: mechanism} and \ref{Th:
mechanism d=1}. It is possible for a
HK basis with a one-dimensional null-space to produce more than
one independent integral. The first examples of this mechanism (with $d=1$) were found in
\cite{KH} and (somewhat implicitly) in \cite{HK}.

We note, however, that HK bases appeared in a disguised form in the continuous time theory long ago. We mention here two relevant examples.
\begin{itemize}
\item Classically, integration of a given system of ODEs in terms of elliptic functions started with the derivation of an equation of the type $\dot y^2=P_4(y)$, where $y$ is one of the components of the solution, and $P_4(y)$ is a polynomial of degree 4 with constant coefficients (depending on parameters of the system and on its integrals of motion), see examples in Sections \ref{Sect: dL}, \ref{Sect: dK}. This can be interpreted as the claim about $\Phi=(\dot{y}^2,\, y^4,\, y^3,\, y^2,\, y,\, 1)$ being a HK basis with a one-dimensional null-space.
\item According to \cite[Sect.\ 7.6.6]{AvM}, for any algebraically integrable system, one can choose projective coordinates $y_0,y_1,\ldots, y_n$ so that {\em quadratic Wronskian equations} are satisfied:
    \[
    \dot{y}_iy_j-y_i\dot{y}_j=\sum_{k,l=0}^n\alpha_{ij}^{kl}y_ky_l,
    \]
    with coefficients $\alpha_{ij}^{kl}$ depending on integrals of motion of the original system. Again, this admits an immediate interpretation in terms of HK bases consisting of the Wronskians and the quadratic monomials of the coordinate functions: $\Phi_{ij}=\big(\dot{y}_iy_j-y_i\dot{y}_j,\, \{y_ky_l\}_{k,l=0}^n\big)$.
\end{itemize}
Thus, these HK bases consist not only of simple monomials, but include also more complicated functions composed of the vector field of the system at hand. We will encounter discrete counterparts of these HK bases, as well.


At present, we cannot give any theoretical sufficient conditions
for existence of a HK basis $\Phi$ for a given map $f$,
and the only way to find such a basis remains the experimental
one. Definition \ref{Def: Hirota mech} requires to verify
condition (\ref{eq: fundamental}) for all $i\in\bbZ$, which is, of
course, impractical. However, it is enough to check this
condition for a finite number of iterates $f^i$.

Typically (that is, for general maps $f$ and general monomial sets $\Phi$), the dimension of the vector space of solutions of the homogeneous system of $s$ linear equations (\ref{eq: fundamental}) with $i=i_0,i_0+1,\ldots,i_0+s-1$ decays with the growing $s$, from $l-1$ for $s=1$ down to $0$ for $s=l$. If, however, $\Phi$ is a HK basis for $f$ with a $d$-dimensional null-space, then this dimension fails to drop starting with $s=l-d+1$. Thus, the dimension of the solution space of the system (\ref{eq: fundamental}) with $i=i_0,i_0+1,\ldots,i_0+s-1$ is equal to $l-s$ for all $1 \leq s \leq l-d$, and remains equal to $d$ for $s=l-d+1$. It is easy to see that this situation can be also characterized as follows: the $d$-dimensional spaces of solutions of the system (\ref{eq: fundamental}) with $i=i_0,i_0+1,\ldots,i_0+l-d-1$ and of the system (\ref{eq: fundamental}) with $i=i_0+1,i_0+2,\ldots,i_0+l-d$ coincide with each other and with $K_\Phi(x)$. The most important particular case of this characterization corresponds to $d=1$ and reads as follows:

\begin{prop}\label{Lem: dim K} A set $\Phi=(\varphi_1,\ldots,\varphi_l)$ is a HK basis for a map $f$ with $\dim K_\Phi(x)=1$ if and only if the unique solution $[c_1:\ldots :c_l]\in\bbR\mathbb P^{l-1}$ of the system (\ref{eq: fundamental}) with $i=i_0,i_0+1,\ldots,i_0+l-2$ coincides with the unique solution of the analogous system with $i=i_0+1,i_0+2,\ldots,i_0+l-1$, in other words, if this unique solution  $[c_1:\ldots :c_l]$ is an integral of motion:
\[
 [c_1(x):\ldots :c_l(x)]=[c_1(f(x)):\ldots :c_l(f(x))].
\]
\end{prop}

At this point it should be mentioned that a numerical testing of the above criterium usually represents no problems, but the corresponding symbolic computations might be extremely complicated, due to the complexity of the iterates $f^i(x)$. While the expression for $f(x)$ is typically of a moderate size, already the second iterate $f^2(x)$ becomes typically prohibitively big. See the detailed discussion of the complexity issue for the case of the Clebsch system in \cite{PPS}. In such a situation it becomes crucial to reduce the number of iterates involved in (\ref{eq: fundamental}) as far as possible. Several tricks which can be used for this aim are also described in \cite{PPS}. Here is one of them.

\begin{prop}
Consider the non-homogeneous system of $l-1$ equations
\begin{equation}\label{eq: smaller system}
c_1\varphi_1(f^i(x)) +\ldots +c_{l-1} \varphi_{l-1}(f^i(x))=
\varphi_l(f^i(x)),\quad i=i_0,i_0+1,\ldots,i_0+l-2.\quad
\end{equation}
Suppose that the index range $i\in[i_0,i_0+l-2]$ in eq. (\ref{eq: smaller system}) contains 0 but is non-symmetric. If the solution of this system $\big(c_1(x,\epsilon),\ldots,c_{l-1}(x,\epsilon)\big)$ is unique and is even with respect to $\epsilon$, then all $c_k(x,\epsilon)$ are conserved quantities of the map $f$, and $\Phi=(\varphi_1,\ldots,\varphi_l)$ is a HK basis for $f$ with $\dim K_\Phi(x)=1$.
\end{prop}
\begin{proof} Considering the non-homogeneous system (\ref{eq: smaller system}) instead of the homogeneous one (\ref{eq: fundamental}) corresponds just to fixing an affine representative of the projective solution $[c_1:\ldots:c_l]$ by $c_l=-1$. The reversibility of the map $f^{-1}(x,\epsilon)=f(x,-\epsilon)$ yields that equations of the system (\ref{eq: smaller system}) are satisfied not only for $i\in[i_0,i_0+l-2]$ but for $i\in[-(i_0+l-2),-i_0]$, as well. Since, by condition, the intervals $[i_0,i_0+l-2]$ and $[-(i_0+l-2),-i_0]$ overlap but do not coincide, their union is an interval containing $\ge l$ integers.
\end{proof}

 Of course, it would be highly desirable to
find some structures, like Lax representation, bi-Hamiltonian
structure, etc., which would allow one to check the conservation
of integrals in a more clever way, but up to now no such
structures have been found for any of the HK type
discretizations.

\section{Weierstrass differential equation}
\label{Sect: weierstrass}

Consider the second-order differential equation
\begin{equation} \label{weier}
\ddot x=6 x^2-\alpha.
\end{equation}
Its general solution is given by the Weierstrass elliptic function $\wp(t)=\wp(t,g_2,g_3)$ with the invariants $g_2=2\alpha$, $g_3$  arbitrary, and by its time shifts. Actually, the parameter $g_3$ can be interpreted as the value of an integral of motion (conserved quantity) of system (\ref{weier}):
\begin{equation} \nonumber
\dot x^2-4x^3+2\alpha x=-g_3.
\end{equation}
Being re-written as a system of first-order equations with a quadratic vector field,
\begin{equation} \label{weiersystem}
\left\{ \begin{array}{l}
\dot x = y, \vspace{.1truecm}\\
\dot y = 6x^2-\alpha,
\end{array} \right.
\end{equation}
equation (\ref{weier}) becomes suitable for an application of the Kahan-Hirota-Kimura discretization:
\begin{equation} \label{weiersystemdiscrete}
 \left\{\begin{array}{l}
 \widetilde x-x=\dfrac{\epsilon}{2} \left(\widetilde y+y\right),  \vspace{.2truecm}\\
 \widetilde y-y=\epsilon\left(6x\widetilde x-\alpha\right).\end{array}\right.
\end{equation}
Eqs. (\ref{weiersystemdiscrete}), put as a linear system for $(\widetilde x, \widetilde y)$, read:
$$
\begin{pmatrix} 1  & -\epsilon/2\\
  -6\epsilon x & 1\end{pmatrix}
\begin{pmatrix}\widetilde{x} \\ \widetilde{y}\end{pmatrix}
=\begin{pmatrix} x +\epsilon y/2  \\ y - \epsilon\alpha  \end{pmatrix}.
$$
This can be immediately solved, thus yielding an explicit birational map $(\widetilde x,\widetilde y)=f(x,y,\epsilon)$:
\begin{equation} \label{weier explicit}
\left\{\begin{array}{l}
\widetilde x={\displaystyle{\frac{x+\epsilon y -\epsilon^2\alpha/2}{1-3\epsilon^2 x}}},\vspace{.2truecm}\\
\widetilde y={\displaystyle{\frac{y+\epsilon(6x^2-\alpha)+3\epsilon^2 xy}{1-3\epsilon^2 x}}}.
\end{array} \right.
\end{equation}
This map turns out to be integrable: it possesses an invariant two-form
\begin{equation} \label{weier 2-form}
\omega=\frac{dx \wedge dy}{1-3\ep^2 x},
\end{equation}
and an integral of motion (conserved quantity):
\begin{equation} \label{weier int xy}
I(x,y,\epsilon)=\frac{ y^2-4x^3+2\alpha x+\ep^2 x(y^2-2\alpha x)-\ep^4 \alpha^2 x}{1-3\ep^2 x}.
\end{equation}
Both these objects are $O(\epsilon^2)$-perturbations of the corresponding objects for the continuous time system (\ref{weiersystem}). The statement about the invariant 2-form (\ref{weier 2-form}) is not difficult to prove. The following argument exemplifies considerations which hold for an arbitrary Kahan discretization (\ref{eq: Kahan gen}). Differentiating eqs. (\ref{weiersystemdiscrete}) with respect to $x$ and to $y$, we obtain the columns of the matrix equation
\[
\begin{pmatrix} 1 & -\epsilon/2 \\ -6\epsilon x & 1\end{pmatrix}\frac{\partial (\widetilde x,\widetilde y)}
{\partial (x,y)}=\begin{pmatrix} 1 & \epsilon/2 \\ 6\epsilon\widetilde x & 1\end{pmatrix},
\]
whence
\[
\det\frac{\partial (\widetilde x,\widetilde y)}{\partial (x,y)}=\frac{1-3\epsilon^2\widetilde x}{1-3\epsilon^2 x}.
\]
This is equivalent to the preservation of (\ref{weier 2-form}). The statement about the conserved quantity is most simply verified with any computer system for symbolic manipulations.

System (\ref{weiersystemdiscrete}) is known in the literature on integrable maps, although in a somewhat different form. Indeed, it is equivalent to the second order difference equation
\begin{equation} \label{weierdiscrete} \nonumber
 \widetilde x-2x+\undertilde{x}=\epsilon^2 \big(3x(\widetilde x+\undertilde{x})-\alpha\big)\quad
 \Leftrightarrow\quad
 \widetilde x-2x+\undertilde{x}=\frac{\epsilon^2(6x^2-\alpha)}{1-3\epsilon^2x}.
\end{equation}
This equation belongs to class of integrable {\em QRT systems} \cite{QRT, Suris0}; in order to see this, one should re-write it as
\[
 \widetilde x-2x+\undertilde{x}=\frac{\epsilon^2(6x^2-\alpha)(1+\epsilon^2x)}{1-2\epsilon^2x-3\epsilon^4x^2}.
\]
This difference equation generates a map $(x,\undertilde x)\mapsto (\widetilde x,x)$ which is symplectic, that is, preserves the two-form $\omega=dx\wedge d\widetilde x$, and possesses a biquadratic integral of motion
\[
I(x,\widetilde x, \epsilon)=(\widetilde x-x)^2-2\epsilon^2x\widetilde x(x+\widetilde x)+\epsilon^2\alpha(x+\widetilde x)-\epsilon^4(3x^2\widetilde x^2-\alpha x\widetilde x).
\]
Under the change of variables $(x,\widetilde x)\mapsto (x,y)$ given by the first equation in (\ref{weier explicit}), these integrability attributes turn into the two-form (\ref{weier 2-form}) and the conserved quantity (\ref{weier int xy}) (up to an additive constant).

We note that a more usual QRT discretization of the Weierstrass second order equation (\ref{weier}) would be
\begin{equation} \label{weier QRT}
 \widetilde x-2x+\undertilde{x}=\frac{\epsilon^2(6x^2-\alpha)}{1-2\epsilon^2x},
\end{equation}
with a simpler conserved quantity
\[
J(x,\widetilde x,\epsilon)=(\widetilde x-x)^2-2\epsilon^2x\widetilde x(x+\widetilde x)+\epsilon^2\alpha(x+\widetilde x).
\]
Eq. (\ref{weier QRT}) is equivalent to
\[
 \widetilde x-2x+\undertilde{x}=\epsilon^2 \big(2x(\widetilde x+\undertilde{x})+2x^2-\alpha\big),
\]
which is not obtained by the Kahan-Hirota-Kimura method.

\section{Some two-dimensional integrable systems}
\subsection{The three-dimensional Suslov system}
\label{Sect: dsuslov}

The three-dimensional nonholonomic Suslov problem \cite{Suslov} is defined by the
following system of differential equations
\beq
\dot m = m \times \omega + \lambda a, \qquad \langle a, \omega \rangle =0,
\label{suslov}
\eeq
where $\omega= (\omega_1 , \omega_2, \omega_3)^{\rm T}$ is the angular velocity,
$m = (m_1,m_2,m_3)^{\rm T}= I \omega$ is the angular momentum, $I$ is the
inertia operator, $a$ is a unit vector fixed in the body, and
$\lambda$ is the Lagrange multiplier.
In a basis where $a=(0,0,1)^{\rm T}$ and
$$
I = \left( \begin {array}{ccc}
I_1 & 0 & I_{13} \\
0 & I_2 & I_{23} \\
I_{13} & I_{23} & I_3
\end {array} \right),
$$
the constraint $\langle a, \omega \rangle =0$ reduces to $\omega_3 = 0$, and
equations of motion (\ref{suslov}) read
\beq \label{susi}
\left\{ \begin{array}{l}
I_1 \dot \omega_1 =  - (I_{13}\omega_1+I_{23}\omega_2)\omega_2,
\vspace{.1truecm} \\
I_2 \dot \omega_2 =  (I_{13}\omega_1+ I_{23}\omega_2)\omega_1,
\vspace{.1truecm} \\
I_{13}\dot\omega_1+I_{23}\dot\omega_2=(I_1-I_2)\omega_1\omega_2 + \lambda.
 \end{array} \right.
\eeq

The first two equations in (\ref{susi}) form a closed system for $\omega_1$ and
$\omega_2$. It possesses a conserved quantity
$H = I_1 \omega_1^2 + I_2 \omega_2^2$.
After the solution of this system is found (Suslov gave it in terms of trigonometric
and exponential functions), one finds the Lagrange multiplier $\lambda$
from the third equation in (\ref{susi}).

To put the first two equations in (\ref{susi}) into a more convenient form, one can introduce the coordinates
$x = I_{13}\omega_1  + I_{23}\omega_2 $, $y =I_{23}I_1\omega_1-I_{13}I_2\omega_2$, and arrives at
\beq \label{susi2}
\left\{ \begin{array}{l}
\dot x =  \alpha x y , \vspace{.1truecm} \\
\dot y = - x^2 ,
 \end{array} \right.
\eeq
where $\alpha= 1/I_1 I_2$. This system admits a conserved quantity
$$
H = x^2 + \alpha y^2.
$$
\begin{prop}{\rm \cite{DraGa}}
HK discretization of system (\ref{susi2}),
$$
\left\{ \begin{array}{l}
\widetilde x - x = \epsilon \alpha ( \widetilde x y + x \til y ),
\vspace{.1truecm}\\
\widetilde y - y = - 2 \epsilon \widetilde x x ,
\end{array} \right.
$$
possesses an invariant two-form and an integral of motion, given by
$$
\omega = \frac{dx \wedge dy}{x (x^2 + \alpha y^2)},\qquad
H(\epsilon) = \frac{ x^2 + \alpha y^2}{ 1 + \epsilon^2 \alpha x^2 }.
$$
\end{prop}
Actually, the invariant two-form was not given in \cite{DraGa}; rather, this paper contains an explicit solution of the discrete time Suslov system and its qualitative analysis.

\subsection{Reduced Nahm equations}
\label{Sect: nahm}

In \cite{HMM} Nahm equations associated with symmetric monopoles are considered. Assuming rotational symmetry groups of regular polytops leads to solutions of Nahm equations in terms of elliptic functions. Reduced equations corresponding to tetrahedrally symmetric monopoles of charge $3$, to octahedrally symmetric monopoles of charge $4$, and to icosahedrally symmetric monopoles of charge $6$ are two-dimensional algebraically integrable systems with quadratic vector fields. More concretely, the reductions of Nahm equations derived in \cite{HMM} read:

(i) tetrahedral symmetry:
\begin{equation} \label{k3}
\left\{ \begin{array}{l}
\dot x = x^2 - y^2, \vspace{.1truecm}\\
\dot y = - 2xy,
\end{array} \right.
\end{equation}
with an integral of motion $H= y (3 x^2 - y^2);$

(ii) octahedral symmetry:
\begin{equation} \label{k4}
\left\{ \begin{array}{l}
\dot x = {\displaystyle{2 x^2 - 12y^2}}, \vspace{.1truecm}\\
\dot y = - 6 xy - 4 y^2,
\end{array} \right.
\end{equation}
with an integral of motion $H= y(2x+3y)(x-y)^2;$

(iii) icosahedral symmetry:
\begin{equation} \label{k6}
\left\{ \begin{array}{l}
\dot x = {\displaystyle{2 x^2 - y^2}}, \vspace{.1truecm}\\
\dot y = - 10 xy+ y^2,
\end{array} \right.
\end{equation}
with an integral of motion $H= y(3x -y)^2(4x +y)^3.$

HK discretizations of systems (\ref{k3}-\ref{k6}) turn out to be algebraically integrable.

\begin{prop}\quad

{\rm(i)} HK discretization of system (\ref{k3}),
\begin{equation}\nn
\left\{ \begin{array}{l}
\widetilde x - x = \epsilon ( \widetilde x x - \widetilde y y), \vspace{.1truecm}\\
\widetilde y - y = - \epsilon (\widetilde x y + x \widetilde y),
\end{array} \right.
\end{equation}
possesses an invariant two-form and an integral of motion, given by
$$
\omega = \frac{dx \wedge dy}{y (3 x^2 - y^2)},\quad
H(\epsilon) = \frac{y (3 x^2 - y^2)}{ 1 - \epsilon^2 (x^2 + y^2)};
$$

\smallskip

{\rm(ii)}  HK discretization of system (\ref{k4}),
\begin{equation} \nn
\left\{ \begin{array}{l}
\widetilde x - x = \epsilon \left( {\displaystyle{2 \widetilde x x- 12 \widetilde y y}} \right), \vspace{.1truecm}\\
\widetilde y - y = - \epsilon (3\widetilde x y + 3x \widetilde y+ 4  \widetilde y y ),
\end{array} \right.
\eeq
possesses an invariant two-form and an integral of motion given by
$$
\omega = \frac{dx \wedge dy}{y(2x+3y)(x-y)}, 
$$
$$
H(\epsilon) = \frac{y(2x+3y)(x-y)^2}{1-10\epsilon^2 (x^2 +4 y^2) + \epsilon^4 (9 x^4 + 272 x^3 y - 352 x y^3+ 696 y^4)};
$$

\smallskip

{\rm(iii)}   HK discretization of system (\ref{k6}),
\begin{equation} \nn
\left\{ \begin{array}{l}
\widetilde x - x = \epsilon \left( {\displaystyle{2 \widetilde x x- \widetilde y y}} \right), \vspace{.1truecm}\\
\widetilde y - y = \epsilon (-5 \widetilde x y - 5 x \widetilde y+ \widetilde y y ),
\end{array} \right.
\eeq
possesses an invariant two-form and an integral of motion given by
$$
\omega= \frac{dx \wedge dy}{y(3x -y)(4x +y)}, \quad
H(\epsilon) = \frac{y(3x -y)^2(4x +y)^3}{ 1 +  \epsilon^2 c_2 +  \epsilon^4  c_4 + \epsilon^6  c_6},
$$
with
\bea
&&c_2 =-  7 (5 x^2 - y^2), \nonumber \\
&&c_4 = 7 (37 x^4 + 22 x^2 y^2 - 2 x y^3 + 2 y^4), \nonumber \\
&&c_6 = -225 x^6 + 3840 x^5 y +80 x y^5 - 514 x^3 y^3 - 19 x^4 y^2 -206 x^2 y^4. \nonumber
\eea

\end{prop}

\section{Euler top}
\label{Sect: dET}

The differential equations of motion of the Euler top read
\begin{equation}\label{eq: ET x}
\left\{ \begin{array}{l}
\dot{x}_1=\alpha_1 x_2 x_3, \vspace{.1truecm} \\
\dot{x}_2=\alpha_2 x_3 x_1, \vspace{.1truecm} \\
\dot{x}_3=\alpha_3 x_1x_2,
\end{array} \right.
\end{equation}
with real parameters $\alpha_i$. This is one
of the most famous integrable systems of the classical mechanics,
with a big literature devoted to it. It can be explicitly integrated in terms of elliptic
functions, and admits two functionally independent integrals of
motion. Actually, a quadratic function
$H(x)=\gamma_1x_1^2+\gamma_2x_2^2+\gamma_3x_3^2$ is an integral
for eqs. (\ref{eq: ET x}) as soon as $\gamma_1\alpha_1+\gamma_2\alpha_2+\gamma_2\alpha_2=0$.
In particular, the following three functions are integrals of
motion:
\begin{equation}
\nonumber
H_1=\alpha_2x_3^2-\alpha_3x_2^2,\qquad
H_2=\alpha_3x_1^2-\alpha_1x_3^2,\qquad
H_3=\alpha_1x_2^2-\alpha_2x_1^2.
\end{equation}
Clearly, only two of them are functionally independent because of
$\alpha_1H_1+\alpha_2H_2+\alpha_3H_3=0$. These integrals appear also on the right-hand sides of the quadratic (in this case even linear) expressions for the Wronskians of the coordinates $x_j$:
\begin{equation}
\label{eq: ET W}
\left\{ \begin{array}{l}
\dot{x}_2x_3-x_2\dot{x}_3=H_1x_1,\vspace{.1truecm}\\
\dot{x}_3x_1-x_3\dot{x}_1=H_2x_2,\vspace{.1truecm}\\
\dot{x}_1x_2-x_1\dot{x}_2=H_3x_3.
\end{array} \right.
\end{equation}
Moreover, one easily sees that the coordinates $x_j$ satisfy the following differential equations with the coefficients depending on the integrals of motion:
\begin{equation}
 \nonumber
\left\{ \begin{array}{l}
\dot{x}_1^2 = (H_3+\alpha_2x_1^2)(\alpha_3x_1^2-H_2),\vspace{.1truecm}\\
\dot{x}_2^2= (H_1+\alpha_3x_2^2)(\alpha_1x_2^2-H_3),\vspace{.1truecm}\\
\dot{x}_3^2 = (H_2+\alpha_1x_3^2)(\alpha_2x_3^2-H_1).
\end{array} \right.
\end{equation}
The fact that the polynomials on the right-hand sides of these equations are of degree four implies that the solutions are given by elliptic functions.

The HK discretization of the Euler top \cite{HK} is:
\begin{equation}\label{eq: dET x}
\renewcommand{\arraystretch}{1.3}
\left\{\begin{array}{l}
\widetilde{x}_1-x_1=\epsilon\alpha_1(\widetilde{x}_2x_3+x_2\widetilde{x}_3),\\
\widetilde{x}_2-x_2=\epsilon\alpha_2(\widetilde{x}_3x_1+x_3\widetilde{x}_1),\\
\widetilde{x}_3-x_3=\epsilon\alpha_3(\widetilde{x}_1x_2+x_1\widetilde{x}_2).
\end{array}\right.
\end{equation}
(In this form it corresponds to the stepsize $2\epsilon$ rather than $\epsilon$.) The map $f:x\mapsto\widetilde{x}\ $ obtained by solving
(\ref{eq: dET x}) for $\widetilde{x}$ is given by:
\begin{equation}
\label{eq: dET map}
\widetilde{x} =f(x,\epsilon)=A^{-1}(x,\epsilon)x, \qquad
A(x,\epsilon)=
\begin{pmatrix}
1 & -\epsilon\alpha_1 x_3 & -\epsilon\alpha_1 x_2 \\
-\epsilon\alpha_2 x_3 & 1 & -\epsilon\alpha_2 x_1 \\
-\epsilon\alpha_3 x_2 & -\epsilon\alpha_3 x_1 & 1
\end{pmatrix} .
\end{equation}
It might be instructive to have a look at the explicit formulas for this map:
\begin{equation}
\label{eq: dET expl}
\left\{ \begin{array}{l}
\widetilde{x}_1=\dfrac{x_1+2\epsilon\alpha_1x_2x_3+
\epsilon^2x_1(-\alpha_2\alpha_3x_1^2+\alpha_3\alpha_1x_2^2+\alpha_1\alpha_2x_3^2)}
{\Delta(x,\epsilon)}\,,\\ \\
\widetilde{x}_2 = \dfrac{x_2+2\epsilon\alpha_2x_3x_1+
\epsilon^2x_2(\alpha_2\alpha_3x_1^2-\alpha_3\alpha_1x_2^2+\alpha_1\alpha_2x_3^2)}
{\Delta(x,\epsilon)}\,,\\ \\
\widetilde{x}_3 = \dfrac{x_3+2\epsilon\alpha_3x_1x_2+
\epsilon^2x_3(\alpha_2\alpha_3x_1^2+\alpha_3\alpha_1x_2^2-\alpha_1\alpha_2x_3^2)}
{\Delta(x,\epsilon)}\,,
\end{array}\right.
\end{equation}
where
\begin{equation}
\label{eq: dET Delta}
\Delta(x,\epsilon)=\det A(x,\epsilon)=
1-\epsilon^2(\alpha_2\alpha_3x_1^2+\alpha_3\alpha_1x_2^2+\alpha_1\alpha_2x_3^2)
-2\epsilon^3\alpha_1\alpha_2\alpha_3x_1x_2x_3.
\end{equation}
We will use the abbreviation dET for this map. As always the case for a HK discretization, dET is birational, with the reversibility property expressed as $f^{-1}(x,\epsilon)=f(x,-\epsilon)$.

\begin{proposition}\label{th: dE F}{\rm\cite{HK,PS}}
The quantities
\begin{equation}
\nonumber
F_1 = \frac{1-\epsilon^2\alpha_3\alpha_1x_2^2}{1-\epsilon^2\alpha_1\alpha_2x_3^2},\qquad
F_2 = \frac{1-\epsilon^2\alpha_1\alpha_2x_3^2}{1-\epsilon^2\alpha_2\alpha_3x_1^2},\qquad
F_3 = \frac{1-\epsilon^2\alpha_2\alpha_3x_1^2}{1-\epsilon^2\alpha_3\alpha_1x_2^2},
\end{equation}
are conserved quantities of dET. Of course, there are only two independent integrals since $F_1F_2F_3=1$.
\end{proposition}
The relation between $F_i$ and the integrals $H_i$ of the continuous time Euler top is straightforward: $F_i=1+\epsilon^2\alpha_iH_i+O(\epsilon^4)$. As a corollary of Proposition \ref{th: dE F}, we find that, for any conserved quantity $H$ of the Euler top which is a linear combination of the integrals $H_1,H_2,H_3$, the three functions $H/(1-\epsilon^2\alpha_j\alpha_kx_i^2)$ are conserved quantities of dET. Hereafter $(i,j,k)$ are cyclic permutations of $(1,2,3)$. In particular, the functions
\begin{equation}\label{eq: dET ints}
    H_i(\epsilon)=\frac{\alpha_jx_k^2-\alpha_kx_j^2}{1-\epsilon^2\alpha_j\alpha_kx_i^2}
\end{equation}
are conserved quantities of dET. Again, only two of them are independent, since $$\alpha_1 H_1 (\epsilon)+ \alpha_2 H_2 (\epsilon) + \alpha_3 H_3  (\epsilon)+
\epsilon^4 \alpha_1 \alpha_2 \alpha_3 H_1 (\epsilon) H_2 (\epsilon) H_3 (\epsilon) =0.$$

\begin{proposition}\label{th: dET inv meas}{\rm\cite{PS}}
The map dET possesses an invariant volume form:
\begin{equation} \nonumber
\det\frac{\partial \widetilde{x}}{\partial x}=
\frac{\phi(\widetilde{x})}{\phi(x)}\quad\Leftrightarrow\quad f^*\omega=\omega,\quad
\omega= \frac{dx_1\wedge dx_2\wedge dx_3}{\phi(x)},
\end{equation}
where $\phi(x)$ is any of the functions
\begin{equation}\nonumber
\phi(x)=(1-\epsilon^2\alpha_i\alpha_j x_k ^2)(1-\epsilon^2\alpha_j\alpha_k x_i ^2)\quad {\rm or}\quad
(1-\epsilon^2\alpha_i\alpha_jx_k^2)^2.
\end{equation}
(The ratio of any two functions $\phi(x)$ is an integral of motion, due to Proposition \ref{th: dE F}).
\end{proposition}
\noindent
The proof is based on formula (\ref{Jac Deltas}) with the matrix $A(x,\epsilon)$ given in (\ref{eq: dET map}). Its determinant is given in (\ref{eq: dET Delta}).

A proper discretization of the Wronskian differential equations (\ref{eq: ET W}) is given by the following statement.
\begin{proposition}\label{th: dET W}
The following relations hold true for dET:
\beq
\label{eq: dET W}
\left\{\begin{array}{l}
 \widetilde{x}_2x_3-x_2\widetilde{x}_3=\epsilon H_1(\epsilon)(\widetilde{x}_1+x_1),\vspace{.1truecm}\\
 \widetilde{x}_3x_1-x_3\widetilde{x}_1=\epsilon H_2(\epsilon)(\widetilde{x}_3+x_3),\vspace{.1truecm}\\
 \widetilde{x}_1x_2-x_1\widetilde{x}_2=\epsilon H_3(\epsilon)(\widetilde{x}_3+x_3),
 \end{array} \right.
\eeq
with the functions $H_i(\epsilon)$ from (\ref{eq: dET ints}).
\end{proposition}
\noindent The proof is based on relations
\beq
\widetilde{x}_i+x_i = \frac{2(1-\epsilon^2\alpha_j\alpha_kx_i^2)(x_i+\epsilon\alpha_ix_jx_k)}
{\Delta(x,\epsilon)},
\label{eq: dET x+x} 
\eeq
\beq
\widetilde{x}_jx_k-x_j\widetilde{x}_k = \frac{2\epsilon(\alpha_jx_k^2-\alpha_kx_j^2)(x_i+\epsilon\alpha_ix_jx_k)}
{\Delta(x,\epsilon)},
\label{eq: dET Xx-xX}
\eeq
which follow easily from the explicit formulas (\ref{eq: dET expl}). They should be compared with
\begin{equation}
\label{eq: dET x-x}
\widetilde{x}_i-x_i=\epsilon\alpha_i (\widetilde{x}_jx_k+x_j\widetilde{x}_k)= \frac{2\epsilon\alpha_i(x_j+\epsilon\alpha_jx_kx_i)(x_k+\epsilon\alpha_kx_ix_j)}
{\Delta(x,\epsilon)}.
\end{equation}

As pointed out in \cite{PPS}, a probable way to the discovery of the conserved quantities of dET in \cite{HK} was through finding the HK bases for this map. In this respect, one has the following results.
\begin{prop}\label{Th: dET full basis}{\rm \cite{PPS}}

${\!\!}$
{\rm (a)} The set
$\Phi=(x_1^2,\, x_2^2,\, x_3^2,\, 1)$
is a HK basis for dET with $\dim K_\Phi(x)=2$. Therefore, any orbit of dET lies on the intersection of two quadrics in $\bbR^3$.

\smallskip

{\rm (b)} The set
$\Phi_0=(x_1^2,\, x_2^2,\, x_3^2)$
is a HK basis for dET with $\dim K_{\Phi_0}(x)=1$. At each point $x\in\bbR^3$ we have:
\[
K_{\Phi_0}(x)=[c_1:c_2:c_3]=[\, \alpha_2 x_3^2-\alpha_3 x_2^2 :
\alpha_3 x_1^2-\alpha_1 x_3^2 : \alpha_1 x_2^2- \alpha_2 x_1^2
\,].
\]
Setting $c_3=-1$, the following functions are integrals of motion of dET:
\begin{equation}
\label{eq: dET c1 c2}
c_1(x)=\frac{\alpha_3 x_2^2-\alpha_2 x_3^2}{\alpha_1 x_2^2-\alpha_2 x_1^2},
\qquad
c_2(x)=\frac{\alpha_1 x_3^2-\alpha_3 x_1^2}{\alpha_1 x_2^2-\alpha_2 x_1^2}.
\end{equation}

\smallskip

{\rm (c)} The set
$\Phi_{12}=(x_1^2,\, x_2^2,\, 1)$
is a further HK basis for dET with $\dim K_{\Phi_{12}}(x)=1$. At each point $x\in\bbR^3$ we have:
$
K_{\Phi_{12}}(x)=[d_1 : d_2 : -1],
$
where
\begin{equation}\label{eq: dET c4 c5}
d_1(x)=-\frac{\alpha_2(1-\epsilon^2\alpha_3\alpha_1 x_2^2)}
{\alpha_1 x_2^2-\alpha_2 x_1^2}\,, \qquad
d_2(x)=\frac{\alpha_1(1-\epsilon^2\alpha_2\alpha_3 x_1^2)}
{\alpha_1 x_2^2-\alpha_2 x_1^2}\,.
\end{equation}
These functions are integrals of motion of dET independent on the integrals (\ref{eq: dET c1 c2}). We have: $K_\Phi(x)=K_{\Phi_0}\oplus K_{\Phi_{12}}$.
\end{prop}
\begin{proof} To prove statement (b), we solve the system
\begin{equation}
\nonumber
\renewcommand{\arraystretch}{1.3}
\left\{ \begin{array}{l}
c_1x_1^2+c_2x_2^2 = x_3^2,\\
c_1\widetilde{x}_1^2+c_2\widetilde{x}_2^2 =
\widetilde{x}_3^2.\end{array}\right.
\end{equation}
The solution is given, according to the Cramer's rule, by ratios of determinants of the type
\begin{equation}
\label{eq: dET for cancel}
\renewcommand{\arraystretch}{1.3}
\left|\begin{array}{cc} x_i^2 & x_j^2 \\ \widetilde{x}_i^2 &
\widetilde{x}_j^2\end{array}\right|=\frac{4\epsilon(\alpha_jx_i^2-\alpha_ix_j^2)
(x_1+\epsilon\alpha_1x_2x_3)(x_2+\epsilon\alpha_2x_3x_1)(x_3+\epsilon\alpha_3x_1x_2)}
{\Delta^2(x,\epsilon)}.
\end{equation}
(Here we used (\ref{eq: dET Xx-xX}), (\ref{eq: dET x-x})).
In the ratios of such determinants everything cancels out, except for the factors $\alpha_jx_i^2-\alpha_ix_j^2$, so we end up with (\ref{eq: dET c1 c2}). The cancelation of the denominators $\Delta^2(x,\epsilon)$ is, of course, no wonder, but the cancelation of all the non-even factors in the numerators is rather {\em remarkable and miraculous} and is not granted by any well-understood mechanism. Since the components of the solution do not depend on $\epsilon$, we conclude that functions (\ref{eq: dET c1 c2}) are integrals of motion of dET.

To prove statement (c), we solve the system
\[
\renewcommand{\arraystretch}{1.3}
\left\{ \begin{array}{l} d_1 x_1^2+d_2 x_2^2 = 1,\\
d_1\widetilde{x}_1^2+d_2\widetilde{x}_2^2 =
1.\end{array}\right.
\]
The solution is given by eq. (\ref{eq: dET c4 c5}), due to eq. \eqref{eq: dET for cancel} and the similar formula
\[
\renewcommand{\arraystretch}{1.3}
\left|\begin{array}{cc} 1 & x_i^2 \\ 1 &
\widetilde{x}_i^2\end{array}\right|
=\frac{4\epsilon\alpha_i(1-\epsilon^2\alpha_j\alpha_kx_i^2)
(x_1+\epsilon\alpha_1x_2x_3)(x_2+\epsilon\alpha_2x_3x_1)(x_3+\epsilon\alpha_3x_1x_2)}
{\Delta^2(x,\epsilon)},
\]
which, in turn, follows from (\ref{eq: dET x+x}) and (\ref{eq: dET Xx-xX}). This time the solution does depend on $\epsilon$, but consists of manifestly even functions of $\epsilon$. Everything non-even luckily cancels, again. Therefore, functions (\ref{eq: dET c4 c5}) are integrals of motion of dET.
\end{proof}

Although each one of the HK bases $\Phi_0$, $\Phi_1$ delivers apparently two integrals of motion (\ref{eq: dET c1 c2}), each pair turns out to be {\em functionally dependent}, as
\[
\alpha_1 c_1(x)+\alpha_2 c_2(x)=\alpha_3, \quad \alpha_1d_1(x)+\alpha_2d_2(x)=\epsilon^2\alpha_1\alpha_2\alpha_3.
\]
However, functions $c_1,c_2$ are independent on $d_1,d_2$, since the former depend on $x_3$, while the latter do not.

Of course, permutational symmetry yields that each of the sets of monomials $\Phi_{23}=(x_2^2,\, x_3^2,\, 1)$ and $\Phi_{13}=(x_1^2,\, x_3^2,\, 1)$ is a HK basis, as well, with $\dim K_{\Phi_{23}}(x)=\dim K_{\Phi_{13}}(x)=1$. But we do not obtain additional linearly independent null-spaces, as any two of the four found one-dimensional null-spaces span the full null-space $K_\Phi(x)$.

Summarizing, we have found a HK basis with a two-dimensional null-space, as well as two functionally independent conserved quantities for the HK discretization of the Euler top. Both results yield integrability of this discretization, in the sense that its orbits are confined to closed curves in $\bbR^3$. Moreover, each such curve is an intersection of two quadrics, which in the general position case is an elliptic curve.

\begin{proposition}\label{th: dET biquad}
Each component $x_i$ of any solution of dET satisfies a relation of the type
$
P_i(x_i,\widetilde{x}_i)=0,
$
where $P_i$ is a biquadratic polynomial whose coefficients are
integrals of motion of dET:
\[
P_i(x_i,\widetilde{x}_i) = p_{i}^{(3)} x_i^2 \widetilde{x}_i^2 +
p_{i}^{(2)} ( x_i^2+\widetilde{x}_i^2) + p_{i}^{(1)} x_i
\widetilde{x}_i + p_{i}^{(0)}=0,
\]
with
\bea
& p_{i}^{(3)}=-4\epsilon^2 \alpha_j \alpha_k, \quad
p_{i}^{(2)}=\big(1+\epsilon^2 \alpha_j H_j(\epsilon)\big)\big(1-\epsilon^2 \alpha_k H_k (\epsilon)\big), & \nonumber\\
& p_{i}^{(1)}=-2\big(1-\epsilon^2 \alpha_j H_j (\epsilon)\big)\big(1+\epsilon^2 \alpha_k H_k (\epsilon)\big), \quad
p_{i}^{(0)}=4\epsilon^2 H_j (\epsilon) H_k (\epsilon). & \nonumber
\eea
\end{proposition}
\begin{proof}
From eqs. (\ref{eq: dET x}) and (\ref{eq: dET W}) there follows:
\[
\frac{(\widetilde x_i-x_i)^2}{(\epsilon\alpha_i)^2}+\epsilon^2 H_i^2(\epsilon)(\widetilde x_i+x_i)^2=2(\widetilde x_j^2x_k^2+x_j^2\widetilde x_k^2).
\]
It remains to express $x_j^2$ and $x_k^2$ through $x_i^2$ and integrals $H_j(\epsilon)$, $H_k(\epsilon)$ given in eq. (\ref{eq: dET ints}).
\end{proof}

It follows from Proposition \ref{th: dET biquad} that solutions $x_i(t)$ as functions of the discrete time $t\in 2\epsilon\mathbb Z$ are given by elliptic functions of order 2 (the order of an elliptic function is the number of the zeroes or poles it possesses in a period parallelogram).

We would like to point out that Propositions \ref{th: dET W} and \ref{th: dET biquad} can be interpreted as existence of further HK bases. For instance, according to Proposition \ref{th: dET W}, each pair $(\widetilde{x}_jx_k-x_j\widetilde{x}_k,\widetilde{x}_i+x_i)$ is a HK basis with a 1-dimensional null-space. Similarly, Proposition \ref{th: dET biquad} says that for each $i=1,2,3$, the set $x_i^p\widetilde{x}_i^q$ $(0\le p,q\le 2)$ is a HK basis with a 1-dimensional null-space. Of course, due to the dependence on the shifted variables $\widetilde x$, these HK bases consist of complicated functions of $x$ rather than of monomials. A further instance of HK bases of this sort is given in the following statement. Compared with Proposition \ref{Th: dET full basis}, it says that for dET, {\em for each HK basis consisting of monomials quadratic in $x$, the corresponding set of monomials bilinear in $x,\wx$ is a HK basis, as well}. This seems to be a quite general phenomenon, further issues of which will appear later several times.

\begin{prop}\label{Th: dET full basis new}
${}$

{\rm (a)} The set
$\Psi=(\wx_1x_1,\, \wx_2x_2,\, \wx_3x_3,\, 1)$
is a HK basis for dET with $\dim K_\Psi(x)=2$.

\smallskip

{\rm (b)} The set
$\Psi_0=(\wx_1x_1,\, \wx_2x_2,\, \wx_3x_3)$
is a HK basis for dET with $\dim K_{\Psi_0}(x)=1$. At each point $x\in\bbR^3$, the homogeneous coordinates $\bar c_i$ of the null-space $K_{\Psi_0}(x)=[\bar c_1:\bar c_2:\bar c_3]$ are given by
\[
\bar c_i=(\alpha_j x_k^2-\alpha_k x_j^2)\big(1-\epsilon^2(\alpha_i\alpha_jx_k^2+\alpha_k\alpha_ix_j^2-\alpha_j\alpha_kx_i^2)\big).
\]
The quotients $\bar c_i/\bar c_j$ are integrals of motion of dET.

\smallskip

{\rm (c)} The set
$\Psi_{12}=(\wx_1x_1,\, \wx_2x_2,\, 1)$
is a further HK basis for dET with $\dim K_{\Psi_{12}}(x)=1$. At each point $x\in\bbR^3$, there holds:
$
K_{\Psi_{12}}(x)=[\bar d_1 : \bar d_2 : -1],
$
where
\begin{eqnarray*}\label{eq: dET c4 c5 new}
&& \bar d_1(x) = -\frac{\alpha_2(1-\epsilon^2\alpha_3\alpha_1 x_2^2)}{\alpha_1 x_2^2-\alpha_2 x_1^2}
\frac{1-\epsilon^2(\alpha_2\alpha_3x_1^2-\alpha_3\alpha_1x_2^2+\alpha_1\alpha_2x_3^2)}
{1-\epsilon^2(\alpha_2\alpha_3x_1^2+\alpha_3\alpha_1x_2^2-\alpha_1\alpha_2x_3^2)}\,, \\
&& \bar d_2(x) = \frac{\alpha_1(1-\epsilon^2\alpha_2\alpha_3 x_1^2)}
{\alpha_1 x_2^2-\alpha_2 x_1^2}
\frac{1-\epsilon^2(\alpha_3\alpha_1x_2^2-\alpha_2\alpha_3x_1^2+\alpha_1\alpha_2x_3^2)}
{1-\epsilon^2(\alpha_3\alpha_1x_2^2+\alpha_2\alpha_3x_1^2-\alpha_1\alpha_2x_3^2)},
\end{eqnarray*}
are integrals of dET. We have: $K_\Psi(x)=K_{\Psi_0}(x)\oplus K_{\Psi_{12}}(x)$.
\end{prop}

\section{Zhukovski-Volterra system}
\label{Sect: dZV}

The gyroscopic Zhukovski-Volterra (ZV) system is a generalization of the Euler top.
It describes the free motion of a rigid body carrying an asymmetric rotor (gyrostat) \cite{Vo}.
Equations of motion of the ZV system read
\beq
\left\{ \begin{array}{l}
\dot{x}_1=  \alpha_1 x_2 x_3 + \beta_3 x_2 - \beta_2 x_3 ,  \vspace{.1truecm} \\
 \dot{x}_2=  \alpha_2 x_3 x_1 + \beta_1 x_3 - \beta_3 x_1, \vspace{.1truecm} \\
\dot{x}_3=  \alpha_3 x_1 x_2 + \beta_2 x_1 - \beta_1 x_2,
\end{array} \right. \label{zv}
\eeq
with $\alpha_i,\beta_i$ being real parameters of the system. For $(\beta_1,\beta_2,\beta_3)=(0,0,0)$, the flow (\ref{zv}) reduces to the Euler top (\ref{eq: ET x}). The ZV system is (Liouville and algebraically) integrable under the condition
\begin{equation}\label{eq: zv cond}
\alpha_1+\alpha_2+\alpha_3=0.
\end{equation}
It can be explicitly integrated in terms of elliptic functions, see \cite{Vo} and also \cite{BA} for a more recent exposition. The following quantities are integrals of motion of the ZV system:
\bea
&& H_{1}= \alpha_2 x_3^2 -\alpha_3 x_2^2 -2 ( \beta_1 x_1 + \beta_2 x_2 + \beta_3 x_3),
 \nonumber\label{eq: ZV i1}\\
 && H_{2}= \alpha_3 x_1^2  - \alpha_1 x_3^2 -2 ( \beta_1 x_1 + \beta_2 x_2 + \beta_3 x_3),
 \label{eq: ZV i2} \\
 && H_{3}=  \alpha_1 x_2^2 - \alpha_2 x_1^2 -2 ( \beta_1 x_1 + \beta_2 x_2 + \beta_3 x_3).
 \nonumber\label{eq: ZV i3}
\eea
Clearly, only two of them are functionally independent because of
$\alpha_1 H_1+\alpha_2 H_2+\alpha_3H_3=0$. Note that
\begin{equation}\nn
H_2-H_1=\alpha_3C,\quad H_3-H_2=\alpha_1C, \quad H_1-H_3=\alpha_2C,
\end{equation}
with $C=x_1^2+x_2^2+x_3^2$.

As in the Euler case, the Wronskians of the coordinates $x_j$ admit quadratic expressions with coefficients dependent on the integrals of motion:
\begin{equation}
\label{eq: ZV Wr}
\left\{ \begin{array}{l}
\dot{x}_2x_3-x_2\dot{x}_3=H_1x_1 + x_1(\beta_1 x_1+\beta_2 x_2+\beta_3 x_3)+\beta_1 C, \vspace{.1truecm}\\
\dot{x}_3x_1-x_3\dot{x}_1=H_2x_2 + x_2(\beta_1 x_1+\beta_2 x_2+\beta_3 x_3)+\beta_2 C, \vspace{.1truecm}\\
\dot{x}_1x_2-x_1\dot{x}_2=H_3x_3 + x_3(\beta_1 x_1+\beta_2 x_2+\beta_3 x_3)+\beta_3 C.
\end{array} \right.
\end{equation}

The HK discretization of the ZV system is:
\beq \label{eq: dZV x}
\left\{ \begin{array}{l}
\widetilde{x}_1-x_1= \epsilon \alpha_1(\widetilde{x}_2 x_3+x_2\widetilde x_3)
+\epsilon\beta_3(\widetilde x_2+x_2)-\epsilon\beta_2(\widetilde x_3+x_3) ,  \vspace{.2truecm} \\
\widetilde{x}_2-x_2=\epsilon  \alpha_2(\widetilde x_3 x_1+x_3\widetilde x_1)
+\epsilon\beta_1(\widetilde x_3+x_3)-\epsilon\beta_3(\widetilde x_1+x_1), \vspace{.2truecm} \\
\widetilde{x}_3-x_3=\epsilon \alpha_3(\widetilde x_1 x_2+x_1\widetilde x_2)
+\epsilon\beta_2(\widetilde x_1+x_1)-\epsilon\beta_1(\widetilde x_2+x_2).
\end{array} \right.
\eeq
The map $f:x\mapsto\widetilde{x}\ $ obtained by solving
(\ref{eq: dZV x}) for $\widetilde{x}$ is given by:
\begin{equation} \nonumber
\widetilde{x} =f(x,\epsilon)=A^{-1}(x,\epsilon)(\mathds{1}+\epsilon B)x,
\end{equation}
with
$$
A(x,\epsilon)=
\begin{pmatrix} 1 & - \epsilon \alpha_1 x_3 & -\epsilon \alpha_1 x_2   \\
-\epsilon\alpha_2 x_3& 1 &-\epsilon \alpha_2 x_1 \\
-\epsilon\alpha_3 x_2 & -\epsilon\alpha_3 x_1 & 1 \end{pmatrix} - \epsilon B,\quad
B=\begin{pmatrix} 0 & \beta_3 &- \beta_2 \\
 - \beta_3& 0 & \beta_1\\
\beta_2&- \beta_1 & 0 \end{pmatrix}.
$$
We will call this map dZV. Formula (\ref{Jac Deltas}) holds true for dZV, as for any HK discretization.

\subsection{ZV system with two vanishing ${\beta_k}$'s}
\label{Sect: dZV1}

In the case when two out of three $\beta_k$'s vanish, say $\beta_2=\beta_3=0$, the condition
 (\ref{eq: zv cond}) is not necessary for integrability of the ZV system. The functions $H_2$ and $H_3$ as given in (\ref{eq: ZV i2}) (with $\beta_2=\beta_3=0$) are in this case conserved quantities without any condition on $\alpha_k$'s, while their linear combinations $H_1$ and $C$ are given by
\begin{eqnarray*}
&&H_1 = -\frac{1}{\alpha_1}(\alpha_2H_2+\alpha_3H_3)=
\alpha_2x_3^2-\alpha_3x_2^2+2\beta_1\frac{\alpha_2+\alpha_3}{\alpha_1}x_1,\\
&&C = \frac{1}{\alpha_1}(H_3-H_2)=x_2^2+x_3^2-\frac{\alpha_2+\alpha_3}{\alpha_1}x_1^2.
\end{eqnarray*}
Wronskian relations (\ref{eq: ZV Wr}) are replaced by
\begin{equation}\label{eq: ZV1 W}
\left\{ \begin{array}{l}
\dot{x}_2x_3-x_2\dot{x}_3=H_1x_1 - \beta_1\dfrac{\alpha_2+\alpha_3}{\alpha_1}\,x_1^2+\beta_1 C, \vspace{.1truecm}\\
\dot{x}_3x_1-x_3\dot{x}_1=H_2x_2 + \beta_1x_1x_2, \vspace{.1truecm}\\
\dot{x}_1x_2-x_1\dot{x}_2=H_3x_3 + \beta_1x_1x_3.
\end{array} \right.
\end{equation}

The HK discretization of the ZV system with  $\beta_2=\beta_3=0$ turns out to possess two conserved quantities (without imposing condition (\ref{eq: zv cond})) and an invariant measure.
\begin{proposition}\label{th: dZV F}
The functions
\bea
&& H_2(\epsilon)= \frac{\alpha_3 x_1^2 - \alpha_1 x_3^2 - 2\beta_1 x_1
+ \epsilon^2 \beta_1^2 \alpha_1 x_2^2} {1-\epsilon^2 \alpha_3 \alpha_1 x_2^2}, \nonumber \\
&& H_3(\epsilon)= \frac{\alpha_1 x_2^2 - \alpha_2 x_1^2 - 2\beta_1 x_1
- \epsilon^2 \beta_1^2 \alpha_1 x_3^2} {1-\epsilon^2 \alpha_1 \alpha_2 x_3^2}, \nonumber
\eea
are conserved quantities of dZV with $\beta_2=\beta_3=0$.
\end{proposition}

\begin{proposition}\label{th: dZV inv meas}
The map dZV  with $\beta_2=\beta_3=0$ possesses an invariant volume form:
\begin{equation} \nonumber
\det\frac{\partial \widetilde{x}}{\partial x}=
\frac{\phi(\widetilde{x})}{\phi(x)}\quad\Leftrightarrow\quad f^*\omega=\omega,\quad
\omega= \frac{dx_1\wedge dx_2\wedge dx_3}{\phi(x)},
\end{equation}
with $\phi(x)=(1-\epsilon^2 \alpha_3 \alpha_1 x_2^2)(1-\epsilon^2\alpha_1  \alpha_2 x_3^2)$.
\end{proposition}

The conserved quantities of Proposition \ref{th: dZV F} appear on the right-hand sides of the following relations which are the discrete versions of the Wronskian relations (\ref{eq: ZV1 W}):

\begin{proposition}\label{prop: dZV1 W}
The following relations hold true for dZV with $\beta_2=\beta_3=0$:
\begin{equation} \nonumber
\left\{ \begin{array}{l}
\widetilde{x}_2x_3-x_2\widetilde{x}_3= \epsilon  c_1(\widetilde
x_1+x_1) +2\epsilon c_2\widetilde x_1 x_1+2\epsilon c_3 , \vspace{.1truecm}\\
\widetilde{x}_3x_1-x_3\widetilde{x}_1 =\epsilon  H_2(\epsilon)
(\widetilde x_2+x_2)+\epsilon\beta_1(\widetilde x_1x_2+x_1\widetilde
x_2), \vspace{.1truecm}\\
\widetilde{x}_1x_2-x_1\widetilde{x}_2 = \epsilon H_3(\epsilon)
(\widetilde x_3+x_3)+\epsilon\beta_1(\widetilde x_1x_3+x_1\widetilde
x_3),
\end{array} \right.
\end{equation}
with
\[
c_1= -\frac{\alpha_2H_2(\epsilon)+\alpha_3H_3(\epsilon)}{\alpha_1\Delta},\quad
c_2=-\frac{\beta_1(\alpha_2+\alpha_3)} {\alpha_1\Delta},\quad
c_3=\frac{\beta_1\big(H_3(\epsilon)-H_2(\epsilon)\big)} {\alpha_1\Delta},
\]
\[
\Delta = 1+\epsilon^4\big(\alpha_2H_3(\epsilon)-\beta_1^2\big)\big(\alpha_3H_2(\epsilon)+\beta_1^2\big).
\]
\end{proposition}

Next, we describe the HK bases found in this case.

\begin{prop}\label{Prop: dZV HK basis}
${}$

{\rm (a)} The set $\Phi= (x_1^2,\, x_2^2,\, x_3^2,\, x_1,\, 1)$ is a HK basis for dZV with $\beta_2=\beta_3=0$, with  $\dim K_{\Phi}(x)=2$. Any orbit of dZV with $\beta_2=\beta_3=0$ is thus confined to the intersection of two quadrics in $\bbR^3$.

\smallskip

{\rm (b)}
The set $\Phi_0=(x_1^2,\, x_2^2,\, x_3^2,\, 1)$ is a HK basis for dZV with $\beta_2=\beta_3=0$, with $\dim K_{\Phi_0}(x)=1$. At each point $x\in\bbR^3$ we have: $K_{\Phi_0}(x)=[ -1 : d_2 : d_3 : d_4]$, where
\begin{equation} \nonumber
d_2=\frac{\alpha_1}{\alpha_2+\alpha_3}\big(1-\epsilon^2\beta_1^2-\epsilon^2\alpha_3H_2(\epsilon)\big),
\quad
d_3=\frac{\alpha_1}{\alpha_2+\alpha_3}\big(1-\epsilon^2\beta_1^2+\epsilon^2\alpha_2H_3(\epsilon)\big),
\end{equation}
\begin{equation}\nonumber
d_4=\frac{1}{\alpha_2+\alpha_3}\big(H_2(\epsilon)-H_3(\epsilon)\big).
\end{equation}

\smallskip

{\rm (c)} The set $\Phi_{23} = (x_2^2,\, x_3^2,\, x_1,\, 1)$ is a HK basis for dZV with $\beta_2=\beta_3=0$, with  $\dim K_{\Phi_{23}}(x)=1$. At each point $x\in\bbR^3$ we have: $K_{\Phi_{23}}(x)=[c_1:c_2:c_3:c_4]$, where
\[
c_1=\alpha_1\big(\alpha_3+\epsilon^2\beta_1^2\alpha_2+\epsilon^2\alpha_2\alpha_3H_2(\epsilon)\big),\quad
c_2=-\alpha_1\big(\alpha_2+\epsilon^2\beta_1^2\alpha_3-\epsilon^2\alpha_2\alpha_3H_3(\epsilon)\big),
\]
\[
c_3=-2\beta_1(\alpha_2+\alpha_3),\quad c_4=-\big(\alpha_2H_2(\epsilon)+\alpha_3H_3(\epsilon)\big).
\]
\end{prop}

Unlike the case of dET, we see that here a HK basis with a one dimensional null-space already provides more than one independent integral of motion.

``Bilinear'' versions of the above HK bases also exist:
\begin{prop}\label{Prop: dZV HK basis bilinear}

The set $\Psi= (x_1\wx_1,\, x_2\wx_2,\, x_3\wx_3,\, x_1+\wx_1,\, 1)$ is a HK basis for dZV with $\beta_2=\beta_3=0$, with  $\dim K_{\Psi}(x)=2$. The sets
\[
\Psi_0=(x_1\wx_1,\, x_2\wx_2,\, x_3\wx_3,\, 1)\quad {and}\quad
\Psi_{23} = (x_2\wx_2,\, x_3\wx_3,\, x_1+\wx_1,\, 1)
\]
are HK bases with one-dimensional null-spaces.
\end{prop}
The following statement is a starting point towards an explicit integration of the map dZV with $\beta_2=\beta_3=0$ in terms of elliptic functions.
\begin{proposition}\label{th: dZV biquad}
The component $x_1$ of the solution of the difference equations
(\ref{eq: dZV x}) satisfies a relation of the type
\begin{equation} \nonumber
P(x_1,\widetilde{x}_1)=p_0x_1^2 \widetilde{x}_1^2 +
p_1 x_1 \widetilde{x}_1(x_1+\widetilde{x}_1)
+p_2(x_1^2+\widetilde{x}_1^2) + p_3x_1\widetilde{x}_1
+p_4(x_1+\widetilde{x}_1) + p_5=0,
\end{equation}
coefficients of the biquadratic polynomial  $P$  being conserved quantities of dZV with $\beta_2=\beta_3=0$.
\end{proposition}
Proof is parallel to that of Proposition \ref{th: dET biquad}.

\subsection{ZV system with one vanishing ${\beta_k}$}
\label{Sect: dZV2}
In the case $\beta_3=0$ (say) and generic values of other parameters, the ZV system has only one integral $H_3$ and is therefore non-integrable. One of the Wronskian relations holds true in this general situation:
\begin{equation}\label{ZV 2betas W}
\dot{x}_1x_2-x_1\dot{x}_2=H_3x_3+\beta_1x_1x_3+\beta_2x_2x_3.
\end{equation}
Under condition (\ref{eq: zv cond}), the ZV system becomes integrable, with all the results formulated in the general case.

Similarly, the map dZV with $\beta_3=0$ and generic values of other parameters possesses one conserved quantity:
$$
H_3(\epsilon)=\frac{\alpha_1 x_2^2 -  \alpha_2
 x_1^2 - 2 (\beta_1 x_1+\beta_2 x_2)
- \epsilon^2 (\beta_1^2  \alpha_1 +\beta_2^2 \alpha_2)x_3^2}
{1-\epsilon^2\alpha_1 \alpha_2 x_3^2}.
$$
Clearly, this fact can be re-formulated as the existence of a HK basis $\Phi=(x_1^2,x_2^2,x_3^2,x_1,x_2,1)$ with $\dim K_{\Phi}=1$. The Wronskian relation (\ref{ZV 2betas W}) possesses a decent discretization:
\begin{equation}\nn
\wx_1x_2-x_1\wx_2= \epsilon H_3(\epsilon)(x_3+\wx_3)+ \epsilon \beta_1(\wx_1x_3+x_1\wx_3)+ \epsilon \beta_2(\wx_2x_3+x_2\wx_3).
\end{equation}
However, it seems that the map dZV with $\beta_3=0$ does not acquire an additional integral of motion under condition (\ref{eq: zv cond}). It might be conjectured that in order to assure the integrability of the dZV map with  $\beta_3=0$, its other parameters have to satisfy some relation which is an $O(\epsilon)$-deformation of (\ref{eq: zv cond}).

\subsection{ZV system with all ${\beta_k}$'s non-vanishing}
\label{Sect: dZV3}

Numerical experiments indicate non-integra\-bility for the map (\ref{eq: dZV x})  with non-vanishing ${\beta_k}$'s, even under condition (\ref{eq: zv cond}). Nevertheless, some other relation between the parameters might yield integrability. In this connection we notice that the map dZV with $(\alpha_1,\alpha_2,\alpha_3)=(\alpha, -\alpha, 0)$ admits a polynomial conserved quantity
$$
H= -\alpha x_3^2-2(\beta_1 x_1+\beta_2 x_2+\beta_3 x_3)+\epsilon^2 \alpha(\beta_2 x_1-\beta_1 x_2)^2.
$$

\section{Volterra chain}
\label{Sect: dV}

\subsection{Periodic Volterra chain with $N=3$ particles}
Equations of motion of the periodic Volterra chain with three particles (VC$_3$, for short):
\begin{equation}\label{eq: VL3}
\left\{ \begin{array}{l}
\dot{x}_1= x_1(x_2-x_3),  \vspace{.1truecm} \\
\dot{x}_2= x_2(x_3-x_1), \vspace{.1truecm} \\
\dot{x}_3= x_3(x_1-x_2).
\end{array} \right.
\end{equation}
This system is Liouville and algebraically integrable, with the following two independent integrals of motion:
\[
H_1=x_1+x_2+x_3, \qquad H_2=x_1x_2x_3.
\]
There hold the following Wronskian relations:
\begin{equation}\label{eq: VL3 Wronski}
\dot{x}_ix_j-x_i\dot{x}_j=H_1x_ix_j-3H_2.
\end{equation}
Eliminating $x_j,x_k$ from equation of motion for $x_i$ with the help of integrals of motion, one arrives at
\begin{equation}\nn
\dot{x}_i^2=x_i^2(x_i-H_1)^2-4H_2x_i,
\end{equation}
which yields a solution in terms of elliptic functions.

The HK discretization of system (\ref{eq: VL3}) (with the time step $2\epsilon$) is:
\begin{equation}\label{eq: dVL3}
\left\{ \begin{array}{l}
\wx_1-x_1=\epsilon x_1(\wx_2-\wx_3)+\epsilon\wx_1(x_2-x_3),  \vspace{.1truecm} \\
\wx_2-x_2=\epsilon x_2(\wx_3-\wx_1)+\epsilon\wx_2(x_3-x_1), \vspace{.1truecm} \\
\wx_3-x_3=\epsilon x_3(\wx_1-\wx_2)+\epsilon\wx_3(x_1-x_2).
\end{array} \right.
\end{equation}
The map $f:x\mapsto\wx$ obtained by solving (\ref{eq: dVL3}) for $\wx$ is given by:
\[
\wx=f(x,\epsilon)=A^{-1}(x,\epsilon)x,
\]
with
\[
A(x,\epsilon)=\begin{pmatrix}
1+\epsilon(x_3-x_2) & -\epsilon x_1 & \epsilon x_1 \\
\epsilon x_2 & 1+\epsilon(x_1-x_3) & -\epsilon x_2 \\
-\epsilon x_3 & \epsilon x_3 & 1+\epsilon(x_2-x_1)
\end{pmatrix}.
\]
Explicitly:
\begin{equation}
\label{eq: dVL3 expl}
    \wx_i=x_i\ \frac{1+2\epsilon(x_j-x_k)+\epsilon^2\big((x_j+x_k)^2-x_i^2\big)}
    {1-\epsilon^2(x_1^2+x_2^2+x_3^2-2x_1x_2-2x_2x_3-2x_3x_1)}.
\end{equation}
This map will be called dVC$_3$. From Proposition \ref{th: dLV type} there follows immediately:
\begin{proposition}\label{th: dVL inv meas}
The map dVC$_3$ possesses an invariant volume form:
\begin{equation} \nonumber
\det\frac{\partial \widetilde{x}}{\partial x}=
\frac{\phi(\widetilde{x})}{\phi(x)}\quad\Leftrightarrow\quad f^*\omega=\omega,\quad
\omega= \frac{dx_1\wedge dx_2\wedge dx_3}{\phi(x)},
\end{equation}
with $\phi(x)=x_1x_2x_3$.
\end{proposition}
Concerning integrability of dVC$_3$, we note first of all that $H_1$ is an obvious conserved quantity. The second one is most easily obtained from the following discretization of Wronskian relations (\ref{eq: VL3 Wronski}).
\begin{proposition}
For the map dVC$_3$, the following relations hold:
\begin{equation}\label{eq: dVL3 Wronski}
\wx_ix_j-x_i\wx_j=\epsilon H_1(\wx_ix_j+x_i\wx_j)-6\epsilon H_2(\epsilon)\left(1-\textstyle\frac{1}{3}\epsilon^2H_1^2\right),
\end{equation}
where $H_2(\epsilon)$ is a conserved quantity, given by
\begin{equation}\label{eq: dVL3 H2}
    H_2(\epsilon)=\frac{x_1x_2x_3}{1-\epsilon^2(x_1^2+x_2^2+x_3^2-2x_1x_2-2x_2x_3-2x_3x_1)}.
\end{equation}
\end{proposition}
\begin{proof}
Define $H_2(\epsilon)$  by eq. (\ref{eq: dVL3 Wronski}). It is easily computed with explicit formulas (\ref{eq: dVL3 expl}). The result given by (\ref{eq: dVL3 H2}) is a manifestly even function of $\epsilon$ and therefore an integral of motion.
\end{proof}

\begin{prop}
${}$

{\rm (a)}
The  set $\Phi_{ij}=(x_ix_j(x_i+x_j),\, x_i^2+x_j^2,\, x_ix_j,\,  x_i+x_j,\, 1)$
is a HK basis for the map dVC$_3$  with $\dim K_{\Phi_{ij}}(x)=1$. In other words, the pairs $(x_i,x_j)$ lie on a cubic curve
\[
P(x_i,x_j)=p_0x_ix_j(x_i+x_j)+p_1(x_i^2+x_j^2)+p_2x_ix_j+p_3(x_i+x_j)+p_4=0,
\]
whose coefficients $p_m$ are constant (expressed through integrals of motion).
\smallskip

{\rm (b)}
The  set $\Psi_{i}=(x_i^2\wx_i^2,\, x_i\wx_i(x_i+\wx_i),\, x_i^2+\wx_i^2,\, x_i\wx_i,\, x_i+\wx_i,\, 1)$ is a HK basis for the map dVC$_3$  with $\dim K_{\Psi_i}(x)=1$. In other words, the pairs $(x_i,\wx_i)$ lie on a symmetric biquadratic curve with constant coefficients (which can be expressed through integrals of motion).
\end{prop}
\begin{proof}
Statement (a) follows by eliminating $x_k$ from (\ref{eq: dVL3 H2}) via $x_k=H_1-x_i-x_j$. Statement (b) is obtained with the help of MAPLE; it implies that $x_i$ as functions of $t$ are elliptic functions of degree 2 (i.e., with two poles within one parallelogram of periods).
\end{proof}

\subsection{Periodic Volterra chain with $N=4$ particles}
Equations of motion of VC$_4$ are:
\begin{equation} \nn
\left\{ \begin{array}{l}
\dot{x}_1= x_1(x_2-x_4),  \vspace{.1truecm} \\
\dot{x}_2= x_2(x_3-x_1), \vspace{.1truecm} \\
\dot{x}_3= x_3(x_4-x_2), \vspace{.1truecm} \\
\dot{x}_4= x_4(x_1-x_3).
\end{array} \right.
\eeq
This system possesses three obvious integrals of motion: $H_1=x_1+x_2+x_3+x_4$, $H_2=x_1x_3$, and $H_3=x_2x_4$. One easily finds that $x_1,x_3$ satisfy the differential equation
\[
\dot{x}_1^2=(x_1^2-H_1x_1+H_2)^2-4H_3x_1^2,
\]
while $x_2,x_4$ satisfy a similar equation with $H_2\leftrightarrow H_3$. This immeadiately leads to solution in terms of elliptic functions. There are two types of Wronskian relations:
\begin{equation}\label{eq: VL4 Wronski1}
    \dot{x}_1x_3-x_1\dot{x}_3=2H_2(x_2-x_4),\qquad \dot{x}_2x_4-x_2\dot{x}_4=2H_3(x_3-x_1),
\end{equation}
and
\begin{equation}\nn
    \dot{x}_1x_2-x_1\dot{x}_2=H_1x_1x_2-2H_2x_2-2H_3x_1.
\end{equation}

HK discretization (denoted by dVC$_4$):
\begin{equation}\nn
\left\{ \begin{array}{l}
\wx_1-x_1=\epsilon x_1(\wx_2-\wx_4)+\epsilon\wx_1(x_2-x_4),  \vspace{.1truecm} \\
\wx_2-x_2=\epsilon x_2(\wx_3-\wx_1)+\epsilon\wx_2(x_3-x_1), \vspace{.1truecm} \\
\wx_3-x_3=\epsilon x_3(\wx_4-\wx_2)+\epsilon\wx_3(x_4-x_2),  \vspace{.1truecm} \\
\wx_4-x_4=\epsilon x_4(\wx_1-\wx_3)+\epsilon\wx_4(x_1-x_4).
\end{array} \right.
\end{equation}
It possesses an obvious integral $H_1=x_1+x_2+x_3+x_4$.
\begin{proposition}
For the map dVC$_4$, the following natural discretization of eqs. (\ref{eq: VL4 Wronski1}) holds:
\begin{eqnarray}
&&\wx_1x_3-x_1\wx_3= 2\epsilon H_2(\epsilon)(x_2+\wx_2-x_4-\wx_4),\nonumber\\
&&\wx_2x_4-x_2\wx_4= 2\epsilon H_3(\epsilon)(x_1+\wx_1-x_3-\wx_3)\nn,
\end{eqnarray}
with the conserved quantities
\begin{equation}\nn
    H_2(\epsilon)=\frac{x_1x_3}{1-\epsilon^2(x_2-x_4)^2},\qquad
    H_3(\epsilon)=\frac{x_2x_4}{1-\epsilon^2(x_1-x_3)^2}.
\end{equation}
\end{proposition}
\begin{proof} This can be shown directly; the fact that $H_2(\epsilon)$, $H_3(\epsilon)$ are even functions of $\epsilon$, assures that they are conserved quantities. One can also show this immediately from equations of motion: for instance, multiplying the equations
\[
\frac{\wx_1}{1+\epsilon(\wx_2-\wx_4)}=\frac{x_1}{1-\epsilon(x_2-x_4)},\qquad
\frac{\wx_3}{1-\epsilon(\wx_2-\wx_4)}=\frac{x_3}{1+\epsilon(x_2-x_4)},
\]
shows that $H_2(\epsilon)$ is a conserved quantity.
\end{proof}

\begin{prop}
${}$

{\rm (a)} For the iterates of map dVC$_4$, the pairs $(x_1,x_2)$ lie on a quartic curve whose coefficients are constant (expressed through integrals of motion).
\smallskip

{\rm (b)}
The pairs $(x_i,\wx_i)$ lie on a biquartic curve of genus 1 with constant coefficients (which can be expressed through integrals of motion).
\end{prop}
\begin{proof}
Statement (a) follows by eliminating $x_3,x_4$ from integrals $H_1,H_2(\epsilon),H_3(\epsilon)$. Statement (b) is obtained with the help of MAPLE; it implies that $x_i$ as functions of $t$ are elliptic functions of degree 4 (i.e., with four poles within one parallelogram of periods).
\end{proof}

Note that the reduction $x_4=0$ of the periodic Volterra chain with $N=4$ leads to the open-end Volterra chain with $N=3$ particles.

\section{Dressing chain ($N=3$)}
\label{Sect: dDC}

The three-dimensional dressing chain (DC$_3$, for short) is described by the following system of quadratic ordinary equations \cite{VS}:
\beq
\left\{ \begin{array}{l}
\dot{x}_1=  x_3^2 -x_2^2  +\alpha_3 -\alpha_2 ,  \vspace{.1truecm} \\
 \dot{x}_2=  x_1^2 -x_3^2  +\alpha_1 -\alpha_3, \vspace{.1truecm} \\
\dot{x}_3=  x_2^2 -x_1^2  +\alpha_2 -\alpha_1,
\end{array} \right. \label{DC}
\eeq
with real parameters $\alpha_i$. The system (\ref{DC}) is (Liouville and algebraically) integrable. The following quantities are integrals of motion:
\bea
&& I_{1}= x_1 +x_2 +x_3,
 \nonumber
 \\
 && I_{2}= (x_1+x_2 )(x_2+x_3) (x_3+x_1)-\alpha_1x_1-\alpha_2x_2-\alpha_3x_3.\nonumber
\eea
Sometimes it is more convenient to use the following integral instead of $I_2$:
\[
H_{2}=x_1^3+x_2^3+x_3^3+3\alpha_1x_1+3\alpha_2x_2+3\alpha_3x_3= I_1^3-3I_2.
\]
There hold the following Wronskian relations:
\begin{equation}
\label{eq: DC Wronski}
    \dot{x}_ix_j-x_i\dot{x}_j=I_1x_k^2+2(\alpha_i-\alpha_k)x_i+2(\alpha_j-\alpha_k)x_j+3\alpha_kI_1-H_2.
\end{equation}

Excluding $x_j,x_k$ from equations of motion for $x_i$ with the help of integrals of motion, one arrives at
\begin{equation}\label{eq: DC ell curves}
\dot{x}_i^2=x_i^4+6a_2x_i^2+4a_3x_i+a_4,
\end{equation}
with
\begin{eqnarray*}
& a_2=-\frac{1}{3}(I_1^2+\alpha_j+\alpha_k-2\alpha_i),\qquad
a_3=(\alpha_j+\alpha_k-\alpha_i)I_1+I_2, &\\
& a_4=I_1^4-2(\alpha_j+\alpha_k)I_1^2+(\alpha_j-\alpha_k)^2-4I_1I_2. &
\end{eqnarray*}
All three elliptic curves corresponding to (\ref{eq: DC ell curves}) with $i=1,2,3$, have equal Weierstrass invariants (expressed through the parameters $\alpha_i$ and the integrals of motion):
\[
g_2=a_4+3a_2^2,\qquad g_3=a_2a_4-a_2^3-a_3^2,\quad {\rm so\;\;that}\quad a_3^2=-4a_2^3+g_2a_2-g_3.
\]
The coefficients in (\ref{eq: DC ell curves}) can be thus parametrized in terms of the Weierstrass elliptic function with the invariants $g_2,g_3$ as follows: $a_2=-\wp(A_i)$, $a_3=\wp'(A_i)$, so that $a_4=g_2-3\wp^2(A_i)$. One can show that $A_1+A_2+A_3=0$ (modulo the period lattice), so that one can introduce $B_i$, defined up to a common additive shift, through $A_i=B_i-B_{i+1}$. The solution of the dressing chain DC$_3$ is then given as
\[
x_i(t)=\zeta(t-B_{i+1})-\zeta(t-B_i)-\zeta(B_i-B_{i+1}).
\]

The HK discretization of system (\ref{DC}) is:
\beq \label{eq: dDC x}
\left\{ \begin{array}{l}
\widetilde{x}_1-x_1=\epsilon ( \widetilde{x}_3 x_3- \widetilde x_2 x_2 + \alpha_3 -\alpha_2),  \vspace{.2truecm} \\
\widetilde{x}_2-x_2=\epsilon (\widetilde{x}_1 x_1- \widetilde x_3 x_3 + \alpha_1 -\alpha_3), \vspace{.2truecm} \\
\widetilde{x}_3-x_3= \epsilon (\widetilde{x}_2 x_2- \widetilde x_1 x_1 + \alpha_2 -\alpha_1) .
\end{array} \right.
\eeq

The map $f:x\mapsto\widetilde{x}\ $ obtained by solving
(\ref{eq: dDC x}) for $\widetilde{x}$ is given by:
\begin{equation} \nonumber
\widetilde{x} =f(x,\epsilon)=A^{-1}(x,\epsilon)(x+\epsilon c),
\end{equation}
with
$$
A(x,\epsilon)=
\begin{pmatrix} 1 &  \epsilon  x_2 & -\epsilon x_3   \\
-\epsilon  x_1& 1 & \epsilon x_3 \\
\epsilon x_1 & -\epsilon x_2 & 1 \end{pmatrix} ,\qquad
c=(\alpha_3-\alpha_2,\alpha_1-\alpha_3,\alpha_2-\alpha_1 )^{\rm T}.
$$
Explicitly:
\begin{equation}\label{eq: dDC expl}
    \wx_i=\frac{x_i+\epsilon(x_k^2-x_j^2+\alpha_k-\alpha_j)+
                \epsilon^2\big(I_1x_jx_k+(\alpha_k-\alpha_i)x_j+(\alpha_j-\alpha_i)x_k\big)}
                {1+\epsilon^2(x_1x_2+x_2x_3+x_3x_1)}.
\end{equation}
This map will be called dDC$_3$. From Proposition \ref{th: ddress type} there follows immediately:
\begin{proposition}\label{th: dDC inv meas}
The map dDC$_3$ possesses an invariant volume form:
\begin{equation} \nonumber
\det\frac{\partial \widetilde{x}}{\partial x}=
\frac{\phi(\widetilde{x})}{\phi(x)}\quad\Leftrightarrow\quad f^*\omega=\omega,\quad
\omega= \frac{dx_1\wedge dx_2\wedge dx_3}{\phi(x)},
\end{equation}
with $\phi(x)=1+\epsilon^2 (x_1 x_2+ x_2 x_3 + x_3 x_1)$.
\end{proposition}
Concerning integrability of dDC$_3$, we note first of all that $I_1$ is an obvious conserved quantity. The second one is most easily obtained from the following discretization of Wronskian relations (\ref{eq: DC Wronski}).
\begin{proposition}
For the map dDC$_3$, the following relations hold:
\begin{equation}\label{eq: dDC Wronski}
\frac{\wx_ix_j-x_i\wx_j}{\epsilon}= I_1x_k\wx_k+(\alpha_i-\alpha_k)(x_i+\wx_i)+(\alpha_j-\alpha_k)(x_j+\wx_j)+3\alpha_kI_1-H_2(\epsilon),
\end{equation}
where $H_2(\epsilon)$ is a conserved quantity, given by
\begin{equation}\label{eq: dDC H2}
    H_2(\epsilon)=\frac{H_2+\epsilon^2 G_2}{1+\epsilon^2(x_1x_2+x_2x_3+x_3x_1)},
\end{equation}
where
\begin{eqnarray}
   G_2 & = & I_1^2x_1x_2x_3+(\alpha_1x_1+\alpha_2x_2+\alpha_3x_3)(x_1x_2+x_2x_3+x_3x_1)\nonumber\\
       &   & +2I_1(\alpha_1x_2x_3+\alpha_2x_3x_1+\alpha_3x_1x_1)\nonumber\\
       &   & -(\alpha_2-\alpha_3)^2x_1-(\alpha_3-\alpha_1)^2x_2-(\alpha_1-\alpha_2)^2x_3.\nn
\end{eqnarray}
\end{proposition}
\begin{proof}
Define $H_2(\epsilon)$  by eq. (\ref{eq: dDC Wronski}). It is easily computed with explicit formulas (\ref{eq: dDC expl}). The result given by (\ref{eq: dDC H2}) is a manifestly even function of $\epsilon$ and therefore an integral of motion.
\end{proof}

\begin{prop}
${}$

{\rm (a)}
The  set $\Phi_{ij}=(x_i^3,\, x_j^3,\, x_ix_j(x_i+x_j),\, x_i^2,\, x_ix_j,\, x_j^2,\, x_i,\, x_j,\, 1)$
is a HK basis for the map dDC$_3$  with $\dim K_{\Phi_{ij}}(x)=1$. In other words, the pairs $x_i,x_j$ satisfy equations of degree 3,
\[
P_{ij}(x_i,x_j)=p_0x_i^3+p_1x_j^3+p_2x_ix_j(x_i+x_j)+p_3x_i^2+p_4x_ix_j+p_5x_j^2+p_6x_i+p_7x_j+p_8=0,
\]
whose coefficients $p_m=p_m^{(ij)}$ are constant (expressed through parameters $\alpha_k$ and integrals of motion).
\smallskip

{\rm (b)}
The  set $\Psi_{i}=(x_i^m\wx_i^n)_{m,n=0}^3$ is a HK basis for the map dDC$_3$  with $\dim K_{\Psi_i}(x)=1$. In other words, the pairs $x_i,\wx_i$ lie on bicubic curves of genus 1:
\[
Q_i(x_i,\wx_i)=\sum_{m,n=0}^3q_{mn}x_i^m\wx_i^n=0,
\]
whose coefficients $q_{mn}=q_{mn}^{(i)}$ are constant (expressed through parameters $\alpha_k$ and integrals of motion). Moreover, $q_{13}=q_{31}$.
\end{prop}
\begin{proof}
Statement (a) follows by eliminating $x_k$ from (\ref{eq: dDC H2}) via $x_k=I_1-x_i-x_j$. Statement (b) is obtained with the help of MAPLE. One can also show that these bicubic curves are of genus 1, so that $x_i$ as functions of $t$ are elliptic functions of degree 3 (i.e., with three poles within one parallelogram of periods).
\end{proof}

\section{Coupled Euler tops}
\label{Sect: d2ET}

In \cite{GMN} a remarkable mechanical system was introduced, which can be interpreted as a chain of coupled
three-dimensional Euler tops. The differential equations governing the system are given by:
\beq \label{mah}
\left\{ \begin{array}{l}
\dot x_1= \al_1 x_2 x_3,  \vspace{.1truecm}\\
\dot x_{2j} = \al_{3j-1} x_{2j-1} x_{2j+1},  \vspace{.1truecm}\\
\dot x_{2j+1 } = \al_{3j} x_{2j} x_{2j-1}  + \al_{3j+1} x_{2j+2} x_{2j+3}, \vspace{.1truecm}\\
\dot x_{2j+2} = \al_{3j+2} x_{2j+1} x_{2j+3}, \vspace{.1truecm}\\
\dot x_{2N+1} = \al_{3N} x_{2N} x_{2N-1}, 
  \end{array} \right.
\eeq
with real parameters $\al_i$. Each triple of variables $(x_{2j-1},x_{2j},x_{2j+1})$ can be considered as a 3D Euler top, coupled with the neighboring triple $(x_{2j+1},x_{2j+2},x_{2j+3})$ via the variable $x_{2j+1}$. We will denote system (\ref{mah}) by CET$_N$. It has $N+1$  independent conserved quantities:
\bea
&&H_1 =\alpha_2 x_1^2 - \alpha_1x_2^2, \nonumber \\
&&H_{j} =\al_{3j-3}\al_{3j-1}x^2_{2j-2} -  \al_{3j-4}\al_{3j-1} x^2_{2j-1} + \al_{3j-4}\al_{3j-2} x^2_{2j}, \qquad 2 \leq j \leq N, \quad \nn\\
&&H_{N+1} =  \al_{3N} x^2_{2N}-  \al_{3N-1} x^2_{2N+1}. \nonumber
\eea
Nothing is known about the possible Hamiltonian formulation of this system, and therefore about its integrability in the Liouville-Arnold sense.

For $N=1$ system (\ref{mah}) reduces to the usual Euler top (\ref{eq: ET x}). We will consider in detail the HK discretization of the system CET$_2$ given by
\beq \label{01}
\left\{ \begin{array}{l}
 \dot x_1 = \alpha_1  x_2 x_3, \vspace{.1truecm}\\
 \dot x_2 = \alpha_2  x_3 x_1,  \vspace{.1truecm}\\
 \dot x_3 = \alpha_3  x_1 x_2 +\alpha_4 x_4 x_5 , \vspace{.1truecm}\\
 \dot  x_4= \alpha_5x_5 x_3, \vspace{.1truecm} \\
 \dot  x_5 = \alpha_6  x_3 x_4.
  \end{array} \right.
\eeq
It can be interpreted as two Euler tops, described by the two sets of variables $(x_1,x_2,x_3)$ and $(x_3,x_4,x_5)$, respectively, coupled via the variable $x_3$. It has three independent integrals of motion:
\bea
 &H_1 = \alpha_2 x_1^2 - \alpha_1x_2^2, \qquad H_3 =\alpha_6 x_4^2- \alpha_5 x_5^2, & \nonumber\\
 &H_2 =\alpha_3  \alpha_5 x_2^2 - \alpha_2 \alpha_5 x_3^2 +
 \alpha_2  \alpha_4x_4^2, &\nonumber
\eea
and it can  be solved in terms of elliptic functions. We will be mainly interested in its particular case which is superintegrable.

\begin{prop}
If the  coefficients $\al_i$ satisfy the following condition,
\beq
\al_1 \al_2 = \al_5 \al_6, \label{con}
\eeq
then the system CET$_2$ is superintegrable: it has two additional integrals,
\beq \nn
H_4 = \al_5 x_2 x_5 - \al_2 x_1 x_4, \qquad H_5=  \al_5 x_1 x_5 - \al_1 x_2 x_4,
\eeq
and among the functions $H_1,\ldots,H_5$ there are four independent ones.
\end{prop}
In this case, the variable $x_3$ satisfies the following differential equation:
\begin{equation}\label{eq: CET x3}
\dot{x}_3^2=\left(x_3^2+\frac{H_2}{\alpha_2\alpha_5}\right)\left(\alpha_1\alpha_2x_3^2+
\frac{\alpha_1}{\alpha_5}H_2+\alpha_3H_1-\alpha_4H_3\right)-\frac{\alpha_3\alpha_4}{\alpha_2\alpha_5}H_4^2,
\end{equation}
so that its time evolution is described by an elliptic function of degree 2.

The HK discretization of the system CET$_2$ reads:
\beq \label{02}
\left\{ \begin{array}{l}
   \wx_1-x_1= \epsilon \al_1(\wx_2 x_3+x_2 \wx_3 ), \vspace{.1truecm}\\
   \wx_2-x_2= \epsilon \al_2(\wx_3 x_1+x_3 \wx_1 ),  \vspace{.1truecm}\\
   \wx_3-x_3= \epsilon \al_3(\wx_1 x_2+x_1 \wx_2)+\epsilon\al_4 (\wx_4 x_5+x_4 \wx_5 ), \vspace{.1truecm}\\
   \wx_4-x_4= \epsilon \al_5(\wx_5 x_3+x_5 \wx_3 ), \vspace{.1truecm} \\
   \wx_5-x_5= \epsilon \al_6(\wx_3 x_4+x_3 \wx_4).
  \end{array} \right.
\eeq
The map $f:x\mapsto\widetilde{x}\ $ obtained by solving (\ref{02}) for $\widetilde{x}$ is given by:
\begin{equation} \nonumber
\widetilde{x} =f(x,\epsilon)=A^{-1}(x,\epsilon)x,
\end{equation}
with
$$
A(x,\epsilon)= \left( \begin {array}{ccccc} 1&\epsilon \alpha_{{1}}
 x_{{3}}&\epsilon  \alpha_{{1}} x_{{2}}&0&0
\\\noalign{\medskip}\epsilon  \alpha_{{2}} x_{{3}}
&1&\epsilon \alpha_{{2}}x_{{1}}&0&0
\\\noalign{\medskip}\epsilon\alpha_{{3}}x_{{2}}
&\epsilon\alpha_{{3}} x_{{1}}&1&\epsilon
 \alpha_{{4}} x_{{5}}&\epsilon \alpha_{{4}} x_{{4}}\\\noalign{\medskip}0&0&
 \epsilon\alpha_{{5}}
 x_{{5}}&1&\epsilon \alpha_{{5}}
x_{{3}}\\\noalign{\medskip}0&0&\epsilon \alpha_{{6}} x_{{4}}&\epsilon \alpha_{{6}}x_{{3}}&1
\end {array} \right).
$$
This map will be called dCET$_2$ in the sequel.
\begin{prop}
The functions
$$
H_1 (\epsilon) = \frac{\al_2 x_1^2 - \al_1 x_2^2}{1 - \epsilon^2  \al_1 \al_2 x_3^2}, \qquad
H_3 (\epsilon)=\frac{\al_6 x_4^2- \al_5 x_5^2}{1 - \epsilon^2 \al_5 \al_6 x_3^2} ,
$$
are conserved quantities of the map dCET$_2$.
\end{prop}

Proposition \ref{ll} gives only two  independent integrals of motion for the map dCET$_2$.
Numerical experiments indicate that the third integral does not exist in general. The situation is different  under condition (\ref{con}). Note that in this case the denominators of the integrals $H_1(\ep)$ and $H_3(\ep)$ coincide.

\begin{prop} \label{ll}
If condition (\ref{con}) holds, then the map dCET$_2$ has in addition to $H_1(\ep)$ and $H_3(\ep)$
also the following conserved quantities
$$
 H_2(\epsilon) = \frac{ \alpha_3  \alpha_5 x_2^2 - \alpha_2 \alpha_5 x_3^2 +
 \alpha_2  \alpha_4x_4^2}{1-   \epsilon^2 \al_1 \al_2 x_3^2},
 $$
 $$
 H_4(\epsilon)=  \frac{\al_5 x_2 x_5 - \al_2 x_1 x_4}{1 - \epsilon^2  \al_1 \al_2 x_3^2}, \qquad
 H_5(\epsilon)=  \frac{\al_5 x_1 x_5 - \al_1 x_2 x_4}{1 - \epsilon^2  \al_1 \al_2 x_3^2} .
 $$
 There are four independent functions among $H_1(\ep),\ldots, H_5(\ep)$.
\end{prop}

 We now present HK bases for the map dCET$_2$.

\begin{prop} \label{thm: interactingtopsfullbasis}
Under condition (\ref{con}), the map dCET$_2$ has the following HK bases.

{\rm (a)} The set $\Phi = (x_1^2,\, x_2^2,\, x_3^2,\, x_4^2,\, x_5^2,\, 1)$ is a HK basis with $\dim K_{\Phi}(x) = 3$.
\smallskip

{\rm (b)} The sets $\Phi_1=(x_1^2,\, x_2^2,\, x_3^2,\, 1)$ and $\Phi_2=(x_3^2,\, x_4^2,\, x_5^2,\, 1)$ are HK bases with
one-dimensional null-spaces. At each point $x\in\bbR^5$ we have: $K_{\Phi_1}(x) = [e_1:e_2:{\epsilon}^2 \alpha_{1 }\alpha_{2 } :-1]$ and $K_{\Phi_2}(x) = [{\epsilon}^2 \alpha_{1 }\alpha_{2 }: f_4:f_5 :-1]$.
The functions $e_i$ and $f_i$ are conserved quantities given by
\begin{eqnarray*}
e_1={\frac {{\alpha_{2}} (1- \ep^2 \al_1 \al_2 x_3^2 ) }{{\alpha_{2 }} x_1^2 -\al_1 x_2^2 }},\qquad
e_2=-{\frac {\alpha_{1} (1- \ep^2 \al_1 \al_2 x_3^2 ) }{{\alpha_{2 }} x_1^2 -\al_1 x_2^2}},
\end{eqnarray*}
and
\begin{eqnarray*}
f_4=\frac {{\alpha_{5}} (1- \ep^2 \al_1 \al_2 x_3^2 ) }{\al_5 x_4^2 -\alpha_4 x_5^2},\qquad
f_5=-\frac {\alpha_{4} (1- \ep^2 \al_1 \al_2 x_3^2 ) }{\al_5 x_4^2 -\alpha_4 x_5^2}.
\end{eqnarray*}
The set $\Phi_3=(x_1^2,\, x_2^2,\, x_3^2,\, x_4^2)$ is a HK basis  with a
one-dimensional null-space. At each point $x\in\bbR^5$ we have: $K_{\Phi_3}(x) = [g_1:g_2:\al_5:-\alpha_4]$.
The functions $g_i$ are conserved quantities given by
\begin{eqnarray*}
g_1= {\frac {\alpha_{3} \al_5 x_2^2 - \al_2 \al_5 x_3^2 +\al_2 \al_4 x_4^2 }{\alpha_{2} x_1^2 -\al_1 x_2^2}} ,\qquad
g_2=- {\frac {\alpha_{3}\alpha_{5 }x_1^2-\al_1 \al_5 x_3^2 +\al_1 \al_4 x_4^2  }{\alpha_{2 } x_1^2 -\al_1 x_2^2  }}.
\end{eqnarray*}
Similar claim hold for the sets $(x_1^2,\, x_2^2,\, x_3^2,\, x_5^2)$, $(x_1^2,\, x_3^2,\, x_4^2,\, x_5^2)$, and  $(x_2^2,\, x_3^2,\, x_4^2,\, x_5^2)$.
\smallskip

{\rm (c)} The set $\Psi = (x_1,\, x_2,\, x_3,\, x_4,\, x_5)$ is a HK basis with $\dim K_{\Psi}(x) = 2$.
\smallskip

{\rm (d)} The sets $\Psi_1=(x_1,\, x_2,\, x_4)$, $\Psi_2=(x_1,\, x_2,\, x_5)$ are HK bases with
one-dimensional null-spaces. At each point $x\in\bbR^5$ we have: $K_{\Psi_1}(x) = [c_1:c_2:-1]$, $K_{\Psi_2}(x) = [d_1:d_2:-1]$.
The functions $c_1,c_2$ and $d_1,d_2$ are conserved quantities given by
$$
c_1=\frac{\al_2 x_1 x_4-\al_5 x_2 x_5}{\alpha_2x_1^2 -\al_1 x_2^2 }, \qquad
c_2=\frac{\al_5 x_1 x_5-\al_1 x_2 x_4}{\alpha_2x_1^2 -\al_1 x_2^2 },
$$
while
$$
d_1=\frac{\al_2}{\al_5} c_2,\qquad
d_2=\frac{\al_1}{\al_5} c_1.
$$
A similar claim holds for the sets $(x_2, x_4,x_5)$ and $(x_1, x_4,x_5)$.
\end{prop}

We see that the map dCET$_2$ under condition (\ref{con}) possesses four functionally independent conserved quantities. It might look paradoxical that $\Phi\cup\Psi$ is a HK basis with a 5-dimensional null-space, thus imposing seemingly 5 restrictions on any orbit of the map, which would yield 0-dimensional invariant sets (instead of invariant curves in the continuous time case). The resolution of this paradox is that the five restrictions are functionally dependent on their common set (i.e., along any orbit). In other words, the HK basis $\Phi\cup\Psi$ is not regular. This is the first and the only instance of a non-regular HK basis in this paper.

The map dCET$_2$ possesses, in its (super)-integrable regime, an invariant volume form:
\begin{proposition}\label{th: dIT inv meas}
Under condition (\ref{con}), the map dCET$_2$ preserves the following volume form:
\[
\det\frac{\partial \wx}{\partial x}=\frac{\phi(\wx)}{\phi(x)}\quad\Leftrightarrow\quad
f^*\omega=\omega,\quad \omega=\frac{dx_1\wedge dx_2\wedge dx_3\wedge dx_4\wedge dx_5}{\phi(x)},
\]
with
$\phi(x)= (1-\epsilon^2 \al_1 \al_2 x_3^2)^3$.
\end{proposition}
In this regime, the solutions can be found in terms of elliptic functions, as the following statement shows.
\begin{proposition}\label{th: dITbiquad}
Under condition (\ref{con}), the component $x_3$ of any orbit of the map dCET$_2$ satisfies a relation of the type
\begin{equation} \nonumber
Q(x_3,\widetilde{x}_3)=q_0x_3^2 \widetilde{x}_3^2 +
q_1 x_3 \widetilde{x}_3(x_3+\widetilde{x}_3)
+q_2(x_3^2+\widetilde{x}_3^2) + q_3 x_3\widetilde{x}_3
+q_4(x_3+\widetilde{x}_3) + q_5=0,
\end{equation}
coefficients of the biquadratic polynomial  $Q$  being conserved quantities of dCET$_2$.
\end{proposition}
This statement is a proper discretization of eq. (\ref{eq: CET x3}).

\section{Three wave system}
\label{Sect: d3W}

The three wave interaction system of ordinary differential equations is
\cite{ALMR}:
\beq
\left\{
\begin{array}{l}
\dot z_1  =  \ri  \alpha_1 \bar z_2 \bar z_3, \vspace{.1truecm} \\
\dot z_2  =   \ri \alpha_2 \bar z_3 \bar z_1 , \vspace{.1truecm}   \\
\dot z_3  =   \ri \alpha_3 \bar z_1 \bar z_2.
\end{array}
\right. \label{flow1}
\eeq
Here $z = (z_1,z_2,z_3)\in \mathbb{C}^3$, while the parameters $\alpha_i$ of the system are supposed to be real numbers. If $(i,j,k)$ stands for any cyclic permutation of (123), then we can write system (\ref{flow1}) in the abbreviated form
\beq \label{flow1per}
\dot z_i  =   \ri \alpha_i \bar z_j \bar z_k,
\eeq
Writing $z_i = x_i + \ri y_i$, $i=1,2,3$, we put system (\ref{flow1per}) into the form
\beq  \label{flow1realper}
\left\{
\begin{array}{l}
\dot x_i  = \alpha_i (x_j y_k + y_j x_k), \vspace{.1truecm} \\
\dot y_i  =  \alpha_i (x_j x_k - y_j y_k) . \vspace{.1truecm}
\end{array}
\right.
\eeq
System (\ref{flow1realper}) is completely integrable and can be integrated in terms of elliptic functions. It has three independent integrals of motion: quadratic ones,
\beq
H_{i} = \alpha_j | z_k |^2 - \alpha_k | z_j |^2, \nn
\eeq
among which there are only two independent ones because of $\alpha_1 H_{1} + \alpha_2 H_{2} + \alpha_3 H_{3}=0$, and a cubic one,
\beq
K = \frac{1}{2} \left( z_1 z_2 z_3 + \bar z_1 \bar z_2 \bar z_3 \right) =
\mathfrak{Re} (z_1 z_2 z_3). \nn
\eeq

The HK discretization of system (\ref{flow1per}) reads
\beq
\tilde z_i - z_i =   \ri\ep\alpha_i(\bar z_j  \tilde{\bar{z}}_k +
\tilde{\bar{z}}_j \bar z_k),\nn
\eeq
or, in the real variables $(x_i,y_i)$,
\beq  \nn
\left\{
\begin{array}{l}
\wx_i-x_i=\ep\alpha_i(x_j\wy_k+\wx_j y_k + y_j\wx_k + \wy_j x_k), \vspace{.1truecm} \\
\wy_i-y_i=\ep\alpha_i(x_j\wx_k+\wx_j x_k - y_j\wy_k - \wy_j y_k). \vspace{.1truecm}
\end{array}
\right.
\eeq
In the matrix form, this can be put as
$$
A(x,y, \ep) \begin{pmatrix} \wx \\ \wy\end{pmatrix} = \begin{pmatrix} x \\ y \end{pmatrix}\quad
\Leftrightarrow\quad \begin{pmatrix} \wx \\ \wy\end{pmatrix}=f (x,y , \ep) = A^{-1}(x,y , \ep)
\begin{pmatrix} x \\ y \end{pmatrix},
$$
where
$$
A(x ,y, \ep)= \begin{pmatrix}
1 & -\ep\alpha_1 y_3 & -\ep\alpha_1 y_2 & 0 & -\ep\alpha_1 x_3 & - \ep\alpha_1 x_2\\
-\ep\alpha_2 y_3 & 1 & -\ep\alpha_2 y_1 & -\ep\alpha_2 x_3 & 0   & -\ep\alpha_2 x_1\\
-\ep\alpha_3 y_2 & -\ep\alpha_3 y_1 & 1 & -\ep\alpha_3 x_2 & -\ep\alpha_3 x_1 & 0\\
0 & -\ep \alpha_1 x_3 & -\ep\alpha_1 x_2 & 1 & \ep\alpha_1 y_3 &  \ep\alpha_1 y_2\\
-\ep\alpha_2 x_3  & 0 & -\ep\alpha_2 x_1 & \ep\alpha_2 y_3  & 1  & \ep\alpha_2 y_1\\
-\ep\alpha_3 x_2  & -\ep\alpha_3 x_1 & 0 & \ep\alpha_3 y_2  & \ep\alpha_3 y_1  & 1
\end{pmatrix}.
$$
The birational map $f: \mathbb{R}^6 \rightarrow \mathbb{R}^6$ will be called d3W hereafter.

\begin{prop} \label{th1}  The map d3W has three independent conserved quantities, namely, any two of
\beq
H_{i} (\ep)=\frac{\alpha_j | z_k |^2 - \alpha_k | z_j |^2}{1 - \ep^2
\alpha_j \alpha_k  | z_i |^2 },\nn
\eeq
supplied with any one of
\beq
K_i(\ep) = \frac{\mathfrak{Re}(z_1z_2z_3) (1- \ep^2 \alpha_k \alpha_i
|z_j|^2)(1- \ep^2 \alpha_i \alpha_j |z_k|^2)}
{\Delta(z, \bar z,\ep)} ,  \nn
\eeq
where
\bea
\Delta (z, \bar z, \ep)& =& \det A(x,y,\ep)\ =\ 1 - 2 \ep^2 (\alpha_2 \alpha_3 | z_1 |^2 +
\alpha_3 \alpha_1 | z_2 |^2 + \alpha_1 \alpha_2 | z_3 |^2)   \nonumber\\
&& +\,\ep^4 (\alpha_2 \alpha_3 | z_1 |^2 + \alpha_3 \alpha_1 | z_2 |^2 +
\alpha_1 \alpha_2 | z_3 |^2)^2 -4 \ep^6 \alpha_1^2 \alpha_2^2 \alpha_3^2
| z_1 |^2  | z_2 |^2  | z_3 |^2 . \nn
\eea
\end{prop}
\begin{prop}
The map d3W possesses an invariant volume form:
$$
\det \frac{\partial \tilde z}{\partial z} = \frac{\phi(\tilde z)}{\phi(z)}\quad\Leftrightarrow\quad
f^*\omega=\omega,\quad \omega=\frac{dx_1\wedge dx_2\wedge dx_3\wedge dy_1\wedge dy_2\wedge dy_3}{\phi(z)},
$$
where
$\phi(z) = \Delta(z, \bar z,\ep)$.
\end{prop}
Next, we give the results on the HK bases for the map d3W, which yield a complete set of integrals of motion.
\begin{prop}
${}$

{\rm (a)}
The sets
$
\Phi_i=(|z_j|^2,\, |z_k|^2,\, 1), \, i=1,2,3,
$
are HK bases for the map d3W with $\dim K_{\Phi_i}(z)=1$. At each point $z\in \mathbb{C}^3$ there holds:
$
K_{\Phi_i}(z)=[d_1 : d_2 : -1],
$
where the coefficients
\begin{equation}\nn
d_1(z)=\frac{\alpha_k (1-  \ep^2 \alpha_i \alpha_j |z_k|^2)}
{\alpha_k |z_j|^2 - \alpha_j |z_k|^2}\,, \qquad
d_2(z)=- \frac{\alpha_j (1- \ep^2\alpha_k \alpha_i |z_j|^2)}
{\alpha_k |z_j|^2 - \alpha_j |z_k|^2}\,,
\end{equation}
are integrals of motion of the map d3W. They are functionally dependent because of
$\alpha_j d_1(z) + \alpha_k d_2(z) =  \ep^2 \alpha_1 \alpha_2 \alpha_3$.
\smallskip

{\rm (b)}
The sets
$
\Psi_i=( \mathfrak{Re}(z_1 z_2 z_3 ),\, |z_i|^2,\, 1), \, i=1,2,3,
$
are HK bases for the map d3W with $\dim K_{\Psi_i}(z)=1$. At each point $z\in \mathbb{C}^3$ there holds:
$
K_{\Psi_i}(z)=[e_1 : e_2 : -1],
$
where the coefficients
\bea
&& e_1(z)=- \frac{\Delta(z,\bar z,\ep)}
{\mathfrak{Re}(  z_1 z_2  z_3) \big(1  - \ep^2 \left(
-\alpha_j \alpha_k | z_i |^2 + \alpha_k \alpha_i | z_j |^2 + \alpha_i
\alpha_j | z_k |^2\right) \big)^2}\,, \nonumber \\
&& e_2(z)=  \frac{ 4 \alpha_j \alpha_k \ep^2 (1- \ep^2 \alpha_k \alpha_i
|z_j|^2)(1- \ep^2 \alpha_i \alpha_j |z_k|^2)}
{\big(1  - \ep^2 \left(
-\alpha_j \alpha_k | z_i |^2 + \alpha_k \alpha_i | z_j |^2 + \alpha_i
\alpha_j | z_k |^2\right)\big)^2}, \nonumber
\eea
are independent integrals of motion of the map d3W.
\end{prop}

\section{Lagrange top}
\label{Sect: dL}

Lagrange top was the second integrable system, after Euler top, to which the HK discretization was successfully applied \cite{KH}. We reproduce and re-derive here the results of that paper, and add some new results.

Equations of motion of the Lagrange top are of the general Kirchhoff type:
\begin{equation}
\label{eq: Kirch}
\left\{\begin{array}{l}
\dot{m}=m\times \nabla_m H 
        +p\times \nabla_p H , \vspace{.1truecm}\\
\dot{p}=p\times \nabla_m H ,
\end{array}\right.
\end{equation}
where $m=(m_1,m_2,m_3)^{\rm T}$ and $p=(p_1,p_2,p_3)^{\rm T}$. Any Kirchhoff type system is Hamiltonian
with the Hamilton function $H=H(m,p)$ with respect to the Lie-Poisson bracket on $\mathfrak{e}(3)^*$,
\begin{equation}\nn
\{m_i,m_j\}=\epsilon_{ijk} m_k,\qquad \{m_i,p_j\} = \epsilon_{ijk} p_k, \qquad \{p_i,p_j\} = 0,
\end{equation}
and admits the Hamilton function $H$ and the Casimir functions
\begin{equation}
\label{eq: e3 Cas}
C_1=p_1^2+p_2^2+p_3^2,  \qquad
C_2=m_1p_1+m_2p_2+m_3p_3,
\end{equation}
as integrals of motion. For the complete integrability of a Kirchhoff type system, it should admit a fourth independent integral of motion.

The Hamilton function of the Lagrange top (LT) is $H=H_1/2$, where
\beq \nn
H_1=m_1^2+m_2^2+\alpha m_3^2+2\gamma p_3.
\eeq
Thus, equations of motion of LT read
\beq  \label{eq: lagrange}
\left\{ \begin{array} {l}
\dot{m}_1= (\alpha-1)m_2m_3 + \gamma p_2, \vspace{.1truecm} \\
\dot{m}_2= (1-\alpha)m_1m_3 - \gamma p_1,  \vspace{.1truecm}  \\
\dot{m}_3 = 0,  \vspace{.1truecm}  \\
\dot{p}_1 =\alpha p_2m_3-p_3m_2,  \vspace{.1truecm}  \\
\dot{p}_2 = p_3m_1-\alpha p_1m_3,  \vspace{.1truecm} \\
\dot{p}_3 = p_1m_2-p_2m_1.
\end{array} \right.
\eeq
It follows immediately that the fourth integral of motion is simply
$
    H_2=m_3.
$
Traditionally, the explicit integration of the LT in terms of elliptic functions starts with the following observation: the component $p_3$ of the solution satisfies the differential equation
\begin{equation}
\label{eq: LT dotp3}
\dot{p}_3^2=P_3(p_3),
\end{equation}
with a cubic polynomial $P_3$ whose coefficients are expressed through integrals of motion:
$$
P_3(p_3)=(H_1-\alpha m_3^2-2\gamma p_3)(C_1-p_3^2)-(C_2-m_3p_3)^2.
$$

We mention also the following Wronskian relation which follows easily from equations of motion:
\begin{equation}\label{eq: LT W}
(\dot{m}_1p_1-m_1\dot{p}_1)+(\dot{m}_2p_2-m_2\dot{p}_2)+(2\alpha-1)(\dot{m}_3p_3-m_3\dot{p}_3)=0.
\end{equation}

Applying the HK discretization scheme to eqs. (\ref{eq: lagrange}), we obtain the following discrete system:
\beq  \nn
\label{eq: dlt}
\left\{ \begin{array} {l}
\widetilde{m}_1-m_1=  \epsilon(\alpha-1)(\widetilde{m}_2m_3 + m_2\widetilde{m}_3) + \epsilon\gamma(p_2+\widetilde{p}_2) ,     \vspace{.1truecm} \\
\widetilde{m}_2-m_2 =  \epsilon(1-\alpha)(\widetilde{m}_1m_3 + m_1\widetilde{m}_3) - \epsilon\gamma(p_1+\widetilde{p}_1)  , \vspace{.1truecm} \\
\widetilde{m}_3-m_3 = 0, \vspace{.1truecm}  \\
\widetilde{p}_1-p_1 = \epsilon \alpha (p_2\widetilde{m}_3+\widetilde{p}_2m_3) - \epsilon  ( p_3\widetilde{m}_2+\widetilde{p}_3m_2) ,  \vspace{.1truecm} \\
\widetilde{p}_2-p_2 = \epsilon (p_3\widetilde{m}_1+\widetilde{p}_3m_1) -  \epsilon \alpha (p_1\widetilde{m}_3+\widetilde{p}_1m_3) ,  \vspace{.1truecm}\\
\widetilde{p}_3-p_3 = \epsilon (p_1\widetilde{m}_2+\widetilde{p}_1m_2-p_2\widetilde{m}_1-\widetilde{p}_2m_1).
\end{array} \right.
\eeq
As usual, this can be solved for $(\widetilde{m},\widetilde{p})$, thus yielding the reversible and birational  map $x\mapsto\wx=f(x,\epsilon)=A^{-1}(x,\epsilon)(\mathds{1}+\epsilon B)x$, where $x=(m_1,m_2,m_3,p_1,p_2,p_3)^{\rm T}$, and
\[
A(x,\epsilon) =
\begin{pmatrix}
1 & \epsilon(1-\alpha)m_3 & \epsilon(1-\alpha)m_2 & 0 & 0 & 0 \\
-\epsilon(1-\alpha)m_3 & 1 & -\epsilon(1-\alpha)m_1 & 0 & 0 &  0 \\
0 & 0 & 1 & 0 & 0 & 0 \\
0 & \epsilon p_3 & -\epsilon\alpha p_2 & 1 & -\epsilon\alpha m_3 & \epsilon m_2 \\
-\epsilon p_3 & 0 & \epsilon\alpha p_1 & \epsilon\alpha m_3 & 1 & -\epsilon m_1 \\
\epsilon p_2 & -\epsilon p_1 & 0 & -\epsilon m_2 & \epsilon m_1 & 1 \\
\end{pmatrix}-\epsilon B,
\]
\[
B=\begin{pmatrix}
0 & 0 & 0 & 0 & \gamma & 0 \\
0 & 0 & 0 & -\gamma & 0 & 0 \\
0 & 0 & 0 & 0 & 0 & 0 \\
0 & 0 & 0 & 0 & 0 & 0 \\
0 & 0 & 0 & 0 & 0 & 0 \\
0 & 0 & 0 & 0 & 0 & 0
\end{pmatrix}.
\]
This map will be called dLT in the sequel. Obviously, $m_3$ serves as a conserved quantity for dLT. The remaining three conserved quantities can be found with the help of the HK bases approach. A simple conserved quantity can be found from the following statement which serves as a natural discretization of the Wronskian relation (\ref{eq: LT W}).
\begin{prop} \label{thm: lagrange simple basis}
The set $\Gamma=(\widetilde{m}_1p_1-m_1\widetilde{p}_1,\, \widetilde{m}_2p_2-m_2\widetilde{p}_2,\, \widetilde{m}_3p_3-m_3\widetilde{p}_3)$ is a HK basis for the map dLT with $\dim K_{\Gamma}(x) = 1$.  At each point $x\in\bbR^6$ we have: $K_{\Gamma}(x) = [1:1:b_3]$, where $b_3$ is a conserved quantity of dLT given by
\begin{equation}\label{eq: dLT W simple}
b_3=\frac{(2\alpha-1)m_3+\epsilon^2(\alpha-1)m_3(m_1^2+m_2^2)+\epsilon^2\gamma(m_1p_1+m_2p_2)}
{m_3\Delta_1},
\end{equation}
where
\begin{equation}\label{eq: dLT Delta1}
\Delta_1=1+\epsilon^2\alpha(1-\alpha)m_3^2-\epsilon^2\gamma p_3.
\end{equation}
\end{prop}
\begin{proof}
A straightforward computation with MAPLE of the quantity
\[
b_3=-\frac{(\widetilde{m}_1p_1-m_1\widetilde{p}_1)+(\widetilde{m}_2p_2-m_2\widetilde{p}_2)}
           {(\widetilde{m}_3p_3-m_3\widetilde{p}_3)}
\]
leads to the value (\ref{eq: dLT W simple}). It is an even function of $\epsilon$ and therefore a conserved quantity.
\end{proof}
Further integrals of motion were found by Hirota and Kimura. We reproduce here their results with new simplified proofs.
\begin{prop} \label{thm: lagrangefullbasis}{\rm \cite{KH}}

${\!\!}$
{\rm (a)} The set $\Phi = (m_1^2+m_2^2,\, p_1m_1+p_2m_2,\, p_1^2+p_2^2,\, p_3^2,\, p_3,\, 1)$ is a HK basis for the map dLT with $\dim K_{\Phi}(x) = 3$.
\smallskip

{\rm (b)} The set $\Phi_1=(1,\, p_3,\, p_3^2,\, m_1^2+m_2^2)$ is a HK basis for the map dLT with a
one-dimensional null-space. At each point $x\in\bbR^6$ we have: $K_{\Phi_1}(x) = [c_0:c_1:c_2:-1]$.
The functions $c_0,c_1,c_2$ are conserved quantities of the map dLT, given by
\begin{eqnarray*}
&&c_0 = \frac{m_1^2+m_2^2+2\gamma p_3+\epsilon^2 c_0^{(4)}+\epsilon^4
c_0^{(6)}+\epsilon^6 c_0^{(8)}+\epsilon^8 c_0^{(10)}}{\Delta_1\Delta_2},\\
&&c_1 = -\frac{2\gamma\big(1-\epsilon^2\alpha(1-\alpha)m_3^2\big)\big(1+\epsilon^2
c_2^{(2)}+\epsilon^4 c_2^{(4)}+\epsilon^6 c_2^{(6)}\big)}{\Delta_1\Delta_2},\\
&&c_2 =-\frac{\epsilon^2\gamma^2\big(1+\epsilon^2 c_2^{(2)}+
\epsilon^4 c_2^{(4)}+\epsilon^6 c_2^{(6)}\big)}{\Delta_1\Delta_2}. \\
\end{eqnarray*}
Here $\Delta_1$ is given in (\ref{eq: dLT Delta1}), and
$
\Delta_2=1+\epsilon^2\Delta_2^{(2)}+\epsilon^4\Delta_2^{(4)}+\epsilon^6\Delta_2^{(6)};
$
coefficients $\Delta^{(q)}$ and $c_k^{(q)}$ are polynomials of degree $q$ in the phase variables. In particular:
\begin{eqnarray*}
&&c_2^{(2)} =m_1^2+m_2^2+(1-2\alpha+2\alpha^2)m_3^2-2\gamma p_3, \\
&&\Delta_2^{(2)} = m_1^2+m_2^2+(1-3\alpha+3\alpha^2)m_3^2-\gamma p_3.
\end{eqnarray*}

{\rm (c)} The set $\Phi_2=(1,\, p_3,\, p_3^2,\, m_1p_1+m_2p_2)$ is a HK basis for the map dLT with a
one-dimensional null-space. At each point $x\in\bbR^6$ we have: $K_{\Phi_2}(x) = [d_0:d_1:d_2:-1]$.
The functions $d_0,d_1,d_2$ are conserved quantities of the map dLT, given by
\begin{eqnarray*}
&&d_0 = \frac{m_1p_1+m_2p_2+m_3p_3+\epsilon^2 d_0^{(4)}+\epsilon^4
d_0^{(6)}+\epsilon^6 d_0^{(8)}+\epsilon^8 d_0^{(10)}}{\Delta_1\Delta_2},\\
&&d_1 =-\frac{m_3+\epsilon^2 d_1^{(3)}+\epsilon^4 d_1^{(5)}
+\epsilon^6 d_1^{(7)}+\epsilon^8 d_1^{(9)}}{\Delta_1\Delta_2}, \\
&&d_2 =
-\frac{\epsilon^2\gamma(1-\alpha)m_3\big(1+\epsilon^2c_2^{(2)}+
\epsilon^4c_2^{(4)}+\epsilon^6c_2^{(6)}\big)}{\Delta_1\Delta_2}, \\
\end{eqnarray*}
where $d_k^{(q)}$ are polynomials of degree $q$ in the phase variables. In particular,
\[
d_1^{(3)}=\gamma(m_1p_1+m_2p_2)-\gamma(3-2\alpha)m_3p_3+
\alpha m_3(m_1^2+m_2^2)+(1-3\alpha+3\alpha^2)m_3^3.
\]

{\rm (d)} The set $\Phi_3=(1,\, p_3,\, p_3^2,\, p_1^2+p_2^2)$ is a HK basis for the map dLT with a
one-dimensional null-space. At each point $x\in\bbR^6$ we have: $K_{\Phi_3}(x) = [e_0:e_1:e_2:-1]$.
The functions $e_0,e_1,e_2$ are conserved quantities of the map dLT, given by
\begin{eqnarray*}
 &&e_0 =
\frac{p_1^2+p_2^2+p_3^2+\epsilon^2 e_0^{(4)}+\epsilon^4
e_0^{(6)}+\epsilon^6 e_0^{(8)}+\epsilon^8 e_0^{(10)}}{\Delta_1\Delta_2},\\
&&e_1= -\frac{2\epsilon^2\big(e_1^{(3)}+\epsilon^2 e_1^{(5)}
+\epsilon^4 e_1^{(7)}+\epsilon^6 e_1^{(9)}\big)}{\Delta_1\Delta_2}, \\
&&e_2 =
-\frac{\big(1+\epsilon^2(1-\alpha)^2m_3^2\big)\big(1+\epsilon^2c_2^{(2)}+
\epsilon^4 c_2^{(4)}+\epsilon^6 c_2^{(6)}\big)}{\Delta_1\Delta_2},
\end{eqnarray*}
where $e_k^{(q)}$ are polynomials of degree $q$ in the phase variables. In particular,
\[
e_1^{(3)}=\gamma(p_1^2+p_2^2+p_3^2)-(1-\alpha)m_3(m_1p_1+m_2p_2+m_3p_3).
\]
\end{prop}
\begin{proof}
(b) We consider a linear system of equations
\begin{equation}\label{eq: Gamma1 syst}
 (c_0+c_1 p_3+c_2 p_3^2)\circ f^i(m,p,\epsilon)=(m_1^2+m_2^2)\circ f^i(m,p,\epsilon),
\end{equation}
for all $i\in\bbZ$. Numerically one sees that it admits a unique solution, and one can identify the linear relation
\begin{equation}\label{eq: Gamma1 rel}
\frac{1}{2}\gamma\epsilon^2 c_1=\big(1-\epsilon^2\alpha(1-\alpha)m_3^2\big)c_2.
\end{equation}
The system of three equations for three unknowns $c_0,c_1,c_2$ consisting of \eqref{eq: Gamma1 syst} with $i=0,1$ and \eqref{eq: Gamma1 rel} can easily be solved with MAPLE. Its solutions are even functions of $\epsilon$, which proves that they are integrals of motion.

(c) This time we consider the linear system of equations
\begin{equation}\label{eq: Gamma2 syst}
 (d_0+d_1 p_3+d_2 p_3^2)\circ f^i(m,p,\epsilon)=(m_1p_1+m_2p_2)\circ f^i(m,p,\epsilon),
\end{equation}
for all $i\in\bbZ$. Numerically we see that it admits a unique solution, and we can identify the linear relation
\begin{equation}\label{eq: Gamma2 rel}
\gamma d_2=(1-\alpha)m_3c_2.
\end{equation}
The system of three equations for the three unknowns $d_0,d_1,d_2$ consisting of \eqref{eq: Gamma2 syst} for $i=0,1$ and of \eqref{eq: Gamma2 rel} with $c_2$ already found in part b) can easily be solved with MAPLE. Its solutions are even functions of $\epsilon$ and therefore are integrals.

(d) Completely analogous to the last two proofs: we solve the linear system of three equations for the three unknowns $e_0,e_1,e_2$, consisting of the equations
\begin{equation}
(e_0+e_1 p_3+e_2 p_3^2)\circ f^i(m,p,\epsilon)=(p_1^2+p_2^2)\circ f^i(m,p,\epsilon), \nn
\end{equation}
for $i=0,1$, and of the linear relation
\begin{equation}
\epsilon^2 \gamma^2 e_2=\big(1+\epsilon^2(1-\alpha)^2 m_3^2\big)c_2,\nn
\end{equation}
and verify that they are even functions of $\epsilon$.
\end{proof}

We note that for $\alpha=1$ the integrals $d_0$, $d_1$, $d_2$ simplify to
\begin{equation}
d_0=\frac{m_1p_1+m_2p_2+m_3p_3}{1-\epsilon^2\gamma p_3},\quad
d_1=-\frac{m_3+\epsilon^2\gamma(m_1p_1+m_2p_2)}{1-\epsilon^2\gamma p_3}, \quad d_2=0.\nn
\end{equation}

It is possible to find a further simple, in fact polynomial, integral for the map dLT.
\begin{proposition}{\rm \cite{KH}}
The function
\[
F=m_1^2+m_2^2+2\gamma p_3-\epsilon^2\big((1-\alpha)m_3m_1+\gamma p_1\big)^2
-\epsilon^2\big((1-\alpha)m_3m_2+\gamma p_2\big)^2,
\]
is a conserved quantity for the map dLT.
\end{proposition}
\begin{proof} Setting
$
C=1-\epsilon^2(1-\alpha)^2m_3^2, \, D=-2\epsilon^2\gamma(1-\alpha)m_3, \, E=-\epsilon^2\gamma^2,
$
one can check that $Cc_1+Dd_1+Ee_1=0$ and $Cc_2+Dd_2+Ee_2=-2\gamma$. This yields for the conserved quantity $F=Cc_0+Dd_0+Ee_0$ the expression given in the Proposition.
\end{proof}

Considering the leading terms of the power expansions in $\epsilon$, one sees immediately that the integrals $c_0$, $d_0$, $e_0$, and $m_3$ are functionally independent. Using exact evaluation of gradients we can also verify independence of other sets of integrals. It turns out that for $\alpha\neq 1$ each one of the quadruples $\{d_0,d_1,d_2,m_3\}$ and $\{e_0,e_1,e_2,m_3\}$  consists of independent integrals.

A direct ``bilinearization'' of the HK bases of Proposition \ref{thm: lagrangefullbasis} provides us with an alternative source of integrals of motion:
\begin{prop} \label{thm: lagrange bilinear basis}

The set $$\Psi = (m_1\wm_1+m_2\wm_2,\, p_1\wm_1+\wip_1m_1+p_2\wm_2+\wip_2m_2,\, p_1\wip_1+p_2\wip_2,\, p_3\wip_3,\, p_3+\wip_3,\, 1)$$ is a HK basis for the map dLT with $\dim K_{\Psi}(x) = 3$. Each of the following subsets of $\Psi$,
\begin{eqnarray*}
&& \Psi_1 =(1,\, p_3+\wip_3,\, p_3\wip_3,\, m_1\wm_1+m_2\wm_2),\\
&& \Psi_2 =(1,\, p_3+\wip_3,\, p_3\wip_3,\, m_1\wip_1+\wm_1p_1+m_2\wip_2+\wm_2p_2),\\
&& \Psi_3 = (1,\, p_3+\wip_3,\, p_3\wip_3,\, p_1\wip_1+p_2\wip_2),
\end{eqnarray*}
is a HK basis with a one-dimensional null-space.
\end{prop}
Concerning solutions of dLT as functions of the (discrete) time $t$, the crucial result is given in the following statement which should be considered as the proper discretization of eq. (\ref{eq: LT dotp3}).
\begin{proposition}\label{th: dLT biquad}{\rm \cite{KH}}
The component $p_3$ of the solution of difference equations (\ref{eq: dlt}) satisfies a relation of the type
\begin{equation} \nonumber
Q(p_3,\widetilde{p}_3)=q_0p_3^2 \widetilde{p}_3^2 +
q_1 p_3 \widetilde{p}_3(p_3+\widetilde{p}_3)
+q_2(p_3^2+\widetilde{p}_3^2) + q_3 p_3\widetilde{p}_3
+q_4(p_3+\widetilde{p}_3) + q_5=0,
\end{equation}
coefficients of the biquadratic polynomial  $Q$  being conserved quantities of dLT. Hence, $p_3(t)$ is an elliptic function of degree 2.
\end{proposition}
Although it remains unknown whether the map dLT admits an invariant Poisson structure, we have the following statement.
\begin{proposition}\label{th: dLT inv meas}
The map dLT possesses an invariant volume form:
\[
\det\frac{\partial \wx}{\partial x}=\frac{\phi(\wx)}{\phi(x)}\quad\Leftrightarrow\quad
f^*\omega=\omega,\quad\omega=\frac{dm_1\wedge dm_2\wedge dm_3\wedge dp_1\wedge dp_2\wedge dp_3}{\phi(x)},
\]
with $\phi(x)=\Delta_2(x,\epsilon)$.
\end{proposition}

\section{Kirchhoff case of the rigid body motion in an ideal fluid}
\label{Sect: dK}

The motion of a rigid body in an ideal fluid is described by Kirchhoff equations (\ref{eq: Kirch})
with $H$ being a quadratic form in $m=(m_1,m_2,m_3)^{\rm T}\in\bbR^3$ and $p=(p_1,p_2,p_3)^{\rm T}\in\bbR^3$. The physical meaning of $m$ is the total angular momentum, whereas $p$ represents the total linear momentum of the
system. A detailed introduction to the general context of rigid body dynamics and its mathematical foundations can be found in \cite{MR}.

The integrable case of this system found in the original paper by Kirchhoff \cite{Kirch} and carrying his name is characterized by the Hamilton function $H=H_1/2$, where
\begin{equation}\nn
H_1=a_1(m_1^2+m_2^2)+a_3m_3^2+b_1(p_1^2+p_2^2)+b_3p_3^2.
\end{equation}
The differential equations of the Kirchhoff case are:
\beq \label{eq:Kirchhoff}
\left\{ \begin{array} {l}
\dot{m_1} = (a_3-a_1)m_2 m_3 + (b_3-b_1)p_2p_3, \vspace{.1truecm}\\
\dot{m_2} = (a_1-a_3)m_1 m_3 + (b_1-b_3)p_1p_3, \vspace{.1truecm}\\
\dot{m_3} = 0, \vspace{.1truecm} \\
\dot{p_1} = a_3p_2 m_3-a_1p_3 m_2,\vspace{.1truecm}\\
\dot{p_2} =  a_1p_3 m_1-a_3p_1m_3, \vspace{.1truecm} \\
\dot{p_3} =a_1(p_1m_2 - p_2m_1).
\end{array} \right.
\eeq
Along with the Hamilton function $H$ and the  Casimir functions (\ref{eq: e3 Cas}), it possesses the obvious fourth integral, due to the rotational symmetry of the system:
$
H_2=m_3.
$
Traditionally, the explicit integration of the Kirchhoff case in terms of elliptic functions starts with the following observation: the component $p_3$ of the solution satisfies the differential equation
\begin{equation}\nn
\dot{p}_3^2=P_4(p_3),
\end{equation}
with a quartic polynomial $P_4$ whose coefficients are expressed through integrals of motion:
\[
P_4(p_3)=a_1\big(H_1-a_3m_3^2-b_1(C_1-p_3^2)-b_3p_3^2\big)(C_1-p_3^2)-a_1^2(C_2-m_3p_3)^2.
\]

We mention also the following Wronskian relation which follows easily from equations of motion:
\begin{equation}\label{eq: KC W}
a_1(\dot{m}_1p_1-m_1\dot{p}_1)+a_1(\dot{m}_2p_2-m_2\dot{p}_2)+(2a_3-a_1)(\dot{m}_3p_3-m_3\dot{p}_3)=0.
\end{equation}

Applying the HK approach to (\ref{eq:Kirchhoff}), we obtain the following system of equations:
\beq \nn
\left\{ \begin{array} {l}
\wm_1-m_1 = \epsilon (a_3-a_1)(\wm_2m_3+m_2\wm_3)+\epsilon(b_3-b_1)(\wip_2p_3+p_2\wip_3), \vspace{.1truecm}\\
\wm_2-m_2 = \epsilon (a_1-a_3)(\wm_1m_3+m_1\wm_3)+\epsilon(b_1-b_3)(\wip_1p_3+p_1\wip_3), \vspace{.1truecm}\\
\wm_3-m_3 = 0, \vspace{.1truecm} \\
\wip_1-p_1 = \epsilon a_3(\wip_2m_3+p_2\wm_3)-\epsilon a_1(\wip_3m_2+p_3\wm_2), \vspace{.1truecm}\\
\wip_2-p_2 =\epsilon a_1(\wip_3m_1+p_3\wm_1)-\epsilon a_3(\wip_1m_3+p_1\wm_3), \vspace{.1truecm} \\
\wip_3-p_3 = \epsilon a_1(\wip_1m_2+p_1\wm_2)-\epsilon a_1(\wip_2m_1+p_2\wm_1).
\end{array} \right.
\eeq
As usual, these equations define a birational map $\wx=f(x,\epsilon)$, $x=(m,p)^{\rm T}$. We will refer to this map as dK. Like in the case of dLT, $m_3$ is a conserved quantity of dK. A further ``simple'' conserved quantity can be found from the following natural discretization of the Wronskian relation (\ref{eq: KC W}).
\begin{prop} \label{thm: Kirchhoff simple basis}
The set $\Gamma=(\widetilde{m}_1p_1-m_1\widetilde{p}_1,\, \widetilde{m}_2p_2-m_2\widetilde{p}_2,\, \widetilde{m}_3p_3-m_3\widetilde{p}_3)$ is a HK basis for the map dK with $\dim K_{\Gamma}(x) = 1$.  At each point $x\in\bbR^6$ we have: $K_{\Gamma}(x) = [1:1:-\gamma_3]$, where $\gamma_3$ is a conserved quantity of dK given by
\begin{equation}\label{eq: dKC W simple}
\gamma_3=\frac{\Delta_0}{a_1\Delta_1},
\end{equation}
where
\begin{eqnarray}
&&\Delta_0 =a_1-2a_3+\epsilon^2a_1^2(a_1-a_3)(m_1^2+m_2^2)+\epsilon^2a_1a_3(b_1-b_3)(p_1^2+p_2^2),
\label{eq: dKC Delta0}\\
&&\Delta_1 = 1+\epsilon^2a_3(a_1-a_3)m_3^2+\epsilon^2a_1(b_1-b_3)p_3^2.\label{eq: dKC Delta1}
\end{eqnarray}
\end{prop}
\begin{proof}
Like in the case of dLT, we let MAPLE compute the quantity
\[
\gamma_3=\frac{(\widetilde{m}_1p_1-m_1\widetilde{p}_1)+(\widetilde{m}_2p_2-m_2\widetilde{p}_2)}
           {(\widetilde{m}_3p_3-m_3\widetilde{p}_3)},
\]
which results in (\ref{eq: dKC W simple}), an even function of $\epsilon$ and therefore a conserved quantity.
\end{proof}
Interestingly enough, this same integral may also be obtained from another HK basis:
\begin{prop}
The set $\Phi_0=(m_1^2+m_2^2,p_1^2+p_2^2,p_3^2,1)$ is a HK Basis for the map dK with
$\dim K_{\Phi_0}(x)$ = 1. The linear combination of these functions vanishing along the orbits can be put as
$\Delta_0-\gamma_3a_1\Delta_1=0$.
\end{prop}
\begin{proof}
The statement of the Proposition deals with the solution of a linear system of equations consisting of
\begin{equation}\label{eq: Kirch simple syst}
(c_1(m_1^2+m_2^2)+c_2(p_1^2+p_2^2)+c_3p_3^2)\circ f^i(m,p,\epsilon) =1
\end{equation}
for all $i\in\bbZ$. We solve this system with $i=-1,0,1$ (numerically or symbolically), and observe that the solutions satisfy $a_3(b_1-b_3)c_1=a_1(a_1-a_3)c_2$. Then, we consider the system of three equations for $c_1$, $c_2$, $c_3$ consisting of the latter linear relation between $c_1$, $c_2$, and of eqs. (\ref{eq: Kirch simple syst}) for $i=0,1$. This system is easily solved symbolically (by MAPLE), its unique solution can be put as in the Proposition. Its components are manifestly even functions of $\epsilon$, thus conserved quantities.
\end{proof}

\begin{prop} \label{thm: kirchhoffullbasis}
${}$

{\rm (a)} The set $\Phi= (m_1^2+m_2^2,\, p_1m_1+p_2m_2,\, p_1^2+p_2^2,\, p_3^2,\, p_3,\, 1)$ is a HK basis for the map dK with $\dim K_{\Phi}(x) = 3$.
\smallskip

{\rm (b)} The set $\Phi_1=(1,\, p_3,\, p_3^2,\, m_1^2+m_2^2)$ is a HK basis for the map dK with a
one-dimensional null-space. At each point $x\in\bbR^6$ we have: $K_{\Phi_1}(x) = [c_0:c_1:c_2:-1]$.
The functions $c_0,c_1,c_2$ are conserved quantities of the map dK, given by
\begin{eqnarray*}
&&c_0 = \frac{a_1(m_1^2+m_2^2)-(b_1-b_3)p_3^2+\epsilon^2 c_0^{(4)}+\epsilon^4
c_0^{(6)}+\epsilon^6 c_0^{(8)}+\epsilon^8 c_0^{(10)}}{a_1\Delta_1\Delta_2},\\
&&c_1=
-\frac{2\epsilon^2 a_3(b_1-b_3) m_3 \big(C_2+\epsilon^2 c_1^{(4)}+\epsilon^4 c_1^{(6)} +\epsilon^6 c_1^{(8)} \big)}{\Delta_1\Delta_2},\\
&&c_2 = \frac{ (b_1-b_3) \big (1 + \epsilon^2 c_2^{(2)}+\epsilon^4 c_2^{(4)}+\epsilon^6 c_2^{(6)} +\epsilon^8 c_2^{(8)} \big)}{a_1\Delta_1\Delta_2}, \\
\end{eqnarray*}
where $\Delta_1$ is given in (\ref{eq: dKC Delta1}), and
$\Delta_2=1+\epsilon^2\Delta_2^{(2)}+\epsilon^4 \Delta_2^{(4)}+\epsilon^6\Delta_2^{(6)}$; coefficients
$c_k^{(q)}$ and $\Delta_2^{(q)}$ are homogeneous polynomials of degree $q$ in the phase variables. In particular:
\begin{eqnarray*}
&&c_2^{(2)}= -2 a_1^2(m_1^2+m_2^2)-(a_1^2-2a_1a_3+3a_3^2)m_3^2+a_1(b_1-b_3)(p_1^2+p_2^2)-a_1(b_1-b_3)p_3^2, \\
&&\Delta_2^{(2)}=a_1^2(m_1^2+m_2^2)+(a_1^2-3a_1a_3+3a_3^2)m_3^2-a_1(b_1-b_3)(p_1^2+p_2^2)+a_1(b_1-b_2)
p_3^2.
\end{eqnarray*}

\smallskip

{\rm (c)} The set $\Phi_2=(1,p_3,p_3^2,m_1p_1+m_2p_2)$ is a HK basis for the map dK with a
one-dimensional null-space. At each point $x\in\bbR^6$ we have: $K_{\Phi_2}(x) = [d_0:d_1:d_2:-1]$.
The functions $d_0,d_1,d_2$ are conserved quantities of the map dK, given by
\begin{eqnarray*}
&&d_0 =\frac{C_2+\epsilon^2 d_0^{(4)}+\epsilon^4
d_0^{(6)}+\epsilon^6 d_0^{(8)}+\epsilon^8 d_0^{(10)}}{\Delta_1\Delta_2},\\
&&d_1 = \frac{ m_3 \big( -1+\epsilon^2 d_1^{(2)}+\epsilon^4 d_1^{(4)}
+\epsilon^6 d_1^{(6)}+\epsilon^8 d_1^{(8)} \big)}{\Delta_1\Delta_2}, \\
&&d_2 =\frac{a_1(b_3-b_1)  \epsilon^2 \big( C_2  +\epsilon^2 c_1^{(4)}+\epsilon^4 c_1^{(6)} +\epsilon^6 c_1^{(8)} \big)}{\Delta_1\Delta_2},\\
\end{eqnarray*}
where $d_k^{(q)}$ are homogeneous polynomials of degree $q$ in the phase variables. In particular,
\[
d_1^{(2)}= -a_1a_3( m_1^2+m_2^2)-(a_1^2-3a_1a_3+3a_3^2)m_3^2+(a_1-a_3)(b_1-b_3)(p_1^2+ p_2^2)-3a_1(b_1-b_3) p_3^2.
\]
\smallskip

{\rm (d)} The set $\Phi_3=(1,\, p_3,\, p_3^2,\, p_1^2+p_2^2)$ is a HK basis for the map dK with a
one-dimensional null-space. At each point $x\in\bbR^6$ we have: $K_{\Phi_3}(x) = [e_0:e_1:e_2:-1]$.
The functions $e_0,e_1,e_2$ are conserved quantities of the map dK, given by
\begin{eqnarray*}
 &&e_0 =
\frac{C_1 +\epsilon^2 e_0^{(4)}+\epsilon^4
e_0^{(6)}+\epsilon^6 e_0^{(8)}+\epsilon^8 e_0^{(10)}}{\Delta_1\Delta_2},\\
&&e_1 = \frac{2\epsilon^2a_1(a_3-a_1)m_3\big(C_2+\epsilon^2 c_1^{(4)}+\epsilon^4 c_1^{(6)} +\epsilon^6 c_1^{(8)} \big)} {\Delta_1\Delta_2}, \\
&&e_2 = \frac{-1 +\epsilon^2 e_2^{(2)}+\epsilon^4
e_2^{(4)}+\epsilon^6 e_6^{(6)}+\epsilon^8 e_8^{(8)}}{\Delta_1\Delta_2}, \\
\end{eqnarray*}
where $e_k^{(q)}$ are polynomials of degree $q$ in the phase variables. In particular,
\[
e_2^{(2)}=-a_1^2(m_1^2+m_2^2)-(2a_1^2-4a_1a_3+3a_3^2)m_3^2+2a_1(b_1-b_3)(p_1^2+p_2^2)-a_1(b_1-b_3)p_3^2.
\]
\end{prop}
\begin{proof}
Statement (b) is proven using direct calculation. Statements (c) and (d) then follow analogously to Proposition \ref{thm: lagrangefullbasis} from the existence of linear relations between $c_1$ and $d_2$, as well as between $c_1$ and $e_1$.
\end{proof}
One can show that each of the sets $\{c_0,c_1,c_2\}$, $\{d_0,d_1,d_2\}$, and $\{e_0,e_1,e_2\}$ consists of three independent integrals of motion. Moreover, each of the sets $\{c_0,c_1,c_2,m_3\}$ and $\{e_0,e_1,e_2,m_3\}$ consists of four independent integrals. As further important results, me mention that Propositions \ref{thm: lagrange bilinear basis} (on the ``bilinear'' HK bases), \ref{th: dLT biquad} (on the invariant biquadratic curve for $(p_3,\wip_3)$), and \ref{th: dLT inv meas} (on the invariant measure) hold literally true for the map dK.

\section{Clebsch case of the rigid body motion in an ideal fluid}
\label{Sect: dC}

Another famous integrable case of the Kirchhoff equations was discovered by Clebsch \cite{C} and is characterized by the Hamilton function $H=H_1/2$, where
\begin{equation}\nn
H_1= \langle m, Am\rangle + \langle p,Bp\rangle=
\frac{1}{2}\sum_{k=1}^3(a_km_k^2+b_kp_k^2),
\end{equation}
where $A={\rm{diag}}(a_1,a_2,a_3)$ and $B={\rm{diag}}(b_1,b_2,b_3)$ satisfy the condition 
\begin{equation}\label{eq: Clebsch cond}
\frac{b_1-b_2}{a_3}+\frac{b_2-b_3}{a_1}+\frac{b_3-b_1}{a_2}=0.
\end{equation}
This condition is also equivalent to saying that the quantity 
\begin{equation}\label{eq: Clebsch theta}
\theta=\frac{b_j-b_k}{a_i(a_j-a_k)}
\end{equation}
takes one and the same value for all permutations $(i,j,k)$ of the indices (1,2,3).

For an embedding of this system into the modern theory of integrable systems see \cite{Per, RSTS}. Note that the Kirchhoff case ($a_1=a_2$ and $b_1=b_2$) can be considered as a particular case of the Clebsch case, but is special in many respects (the symmetry resulting in the existence of the Noether integral $m_3$, solvability in elliptic functions, in contrast to the general Clebsch system being solvable in terms of theta-functions of genus 2, etc.). Equations of motion of the Clebsch case are:
\begin{equation}\label{gcl}
\left\{ \begin{array}{l}
\dot{m} =  m\times Am+p\times Bp\,, \vspace{.1truecm}\\
\dot{p} =  p\times Am,
\end{array} \right.
\end{equation}
or in components
\beq
\label{eq: gClebsch}
\left\{ \begin{array}{l}
\dot{m}_1 = (a_3-a_2)m_2m_3+(b_3-b_2)p_2p_3,  \vspace{.1truecm}\\
\dot{m}_2 = (a_1-a_3)m_3m_1+(b_1-b_3)p_3p_1,  \vspace{.1truecm} \\
\dot{m}_3 =(a_2-a_1)m_1m_2+(b_2-b_1)p_1p_2, \vspace{.1truecm} \\
\dot{p}_1 = a_3m_3p_2-a_2m_2p_3, \vspace{.1truecm}  \\
\dot{p}_2 = a_1m_1p_3-a_3m_3p_1, \vspace{.1truecm}\\
\dot{p}_3 =a_2m_2p_1-a_1m_1p_2.
\end{array} \right.
\end{equation}
Condition (\ref{eq: Clebsch cond}) can be resolved for $a_i$ as 
\begin{equation}\nn
a_1=\frac{b_2-b_3}{\omega_2-\omega_3}\,,\quad
a_2=\frac{b_3-b_1}{\omega_3-\omega_1}\,,\quad
a_3=\frac{b_1-b_2}{\omega_1-\omega_2}\,.
\end{equation}
For fixed values of $\omega_i$ and varying values of $b_i$, equations of motion of the Clebsch case share the integrals of motion: the Casimirs $C_1$, $C_2$, cf. eq. (\ref{eq: e3 Cas}), and the Hamiltonians
\begin{equation}\nn
I_i=p_i^2+\frac{m_j^2}{\omega_i-\omega_k}+\frac{m_k^2}{\omega_i-\omega_j}.
\end{equation}
There are four independent functions among $C_i$, $I_i$, because of $C_1=I_1+I_2+I_3$. Note that  $H_1=b_1I_1+b_2I_2+b_3I_3$. One can denote all models with the same $\omega_i$ as a hierarchy, single flows of which are characterized by the parameters $b_i$. Usually, one denotes as ``the first flow'' of this hierarchy the one corresponding to the choice $b_i=\omega_i$, so that $a_i=1$. Thus, the first flow is characterized by the value $\theta=\infty$ of the constant (\ref{eq: Clebsch theta}).

\subsection{First flow of the Clebsch system}
The first flow of the Clebsch hierarchy is generated by the Hamilton function $H=H_1/2$, where
\begin{equation}\nn
    H_1=m_1^2+m_2^2 +m_3^2+\omega_1p_1^2+\omega_2p_2^2+\omega_3p_3^2.
\end{equation}
The corresponding equations of motion read:
\begin{equation}\nonumber
\left\{\begin{array}{l}
\dot{m}=p\times\Omega p,\vspace{.1truecm}\\
\dot{p}=p\times m,
\end{array}\right.
\end{equation}
where $\Omega={\rm diag}(\omega_1,\omega_2,\omega_3)$ is the matrix of parameters, or in components:
\beq \nn
\left\{\begin{array}{l}
\dot{m}_1= (\omega_3-\omega_2)p_2p_3, \vspace{.1truecm}  \\
\dot{m}_2 = (\omega_1-\omega_3)p_3p_1, \vspace{.1truecm} \\
\dot{m}_3 =(\omega_2-\omega_1)p_1p_2, \vspace{.1truecm}\\
\dot{p}_1=m_3 p_2 - m_2 p_3,  \vspace{.1truecm}\\
\dot{p}_2 = m_1 p_3 - m_3 p_1,  \vspace{.1truecm} \\
\dot{p}_3 = m_2 p_1 - m_1 p_2.
\end{array}\right.
\end{equation}
The fourth independent quadratic integral can be chosen as
\begin{equation}\nn
H_2=\omega_1m_1^2+\omega_2m_2^2+\omega_3m_3^2 -
         \omega_2\omega_3p_1^2-\omega_3\omega_1p_2^2-\omega_1\omega_2p_3^2.
\end{equation}
Note that $H_1=\omega_1I_1+\omega_2I_2+\omega_3I_3$, $H_1=-\omega_2\omega_3I_1-\omega_3\omega_1I_2-\omega_1\omega_2I_3$.

We mention the following Wronskian relation:
\begin{equation}
\label{eq: Clebsch W}
(\dot{m}_1p_1-m_1\dot{p}_1)+(\dot{m}_2p_2-m_2\dot{p}_2)+(\dot{m}_3p_3-m_3\dot{p}_3)=0,
\end{equation}
which holds true for the first Clebsch flow.

The HK discretization of the first Clebsch flow (proposed in \cite{RC}) is:
\beq\nn
\left\{\begin{array}{l}
\widetilde{m}_1-m_1 =\epsilon(\omega_3-\omega_2)
(\widetilde{p}_2p_3+p_2\widetilde{p}_3),        \vspace{.1truecm}  \\
\widetilde{m}_2-m_2 = \epsilon(\omega_1-\omega_3)
(\widetilde{p}_3p_1+p_3\widetilde{p}_1),         \vspace{.1truecm}   \\
\widetilde{m}_3-m_3 = \epsilon(\omega_2-\omega_1)
(\widetilde{p}_1p_2+p_1\widetilde{p}_2),        \vspace{.1truecm}   \\
\widetilde{p}_1-p_1 =
\epsilon(\widetilde{m}_3p_2+m_3\widetilde{p}_2)-
\epsilon(\widetilde{m}_2p_3+m_2\widetilde{p}_3), \vspace{.1truecm}   \\
\widetilde{p}_2-p_2 =
\epsilon(\widetilde{m}_1p_3+m_1\widetilde{p}_3)-
\epsilon(\widetilde{m}_3p_1+m_3\widetilde{p}_1), \vspace{.1truecm}   \\
\widetilde{p}_3-p_3 =
\epsilon(\widetilde{m}_2p_1+m_2\widetilde{p}_1)-
\epsilon(\widetilde{m}_1p_2+m_1\widetilde{p}_2).
\end{array}\right.
\end{equation}
As usual, it leads to a reversible birational map $\wx=f(x,\epsilon)$, $x=(m,p)^{\rm T}$, given by $f(x,\epsilon)=A^{-1}(x,\epsilon)x$ with
\[
A(m,p,\epsilon) = \begin{pmatrix}
1 & 0 & 0 & 0 & \epsilon\omega_{23}p_3 & \epsilon\omega_{23}p_2 \\
0 & 1 & 0 & \epsilon\omega_{31}p_3 & 0 & \epsilon\omega_{31}p_1 \\
0 & 0 & 1 & \epsilon\omega_{12}p_2 & \epsilon\omega_{12}p_1 & 0 \\
0 & \epsilon p_3 & -\epsilon p_2 & 1 & -\epsilon m_3 & \epsilon m_2 \\
-\epsilon p_3 & 0 & \epsilon p_1 & \epsilon m_3 & 1 & -\epsilon m_1  \\
\epsilon p_2 & -\epsilon p_1 & 0 & -\epsilon m_2 & \epsilon m_1 & 1
\end{pmatrix},
\]
where the abbreviation $\omega_{ij}=\omega_i-\omega_j$ is used. This map will be referred to as dC.

A ``simple'' conserved quantity can be found from the following natural discretization of the Wronskian relation (\ref{eq: Clebsch W}).
\begin{prop} \label{thm: Clebsch simple basis}
The set $\Gamma=(\widetilde{m}_1p_1-m_1\widetilde{p}_1,\, \widetilde{m}_2p_2-m_2\widetilde{p}_2,\, \widetilde{m}_3p_3-m_3\widetilde{p}_3)$ is a HK basis for the map dC with $\dim K_{\Gamma}(x) = 1$.  At each point $x\in\bbR^6$ we have: $K_\Gamma(x)=[e_1:e_2:e_3]$, where
\begin{equation}\label{eq: Clebsch e}
e_i=1+\epsilon^2(\omega_i-\omega_j)p_j^2+\epsilon^2(\omega_i-\omega_k)p_k^2.
\end{equation}
\end{prop}
\noindent The conserved quantities $e_i/e_j$ can be put as
$
e_i/e_j=(1+\epsilon^2\omega_iJ)/(1+\epsilon^2\omega_jJ),
$
where $J$ is a nice and symmetric integral,
\begin{equation}\nn
J=\frac{p_1^2+p_2^2+p_3^2}
{1-\epsilon^2(\omega_1p_1^2+\omega_2p_2^2+\omega_3p_3^2)}.
\end{equation}
Remarkably, it can be obtained also from a different (monomial) HK basis, see part b) of the following statement.

\begin{prop}\label{Th: dClebsch max basis}{\rm \cite{PPS}}

${\!\!}$
{\rm (a)} The set of functions $\Phi=(p_1^2,\, p_2^2,\, p_3^2,\, m_1^2,\, m_2^2,\, m_3^2,\, m_1p_1,\, m_2p_2,\, m_3p_3,\, 1)$ is a HK basis for the map dC with $\dim K_\Phi(m,p)=4$. Thus, any orbit of the map dC lies on an intersection of four quadrics in $\bbR^6$.
\smallskip

{\rm (b)} The set of functions $\Phi_0=(p_1^2,\, p_2^2,\, p_3^2,\, 1)$ is a HK basis for the map dC\, with $\dim K_{\Phi_0}(m,p)=1$. At each point $(m,p)\in\bbR^6$ there holds:
\begin{eqnarray}\nn
K_{\Phi_0}(m,p) & = & [e_1:e_2:e_3:-(p_1^2+p_2^2+p_3^2)] \nonumber\\
        & = &
\left[\,\frac{1}{J}+\epsilon^2\omega_1:\frac{1}{J}+\epsilon^2\omega_2:\frac{1}{J}+\epsilon^2\omega_3:-1
\right],\nn
\end{eqnarray}
with the quantities $e_i$ given in (\ref{eq: Clebsch e}).
\smallskip

{\rm (c)}
The sets of functions
\begin{eqnarray}
&&\Phi_1 = (p_1^2,\, p_2^2,\, p_3^2,\, m_1^2,\, m_2^2,\, m_3^2,\, m_1p_1), 
\label{eq: dC basis Phi1}\\
&&\Phi_2 = (p_1^2,\, p_2^2,\, p_3^2,\, m_1^2,\, m_2^2,\, m_3^2,\, m_2p_2), 
\label{eq: dC basis Phi2}\\
&&\Phi_3 = (p_1^2,\, p_2^2,\, p_3^2,\, m_1^2,\, m_2^2,\, m_3^2,\, m_3p_3),
\label{eq: dC basis Phi3}
\end{eqnarray}
are HK bases for the map dC with $\dim K_{\Phi_1}(m,p)=\dim K_{\Phi_2}(m,p)=\dim K_{\Phi_3}(m,p)=1$. At
each point $(m,p)\in\bbR^6$ there holds:
\begin{eqnarray}
&&K_{\Phi_1}(m,p) =
[\alpha_1:\alpha_2:\alpha_3:\alpha_4:\alpha_5:\alpha_6:-1], \nonumber \\
&&K_{\Phi_2}(m,p) =
[\beta_1:\beta_2:\beta_3:\beta_4:\beta_5:\beta_6:-1],\nonumber\\
&&K_{\Phi_3}(m,p) =
[\gamma_1:\gamma_2:\gamma_3:\gamma_4:\gamma_5:\gamma_6:-1],\nonumber
\end{eqnarray}
where $\alpha_j$,$\beta_j$, and $\gamma_j$ are rational functions of $(m,p)$, even with respect to $\epsilon$. They are conserved quantities of the map dCS. For $j=1,2,3$, they are of the form
\begin{equation}\nn
h=\frac{h^{(2)}+\epsilon^2h^{(4)}+\epsilon^4h^{(6)}+
\epsilon^6h^{(8)}+ \epsilon^8h^{(10)}+\epsilon^{10}h^{(12)}}
{2\epsilon^2(p_1^2+p_2^2+p_3^2)\Delta}\,,
\end{equation}
where $h$ stands for any of the functions $\alpha_j,\beta_j,\gamma_j$, $j=1,2,3$,
\begin{equation}\nn
\Delta=m_1p_1+m_2p_2+m_3p_3+\epsilon^2\Delta^{(4)}+\epsilon^4\Delta^{(6)}+
\epsilon^6\Delta^{(8)},
\end{equation}
and the corresponding $h^{(2q)}$, $\Delta^{(2q)}$ are homogeneous polynomials in phase variables of degree $2q$. For instance,
\begin{equation}\nn
\renewcommand{\arraystretch}{1.2}
\begin{array}{lll}
\alpha_1^{(2)}=C_1-I_1, \quad & \alpha_2^{(2)}=-I_1, &
\alpha_3^{(2)}=-I_1,\\
\beta_1^{(2)}=-I_2, & \beta_2^{(2)}=C_1-I_2,  \quad  &
\beta_3^{(2)}=-I_2,\\
\gamma_1^{(2)}=-I_3, & \gamma_2^{(2)}=-I_3, &
\gamma_3^{(2)}=C_1-I_3.
\end{array}
\end{equation}
For $j=4,5,6$, the functions $\alpha_j,\beta_j,\gamma_j$ are given by
\begin{equation}\nn
\begin{pmatrix} \alpha_4 & \alpha_5 & \alpha_6 \\
                \beta_4  & \beta_5  & \beta_6  \\
                \gamma_4 & \gamma_5 & \gamma_6
\end{pmatrix}
=\begin{pmatrix} D & A_1/(\omega_1-\omega_3) & A_1/(\omega_1-\omega_2) \\
                A_2/(\omega_2-\omega_3)  & D  & A_2/(\omega_2-\omega_1)  \\
                A_3/(\omega_3-\omega_2) & A_3/(\omega_3-\omega_1) & D
\end{pmatrix},
\end{equation}
where
\begin{eqnarray}
&&A_k =
\frac{1+\epsilon^2A_k^{(2)}+\epsilon^4A_k^{(4)}+\epsilon^6A_k^{(6)}+\epsilon^8A_k^{(8)}}
{2\epsilon^2\Delta}\,,  \nn
\\
&&D =
\frac{p_1^2+p_2^2+p_3^2+\epsilon^2D^{(4)}+\epsilon^4D^{(6)}+\epsilon^6D^{(8)}}
{2\Delta}\,,  \nn
\end{eqnarray}
and $A_k^{(2q)}$, $D^{(2q)}$ are homogeneous polynomials of degree $2q$ in phase variables, for instance,
\[
A_k^{(2)}=m_1^2+m_2^2+m_3^2+(\omega_2+\omega_3-2\omega_1)p_1^2+
(\omega_3+\omega_1-2\omega_2)p_2^2+(\omega_1+\omega_2-2\omega_3)p_3^2.
\]
The four conserved quantities $J$, $\alpha_1$, $\beta_1$ and $\gamma_1$ are functionally independent.
\end{prop}
Our paper \cite{PPS} contains a much more detailed information about the HK bases of the map dC, for instance, the further basis with a one-dimensional null-space: $\Theta=(p_1^2,\, p_2^2,\, p_3^2,\, m_1p_1,\, m_2p_2,\, m_3p_3)$. However, the following finding about the ``bilinear'' versions of the above bases is new.
\begin{prop}
Each of the sets of functions
\begin{eqnarray}
&&\Psi_0 = (\wip_1p_1,\, \wip_2p_2,\, \wip_3p_3,\, 1),\\
&&\Psi_1 = (\wip_1p_1,\, \wip_2p_2,\, \wip_3p_3,\, \wm_1m_1,\, \wm_2m_2,\, \wm_3m_3,\, \wm_1p_1+m_1\wip_1),
\label{eq: dC basis Psi1}
\\
&&\Psi_2 =(\wip_1p_1,\, \wip_2p_2,\, \wip_3p_3,\, \wm_1m_1,\, \wm_2m_2,\, \wm_3m_3,\, \wm_2p_2+m_2\wip_2),
\label{eq: dC basis Psi2}
\\
&&\Psi_3 = (\wip_1p_1,\, \wip_2p_2,\, \wip_3p_3,\, \wm_1m_1,\, \wm_2m_2,\, \wm_3m_3,\, \wm_3p_3+m_3\wip_3),
\label{eq: dC basis Psi3}
\end{eqnarray}
is a HK basis for the map dC with a one-dimensional null-space.
\end{prop}

\subsection{General flow of the Clebsch system}
\label{sect: Clebsch gen}

The HK discretization of the flow (\ref{gcl}) reads
\begin{equation}\nonumber
 \left\{
\begin{array}{l}
\widetilde{m}-m= \ep(\widetilde{m} \times Am + m \times
A\widetilde{m} + \widetilde{p} \times Bp+p\times B\widetilde{p}\,),
\vspace{.1truecm}\\
\widetilde{p}-p = \ep\left(  \widetilde{p}\times Am+p\times A\widetilde{m} \right),
\end{array}
\right.
\end{equation}
in components:
\beq\label{eq:dC2}
 \left\{ \begin{array}{l}
\widetilde{m}_1-m_1 =
\epsilon(a_3-a_2)(\widetilde{m}_2m_3+m_2\widetilde{m}_3)+
\epsilon(b_3-b_2)(\widetilde{p}_2p_3+p_2\widetilde{p}_3),
\vspace{.1truecm} \\
\widetilde{m}_2-m_2 =
\epsilon(a_1-a_3)(\widetilde{m}_3m_1+m_3\widetilde{m}_1)+
\epsilon(b_1-b_3)(\widetilde{p}_3p_1+p_3\widetilde{p}_1),
\vspace{.1truecm}\\
\widetilde{m}_3-m_3 =
\epsilon(a_2-a_1)(\widetilde{m}_1m_2+m_1\widetilde{m}_2)+
\epsilon(b_2-b_1)(\widetilde{p}_1p_2+p_1\widetilde{p}_2),
\vspace{.1truecm} \\
\widetilde{p}_1-p_1 = \epsilon
a_3(\widetilde{m}_3p_2+m_3\widetilde{p}_2)-\epsilon
a_2(\widetilde{m}_2 p_3+m_2\widetilde{p}_3),
\vspace{.1truecm}\\
\widetilde{p}_2-p_2= \epsilon
a_1(\widetilde{m}_1p_3+m_1\widetilde{p}_3)-\epsilon
a_3(\widetilde{m}_3p_1+m_3\widetilde{p}_1),
\vspace{.1truecm} \\
\widetilde{p}_3-p_3 = \epsilon
a_2(\widetilde{m}_2p_1+m_2\widetilde{p}_1)-\epsilon
a_1(\widetilde{m}_1 p_2+m_1\widetilde{p}_2).
\end{array}\right.
\end{equation}
In what follows, we will use the abbreviations $b_{ij}=b_i-b_j$
and $a_{ij}=a_i-a_j$. The linear system (\ref{eq:dC2}) defines an
explicit, birational map $f:\bbR^6 \rightarrow \bbR^6$,
\begin{equation}\nn
\begin{pmatrix}  \widetilde{m} \\ \widetilde{p} \end{pmatrix} = f(m,p,\epsilon) =
M^{-1}(m,p,\epsilon)\begin{pmatrix} m \\ p \end{pmatrix},
\end{equation}
where
\begin{equation}\nonumber
M(m,p,\epsilon) = \begin{pmatrix}
1 & \epsilon a_{23}m_3 & \epsilon a_{23}m_2 & 0 & \epsilon b_{23}p_3 & \epsilon b_{23}p_2 \\
\epsilon a_{31}m_3 & 1 & \epsilon a_{31}m_1 & \epsilon b_{31}p_3 & 0 & \epsilon b_{31}p_1 \\
\epsilon a_{12}m_2 & \epsilon a_{12}m_1 & 1 & \epsilon b_{12}p_2 & \epsilon b_{12}p_1 & 0 \\
0 & \epsilon a_2p_3 & -\epsilon a_3p_2 & 1 & -\epsilon a_3m_3 & \epsilon a_2m_2 \\
-\epsilon a_1p_3 & 0 & \epsilon a_3p_1 & \epsilon a_3m_3 & 1 & -\epsilon a_1m_1 \\
\epsilon a_1p_2 & -\epsilon a_2p_1 & 0 & -\epsilon a_2m_2 & \epsilon a_1m_1 & 1 \\
\end{pmatrix}.
\end{equation}
This map will be denoted dGC in what follows.

A ``simple'' integral of the map dGC can be obtained by discretizing the following Wronskian relation with constant coefficients, which holds for the general flow of the Clebsch system (\ref{eq: gClebsch}):
\begin{equation}\nn
    A_1(\dot{m}_1p_1-m_1\dot{p}_1)+A_2(\dot{m}_2p_2-m_2\dot{p}_2)+A_3(\dot{m}_3p_3-m_3\dot{p}_3)=0,
\end{equation}
with
\begin{equation}\nn
A_i=a_ia_j+a_ia_k-a_ja_k.
\end{equation}
\begin{prop} \label{thm: gClebsch simple basis}
The set $\Gamma=(\widetilde{m}_1p_1-m_1\widetilde{p}_1,\, \widetilde{m}_2p_2-m_2\widetilde{p}_2,\, \widetilde{m}_3p_3-m_3\widetilde{p}_3)$ is a HK basis for the map dGC with $\dim K_{\Gamma}(x) = 1$.  At each point $x\in\bbR^6$ we have: $K_\Gamma(x)=[e_1:e_2:e_3]$, where, for $(i,j,k)={\rm c.p.}(1,2,3)$,
\begin{equation}\nn
e_i=A_i+\epsilon^2a_i(b_i-b_j)A_k\Theta_j+\epsilon^2a_i(b_i-b_k)A_j\Theta_k,
\end{equation}
with
\begin{equation}\nn
\Theta_i=p_i^2+\frac{a_i}{\theta a_ja_k}m_i^2 .
\end{equation}
(Recall that $\theta$ is defined by eq. (\ref{eq: Clebsch theta}); we assume here that $\theta\neq\infty$.)
\end{prop}

As in the case of the first flow, the integrals $e_i/e_j$ can be expressed through one symmetric integral:
$e_i/e_j=(A_i-\theta a_iL)/(A_j-\theta a_jL)$, where
\begin{equation} \nonumber
L=\frac{a_2a_3A_1\Theta_1+a_3a_1A_2\Theta_2+a_1a_2A_3\Theta_3}
{1+\epsilon^2\theta a_1a_2a_3(\Theta_1+\Theta_2+\Theta_3)}.
\end{equation}
The quantities $e_i$ and the integral $L$ can be also obtained from a different (monomial) HK basis, given in part b) of the following Proposition.
\begin{prop}
${}$

{\rm (a)} The set $\Phi=(p_1^2,\, p_2^2,\, p_3^2,\, m_1^2,\, m_2^2,\, m_3^2,\, m_1p_1,\, m_2p_2,\, m_3p_3,\, 1)$ is a HK basis for the map dGC with $\dim K_\Phi(m,p)=4$. Thus, any orbit of the map dGC lies on an intersection of four quadrics in $\bbR^6$.
\smallskip

{\rm (b)} The set of functions $\Phi_0=(p_1^2,\, p_2^2,\, p_3^2,\, m_1^2,\, m_2^2,\, m_3^2,\, 1)$ is a HK basis for the map dGC\, with $\dim K_{\Phi_0}(m,p)=1$. At each point $(m,p)\in\bbR^6$ there holds:
\[
K_{\Phi_0}(m,p)=[a_2a_3e_1:a_3a_1e_2:a_1a_2e_3:(a_1/\theta)e_1:(a_2/\theta)e_2:(a_3/\theta)e_3:-e_0],
\]
where 
\begin{equation}\nn
e_0=a_2a_3A_1\Theta_1+a_3a_1A_2\Theta_2+a_1a_2A_3\Theta_3
\end{equation}
is an integral of motion of the continuous time flow (\ref{eq: gClebsch}).
\smallskip

{\rm (c)}
The sets of functions (\ref{eq: dC basis Phi1})--(\ref{eq: dC basis Phi3}) are HK bases for the map dGC with one-dimensional null-spaces.
\smallskip

{\rm (d)}
Each of the sets of functions $\Psi_0=(\wip_1p_1,\, \wip_2p_2,\, \wip_3p_3,\, \wm_1m_1,\, \wm_2m_2,\, \wm_3m_3,\, 1)$ and (\ref{eq: dC basis Psi1})--(\ref{eq: dC basis Psi3}) is a HK basis for the map dGC with a one-dimensional null-space.
\end{prop}

\section{$\mathfrak{su}(2)$ rational Gaudin system with $N=2$ spins}
\label{Sect: dG}

The Gaudin system \cite{G1} describes an interaction of $N$ quantum spins ${{y}}_i$, $i=1,\ldots,N$, with a
homogeneous constant external field ${{p}}$. Its classical version  is given by the following quadratic system of differential equations \cite{PS2}:
\beq
\dot{y}_i=\left( \l_i\, {{p}}  +  \sum_{j=1}^N {{y}}_j \right) \times y_i , \qquad 1 \leq i \leq N, \label{1gintro}
\eeq
where $y_i\in\mathfrak{su}(2)\simeq\mathbb{R}^3$, $p\in\mathfrak{su}(2)\simeq\mathbb{R}^3$ is a constant vector, and pairwise distinct numbers $\l_i$ are parameters of the model. The flow (\ref{1gintro}) is Hamiltonian with respect to the Lie-Poisson bracket of the direct sum of $N$ copies of $\mathfrak{su}(2)$, admits $2N$ independent conserved quantities in involution: the $N$ Casimir functions
\beq \label{poi} \nonumber
C_{k} =
\langle  {{y}}_k,{{y}}_k \rangle,
\eeq
and the following $N$ Hamiltonians:
\beq \nonumber
H_{k} = \langle p,y_k \rangle +
\sum_{\stackrel{\scriptstyle{j=1}} {j \neq k}}^N \frac{\langle y_k,y_j \rangle }{\l_k - \l_j},
\eeq
where $\langle \cdot, \cdot \rangle$ denotes the scalar product in $\su2\simeq\mathbb{R}^3$.
Note that the Hamilton function of the flow (\ref{1gintro}) is
\beq
H=\sum_{k=1}^N \l_k \, H_k= \half \sum_{\stackrel{\scriptstyle{i,j=1}} {i \neq
j}}^N \langle  {{y}}_i , {{y}}_j \rangle +
\sum_{i=1}^N \l_i \langle {{p}} , {{y}}_i \rangle.\nonumber
\eeq

In \cite{MPRS,PS2} it has been proved that a contraction of $N$ simple poles to one pole of order $N$
provides the integrable flow of the so called {\em {one-body rational $\su2$
tower}}, whose simplest instance with $N=2$ describes the
dynamics of the three-dimensional Lagrange top in the rest frame.

We consider here the HK discretization of thes flow (\ref{1gintro}) with $N=2$. We set
$y_1 = (x_1,x_2,x_3)^{\rm T}$, $y_2 = (z_1,z_2,z_3)^{\rm T}$, and choose the constant gravity vector $p=(0,0,1)^{\rm T}$. We thus obtain the following system of differential equations:
\beq \label{gaudincon}
\left\{\begin{array}{l}
\dot x_1 = x_2 z_3 - x_3 z_2 + \l_1 x_2 , \vspace{.1truecm}\\
\dot  x_2 = x_3 z_1 - x_1 z_3 - \l_1 x_1 ,  \vspace{.1truecm}\\
\dot  x_3 = x_1 z_2 - x_2 z_1  ,  \vspace{.1truecm}\\
\dot z_1 = z_2 x_3 - z_3 x_2 +  \l_2 z_2  , \vspace{.1truecm}\\
 \dot  z_2=z_3 x_1 - z_1 x_3 -  \l_2 z_1  , \vspace{.1truecm} \\
 \dot  z_3 = z_1 x_2 - z_2 x_1 ,
 \end{array}\right.
\eeq
with $\lambda_1$, $\lambda_2$ being real parameters.
The system (\ref{gaudincon}) has the following four independent integrals of motion:
$$
C_1 = x_2^2+x_2^2+x_3^2,\qquad C_2=z_1^2+z_2^2+z_3^2,
$$
$$
H_1 = x_3 + \frac{x_1 z_1 + x_2 z_2 + x_3 z_3 }{\l_1 - \l_2}, \qquad H_2 = z_3 + \frac{x_1 z_1 + x_2 z_2 + x_3 z_3 }{\l_2 - \l_1}.
$$
Note that the quantity $H_1+H_2 = x_3 + z_3$ is a linear integral of motion. We mention also the following Wronskian relation with constant coefficients:
\begin{equation}
\label{eq: G Wronski}
(x_3+z_3)(\dot{x}_1z_1-x_1\dot{z}_1+\dot{x}_2z_2-x_2\dot{z}_2)+
(\lambda_1+\lambda_2+x_3+z_3)(\dot{x}_3z_3-x_3\dot{z}_3)=0.
\end{equation}

The HK discretization of (\ref{gaudincon}) reads:
\beq \label{gaudinD}
\left\{\begin{array}{l}
\wx_1 - x_1= \epsilon(\wx_2 z_3 + x_2 \wz_3 - \wx_3 z_2  - x_3 \wz_2) + \epsilon\l_1 (\wx_2+ x_2), \vspace{.1truecm}\\
\wx_2 -x_2 = \epsilon(\wx_3 z_1 + x_3 \wz_1 - \wx_1 z_3  - x_1 \wz_3) - \epsilon\l_1 (\wx_1+ x_1),  \vspace{.1truecm}\\
\wx_3 - x_3= \epsilon (\wx_1 z_2 + x_1\wz_2 - \wx_2 z_1 - x_2 \wz_1) ,  \vspace{.1truecm}\\
\wz_1 -z_1=  \epsilon (\wz_2 x_3 + x_2\wx_3 - \wz_3 x_2 - z_3 \wx_2) +  \epsilon\l_2 (\wz_2+ z_2), \vspace{.1truecm}\\
\wz_2 - z_2=\epsilon (\wz_3 x_1 + z_3 \wx_1 - \wz_1 x_3 - z_1 \wx_3) - \epsilon\l_2 (\wz_1+ z_1), \vspace{.1truecm} \\
\wz_3 - z_3 =  \epsilon(\wz_1 x_2 + z_1\wx_2 - \wz_2 x_1- z_2 \wx_1).
\end{array}\right.
\eeq
The map $f:x\mapsto\widetilde{x}\ $ obtained by solving
(\ref{gaudinD}) for $\widetilde{x}$ is given by:
\beq
\nonumber
\widetilde{x} =f(x,\epsilon)=A^{-1}(x,\epsilon)(\mathds{1}+\epsilon B)x,
\eeq
where $x=(x_1,x_2,x_3,z_1,z_2,z_3)^{\rm T}$, and
\[
A(x,\epsilon) =
\begin{pmatrix}
1&-\epsilon  z_3&\epsilon z_2&0&\epsilon x_3&-\epsilon x_2& \\
-\epsilon  z_3& 1 & -\epsilon z_1&-\epsilon x_3&0&\epsilon x_1& \\
-\epsilon z_2&\epsilon z_1&1&\epsilon x_2&-\epsilon x_1&0& \\
0&\epsilon z_3&-\epsilon z_2&1&-\epsilon x_3&\epsilon x_2& \\
-\epsilon z_3&0&\epsilon z_1& \epsilon x_3 &1&-\epsilon x_1& \\
\epsilon z_2&-\epsilon z_1&0&-\epsilon x_2&\epsilon x_1&1
\end{pmatrix}-\epsilon B,
\]
\[
B=\begin{pmatrix}
0 & \lambda_1 & 0 & 0 &0 & 0 \\
-\lambda_1  & 0 & 0 & 0 & 0 & 0 \\
0 & 0 & 0 & 0 & 0 & 0 \\
0 & 0 & 0 & 0 &  \lambda_2 & 0 \\
0 & 0 & 0 & - \lambda_2 & 0  & 0 \\
0 & 0 & 0 & 0 & 0 & 0
\end{pmatrix}.
\]
This map will be called dG in the sequel. The quantity $x_3 + z_3$ is obviously preserved by the map dG. Other conserved quantities may now be found using the HK bases approach. A ``simple'' integral follows, as in the previous sections, by discretizing the Wronskian relation (\ref{eq: G Wronski}).
\begin{prop}
The set of functions $\Gamma=(\wx_1z_1-x_1\wz_1,\, \wx_2z_2-x_2\wz_2,\, \wx_3z_3-x_3\wz_3)$ is a HK basis for the map dG with the one-dimensional null-space $K_\Gamma(x,z)=[x_3+z_3:x_3+z_3: I]$, with
\begin{equation}\nn
I(x,z)=\frac{\lambda_1+\lambda_2+x_3+z_3+\epsilon^2\lambda_1(x_1^2+x_2^2)+\epsilon^2\lambda_2(z_1^2+z_2^2)+
\epsilon^2(\lambda_1+\lambda_2)(x_1z_1+x_2z_2)}{1-\epsilon^2(\lambda_1x_3+\lambda_2z_3+\lambda_1\lambda_2)}.
\end{equation}
\end{prop}

A full set of integrals is found in the following Proposition. The roles of the variables $x_i$ and $z_i$ are not quite symmetric there, and interchanging them is of course admissible but does not lead to new integrals of motion.
\begin{prop} \label{thm: gaudinfullbasis}\quad

{\rm (a)} The set $\Phi= (x_1^2+x_2^2,\, z_1^2+z_2^2, \, z_3^2, \, x_1 z_1 + x_2 z_2, \, z_3, \, 1)$ is a HK basis for the map dG with $\dim K_{\Phi}(x) = 3$.
\smallskip

{\rm (b)} The set $\Phi_1=(1,\, z_3,\, z_3^2,\, x_1^2+x_2^2)$ is a HK basis for the map dG with a
one-dimensional null-space. At each point $(x,z)\in\bbR^6$ we have: $K_{\Phi_1}(x,z) = [c_0:c_1:c_2:-1]$.
The functions $c_0,c_1,c_2$ are conserved quantities of the map dG, given by
\begin{eqnarray*}
&&c_0 = \frac{x_1^2+x_2^2-2 x_3 z_3-z_3^2+\epsilon^2 c_0^{(4)}+\epsilon^4
c_0^{(6)}+\epsilon^6 c_0^{(8)}+\epsilon^8 c_0^{(10)}}{\Delta_1\Delta_2},\\
&&c_1 = \frac{2 (x_3+z_3)\left(1+\epsilon^2
c_1^{(3)}+\epsilon^4 c_1^{(5)}+\epsilon^6 c_1^{(7)} +\epsilon^8 c_1^{(9)} \right)}{\Delta_1\Delta_2},\\
&&c_2 = -\frac{\left(1+\epsilon^2 c_2^{(2)}+
\epsilon^4 c_2^{(4)}+\epsilon^6 c_2^{(6)}+\epsilon^8 c_2^{(10)}\right)}{\Delta_1\Delta_2}, \\
\end{eqnarray*}
where
\[
\Delta_1=1-\epsilon^2(\l_1 x_3+\l_2 z_3+\lambda_1\lambda_2),\quad \Delta_2=
1+\epsilon^2\Delta_2^{(2)}+\epsilon^4 \Delta_2^{(4)}+ \epsilon^6\Delta_2^{(6)}.
\]
Here $\Delta^{(q)}$ and $c_k^{(q)}$ are polynomials of degree $q$ in the phase variables. In particular:
\[
c_2^{(2)}=-(x_1^2+x_2^2+x_3^2) - (z_1^2+z_2^2+z_3^2) - 2 (x_1 z_1 +x_2z_2+x_3z_3)
-2(\l_2 x_3+ \l_1 z_3)-(\lambda_1^2+\lambda_2^2),
\]
and $\Delta_2^{(2)}=c_2^{(2)}+\l_1 x_3+\l_2 z_3+\lambda_1\lambda_2$.
\smallskip

{\rm (c)} The set $\Phi_2=(1,z_3,z_3^2,x_1z_1+x_2z_2)$ is a HK basis for the map dG with a
one-dimensional null-space. At each point $(x,z)\in\bbR^6$ we have: $K_{\Phi_2}(x,z) = [d_0:d_1:d_2:-1]$.
The functions $d_0,d_1,d_2$ are conserved quantities of the map dG, given by
\begin{eqnarray*}
&&d_0 = \frac{x_1z_1+x_2z_2+x_3z_3+(\lambda_2 - \lambda_1)z_3+\epsilon^2 d_0^{(4)}+\epsilon^4
d_0^{(6)}+\epsilon^6 d_0^{(8)}+\epsilon^8 d_0^{(10)}}{\Delta_1\Delta_2},\\
&&d_1 = \frac{\lambda_1-\lambda_2-x_3-z_3+\epsilon^2 d_1^{(3)}+\epsilon^4 d_1^{(5)}
+\epsilon^6 d_1^{(7)}+\epsilon^8 d_1^{(9)}}{\Delta_1\Delta_2}, \\
&&d_2 =
-\frac{1+\epsilon^2c_2^{(2)}+
\epsilon^4c_2^{(4)}+\epsilon^6c_2^{(6)}+\epsilon^8c_2^{(8)}}{\Delta_1\Delta_2},
\end{eqnarray*}
where $d_k^{(q)}$ are polynomials of degree $q$ in the phase variables.
\smallskip

{\rm (d)} The set $\Phi_3=(1,\, z_3,\, z_3^2,\, z_1^2+z_2^2)$ is a HK basis for the map dG with a
one-dimensional null-space. At each point $(x,z)\in\bbR^6$ we have: $K_{\Phi_3}(x,z) = [e_0:e_1:e_2:-1]$.
The functions $e_0,e_1,e_2$ are conserved quantities of the map dG, given by
\begin{eqnarray*}
&&e_0 =
\frac{z_1^2+z_2^2+z_3^2+\epsilon^2 e_0^{(4)}+\epsilon^4
e_0^{(6)}+\epsilon^6 e_0^{(8)}+\epsilon^8 e_0^{(10)}}{\Delta_1\Delta_2},\\
&&e_1 = \frac{2\epsilon^2 \left(e_1^{(2)}+\epsilon^2 e_1^{(4)}
+\epsilon^4 e_1^{(6)}+\epsilon^6 e_1^{(8)}\right)}{\Delta_1\Delta_2}, \\
&&e_2 =
-\frac{1+\epsilon^2 e_2^{(2)}+
\epsilon^4 e_2^{(4)}+\epsilon^6 e_2^{(6)} +\epsilon^8 e_2^{(8)} }{\Delta_1\Delta_2},
\end{eqnarray*}
where $e_k^{(q)}$ are polynomials of degree $q$ in the phase variables.
\end{prop}
It can be shown that each of the sets $\{c_0,c_1,c_2,x_3+z_3\}$, $\{d_0,d_1,d_2,x_3+z_3\}$ and
$\{e_0,e_1,e_2,x_3+z_3\}$ contains four independent integrals.

Analogously to the situation for the map dLT, it is possible to obtain a polynomial integral and an invariant volume form for the map dG.
\begin{proposition}
The function
$$
G=\frac12(x_1 + z_1)^2 +\frac12 (x_2 + z_2 )^2+\l_1 x_3 + \l_2z_3-\frac{\ep^2}{2}  \big( (\l_1 x_1+ \l_2z_1)^2+(\l_1x_2 + \l_2z_2)^2 \big),
$$
is a conserved quantity for the map dG.
\end{proposition}
\begin{proposition}\label{th: dG inv meas}
The map dG possesses an invariant volume form:
\[
\det\frac{\partial(\wx,\wz)}{\partial(x,z)}=\frac{\phi(\wx,\wz)}{\phi(x,z)}\quad\Leftrightarrow\quad
f^*\omega=\omega,\quad \omega=\frac{dx_1\wedge dx_2\wedge dx_3\wedge dz_1\wedge dz_2\wedge dz_3}{\phi(x,z)},
\]
where $\phi(x,z))=\Delta_2(x,z)$.
\end{proposition}
Explicit integration of the map dG could be based on the following claim.
\begin{proposition}\label{th: dG biquad}
The component $x_3$ of the solution of the dG map satisfies a relation of the type
\begin{equation} \nonumber
Q(x_3,\widetilde{x}_3)=q_0x_3^2 \widetilde{x}_3^2 +
q_1 x_3 \widetilde{x}_3(x_3+\widetilde{x}_3)
+q_2(x_3^2+\widetilde{x}_3^2) + q_3 x_3\widetilde{x}_3
+q_4(x_3+\widetilde{x}_3) + q_5=0,
\end{equation}
coefficients of the biquadratic polynomial  $Q$  being conserved quantities of dG. Thus, $x_3(t)$ is an elliptic function of degree 2. An analogous statement holds for the component $z_3$.
\end{proposition}

\section{Conclusions}

The initial motivation for this study was the hope that HK discretization would preserve the integrability for all algebraically  integrable systems. This was formulated as a conjecture in \cite{PPS}. The list of integrable discretizations given in the present overview contains more than a dozen issues and is rather impressive. It includes systems integrable in terms of elliptic functions as well as those integrable in terms of theta-functions of genus $g=2$ (Clebsch system). This list might look like a convincing argument in favor of the integrability conjecture. However, at present we have also found examples which indicate that this conjecture might be wrong (e.g., the Zhukovski-Volterra system with all $\beta_k\neq 0$, or integrable chains, Volterra and dressing ones, with a big number of particles, say $N\ge 5$). We do not have rigorous proofs of the non-integrability in these cases, but the numerical evidence is rather strong. Since HK discretizations are (probably) not always integrable, the big number of integrable cases becomes still more intriguing: it is hard to imagine that all their common features come as a pure coincidence. We are after a theory which would clarify the problem of integrability of HK discretizations, but at present such a theory seems to remain a rather remote goal. Still, even without a general framework, HK discretizations represent a new fascinating chapter in the theory of integrable systems: we are now in a possession of a big and potentially growing stock of birational maps, integrable in terms of Abelian functions, highly non-trivial from the point of view of algebraic geometry and very different in nature from anything known before. The immediate goal of this review will be achieved if HK discretizations will attract attention of experts in the theory of integrable systems and in algebraic geometry.


\end{document}